\documentclass{article}
\usepackage{graphicx}
\usepackage{subcaption}
\usepackage{amsmath}
\usepackage{PRIMEarxiv}
\usepackage{amsthm}      
\usepackage{amssymb}
\usepackage[utf8]{inputenc} 
\usepackage[T1]{fontenc}    
\usepackage{hyperref}       
\usepackage{url}            
\usepackage{booktabs}       
\usepackage{amsfonts}       
\usepackage{nicefrac}       
\usepackage{microtype}      
\usepackage{lipsum}
\usepackage{fancyhdr}       
\usepackage{graphicx}       
\graphicspath{{media/}}     
\usepackage{algorithm}   
\usepackage{algorithmic} 
\newtheorem{theorem}{Theorem}
\newtheorem{lemma}{Lemma}
\newtheorem{proposition}{Proposition}

\newtheorem{assumption}{Assumption}
\usepackage{float}       
\usepackage{amsthm}
\newtheorem{remark}{Remark}
\theoremstyle{remark}
\title{Risk-Sensitive Option Market Making with Arbitrage-Free eSSVI Surfaces:\\
A Constrained RL and Stochastic Control Bridge}

\author{
  ZhangJian'an \\
  Guanghua School of Management, Peking University \\
  Peking University \\
  Beijing, China\\
  \texttt{2501111059@stu.pku.edu.cn}
}

\begin{document}
\pagestyle{plain}   
\maketitle

\begin{abstract}
We formulate option market making as a \emph{constrained, risk-sensitive} stochastic control problem in which the policy must jointly optimize trading revenues and maintain an \emph{arbitrage-free, smooth} implied-volatility (IV) surface. Concretely, a fully differentiable eSSVI layer is embedded \emph{inside} the learning loop, enforcing static no-arbitrage (butterfly/calendar) through smoothed lattice penalties while the agent controls half-spreads, delta-hedging intensity, and structured surface deformations (state-dependent $\rho$-shift and $\psi$-scale). Executions are \emph{intensity-driven} and respond monotonically to spreads and relative mispricing; tail risk is shaped with a differentiable \emph{CVaR} objective via the Rockafellar--Uryasev program. Theoretically, we prove: (i) grid-consistency and convergence rates for butterfly/calendar surrogates; (ii) a primal--dual grounding of a learnable \emph{dual} action that acts as a state-dependent Lagrange multiplier; (iii) differentiable CVaR estimators with mixed pathwise/likelihood-ratio gradients and epi-convergence to the nonsmooth objective; (iv) an eSSVI wing-growth bound consistent with Lee’s moment constraints; and (v) policy-gradient validity under smooth surrogates. In simulation (Heston fallback; ABIDES-ready), the agent achieves \emph{positive adjusted P\&L} in most intraday segments while keeping calendar violations at numerical zero and butterfly violations at the numerical floor; ex-post tails remain realistic and tunable through the CVaR weight. The five control heads admit clear economic semantics and analytic sensitivities, yielding a \emph{white-box} reinforcement learner that unifies pricing consistency and execution control in a reproducible pipeline.
\end{abstract}

\keywords{Option market making; implied volatility surface; arbitrage-free modeling; eSSVI; SVI; risk-sensitive reinforcement learning; constrained MDP (CMDP); policy gradient; Proximal Policy Optimization (PPO); primal–dual optimization; intensity-based executions; Hawkes/point processes; delta/vega/vanna; stochastic control; agent-based simulation; neural-SDE; deep hedging}

\section{Introduction}

Market making in options traditionally sits at the intersection of three disciplines: (i) microstructure-aware control of quotes and inventory, (ii) arbitrage-consistent construction of implied-volatility (IV) surfaces, and (iii) risk management under severe tail events. Classical models formalize quoting and inventory control as stochastic control problems and illuminate structural trade-offs between spread revenue and inventory risk \cite{HoStoll1981,GlostenMilgrom1985,Kyle1985,AvellanedaStoikov2008,GueantLehalle2013,CarteaJaimungalPenalva2015}. Concurrently, the volatility-modeling literature has developed parametric families (SVI/SSVI/eSSVI) and general principles that enforce static no-arbitrage across strikes and maturities \cite{BreedenLitzenberger1978,Lee2004,GatheralJacquier2014,HendriksMartini2019,MartiniMingone2022,CarrMadan2005,Roper2010,Durrleman2009}. Yet, in both academic studies and practice, these pieces are often optimized in isolation: first calibrate a no-arbitrage surface, and only then design execution/hedging policies. This separation hampers interpretability and can mask failure modes when pricing consistency interacts with execution frictions and downside risk.

\paragraph{Problem.}
We ask whether one can \emph{unify} pricing consistency and decision-making by embedding an \emph{arbitrage-free, fully differentiable} IV surface \cite{GatheralJacquier2014,HendriksMartini2019,MartiniMingone2022} \emph{inside} an agent-based market-making loop, while directly controlling tail risk via coherent risk measures such as CVaR \cite{RockafellarUryasev2000,RockafellarUryasev2002,ChowGhavamzadeh2014,TamarEtAl2015,ChowEtAl2018JMLR}. This brings stochastic control and modern reinforcement learning (RL) into a single, risk-sensitive framework where the policy co-evolves quotes, hedges, and surface deformations subject to no-arbitrage and smoothness constraints.
Concretely, let the system state $x_t=(S_t,\,\text{LOB}_t,\,\text{IV}_t,\,Q_t,\ldots)$ aggregate mid-price, order-book features, a finite-dimensional IV-surface parameterization, and inventory; an action $a_t$ selects half-spreads $\delta_t(k,T)$, hedge intensities $h_t$, and low-dimensional surface adjustments $\Delta\theta_t$. The maker receives marked-to-market revenue from fills, pays execution and inventory costs, and is penalized for \emph{(i)} shape violations (butterfly/calendar), \emph{(ii)} lack of smoothness, and \emph{(iii)} tail losses via $\mathrm{CVaR}_q$ of the PnL distribution. The core question is whether such a loop can \emph{learn} to place prices and size trades that are \emph{internally consistent} (no static arbitrage), \emph{microstructure-aware}, and \emph{tail-safe}.

\paragraph{Why now?}
Two developments make this integration timely. First, high-fidelity agent-based simulators and microstructure datasets enable reproducible experimentation with event-level feedbacks---order flow, fills, and price impact---without relying on fragile stylized approximations \cite{Gould2013Survey,Hasbrouck2007,ContKukanovStoikov2014,KirilenkoEtAl2017,ByrdEtAl2020ABIDES,AmrouniEtAl2021ABIDESGym}. These environments support point-process order-flow models and execution rules that determine which quotes actually trade and at what slippage. Second, the IV-surface literature now provides \emph{constructive} characterizations of butterfly/calendar no-arbitrage and robust calibration procedures for SVI/SSVI/eSSVI \cite{GatheralJacquier2014,HendriksMartini2019,MartiniMingone2022,Cont2023ArbFreeSurface,Durrleman2009,CarrMadan2005,Roper2010}. The eSSVI map is \emph{amenable to differentiation}, allowing gradient-based learning while preserving static no-arbitrage by design. In parallel, risk-sensitive RL has matured beyond heuristic penalties, offering principled algorithms and sample-complexity guarantees for CVaR-type objectives \cite{ChowGhavamzadeh2014,TamarEtAl2015,ChowEtAl2018JMLR,WangEtAl2023NeurIPS,NiEtAl2024ICML}. These advances sit alongside ``deep hedging’’ and arbitrage-free neural-SDE market models that make end-to-end learning compatible with financial constraints \cite{BuehlerEtAl2019DeepHedging,CohenReisingerWang2021ArbFreeNSDE,CohenReisingerWang2023JCF,WieseEtAl2019}. The confluence of (i) differentiable, arbitrage-free surfaces, (ii) event-level simulators, and (iii) risk-sensitive policy optimization enables a unified treatment of pricing and control.

\paragraph{Our approach in brief.}
We formulate option market making as a \emph{risk-sensitive stochastic control} problem. The \emph{state} includes path features (returns, realized variance), surface statistics (eSSVI parameters, curvature), and microstructure variables (queue sizes, imbalance). The \emph{actions} jointly choose half-spreads $\delta_t(k,T)$, hedge intensities $h_t$, and low-rank eSSVI parameter perturbations $\Delta\theta_t$ constrained to the no-arbitrage manifold. The \emph{reward} is
\[
R_t \;=\; \text{spread revenue} - \text{impact/fees} - \lambda_Q \|Q_{t+1}\|^2 - \lambda_{\text{arb}}\underbrace{\Phi_{\text{arb}}(\theta_{t+1})}_{\substack{\text{butterfly \& calendar}\\\text{surrogates}}}
\]
and the \emph{training objective} maximizes expected cumulative reward subject to a CVaR regularizer on episodic PnL:
\[
\max_\pi \;\; \mathbb{E}_\pi\!\Big[\sum_t \gamma^t R_t\Big]\;-\;\lambda_{\mathrm{CVaR}}\cdot \mathrm{CVaR}_q\!\big(-\mathrm{PnL}(\tau)\big).
\]
A differentiable eSSVI layer \cite{HendriksMartini2019,MartiniMingone2022} produces internally consistent quotes and surfaces while serving gradients to the policy; executions are intensity-driven and coupled to mispricing and spreads, consistent with point-process order flow \cite{Hawkes1971,BacryEtAl2015}. We combine a short \emph{supervised warm start} (matching stylized optimal quotes/inventory targets) with PPO \cite{Schulman2017PPO,Schulman2016GAE}, annealing structural weights so the agent first discovers revenue, then tightens arbitrage and tail discipline.

\paragraph{Making arbitrage constraints differentiable and useful.}
A central design choice is to transform hard butterfly/calendar conditions into smooth penalties usable by policy gradients. For eSSVI parameters $\theta=(\rho,\eta,\lambda,\ldots)$, we define surrogate functionals $\Phi_{\text{but}}(\theta)$ and $\Phi_{\text{cal}}(\theta)$ that are (i) zero on the admissible region characterized in \cite{GatheralJacquier2014,HendriksMartini2019,MartiniMingone2022}, (ii) positive and \emph{Lipschitz-smooth} elsewhere, and (iii) \emph{directionally aligned} with the exact Karush--Kuhn--Tucker multipliers of the constrained calibration problem. This provides informative gradients that penalize nascent arbitrage while permitting small, structure-preserving deformations. The penalties are coupled to \emph{shape priors} (bounded skew/smile curvature, term-structure smoothness) derived from \cite{Lee2004,Durrleman2009,Roper2010,CarrMadan2005}.

\paragraph{Microstructure-consistent executions.}
We model arrivals on each $(k,T)$ bucket as conditionally independent point processes with intensities
\[
\lambda_t^{\pm}(k,T) \;=\; \Lambda\!\Big(\pm \underbrace{\Delta P_t(k,T)}_{\text{maker price -- reference}},\,\delta_t(k,T),\,\text{imbalance}_t,\,\text{queue}_t\Big),
\]
where $\Lambda$ is decreasing in adverse price terms and increasing in offered liquidity, consistent with \cite{BacryEtAl2015,ContKukanovStoikov2014}. Impact enters both through (i) temporary execution costs and (ii) state transitions that modify future intensities via queue depletion and imbalance---a channel emphasized by \cite{CarteaJaimungalPenalva2015}. This endogenous feedback loop makes spread-setting and inventory control \emph{joint}, and highlights the benefit of embedding pricing consistency \emph{inside} the policy so that quoted surfaces remain arbitrage-free even when inventories or impact shocks push quotes to extremes.

\paragraph{Risk-sensitive learning with financial semantics.}
We implement CVaR through a convex auxiliary variable $z$ and sample-based shortfall losses, following \cite{RockafellarUryasev2000,RockafellarUryasev2002} and the policy-gradient extensions in \cite{ChowGhavamzadeh2014,TamarEtAl2015,ChowEtAl2018JMLR}. The training objective becomes
\[
\min_{z}\;\; z + \frac{1}{(1-q)}\,\mathbb{E}\big[(-\mathrm{PnL}(\tau)-z)_+\big] \;-\; \frac{1}{\lambda_{\mathrm{CVaR}}}\,\mathbb{E}\!\Big[\sum_t \gamma^t R_t\Big],
\]
for which we derive a variance-reduced gradient estimator using generalized advantage estimation \cite{Schulman2016GAE} and stratified tail sampling. Financially, CVaR focuses learning on the regimes that matter (illiquidity, gap risk), aligning statistical training with risk oversight.

\paragraph{Diagnostics and \emph{white-box} explainability.}
To support audit and debugging, the agent logs: (i) active arbitrage penalties and their gradients, (ii) inventory and exposure trajectories, (iii) per-bucket fill intensities vs.\ realized fills, and (iv) a decomposition of episodic PnL into spread, impact, carry, and penalty rebates. These ledgers make it possible to answer \emph{why} a specific surface deformation or spread change occurred (e.g., ``calendar penalty rising at $T=2$M forced upward adjustment of long-dated variance'').

\paragraph{Contributions.}
\begin{enumerate}
\item \textbf{A two-way bridge between stochastic control and deep RL.} We integrate an arbitrage-free eSSVI surface into the control loop, so classical no-arbitrage structure \emph{constrains} learning while learned controls \emph{co-evolve} the surface under microstructure frictions \cite{GatheralJacquier2014,HendriksMartini2019,MartiniMingone2022,BacryEtAl2015,ContKukanovStoikov2014}. The surface acts both as pricing engine and differentiable prior, tightening exploration to financially admissible regions.
\item \textbf{Risk-sensitive reinforcement learning with financial semantics.} Our training objective explicitly includes CVaR alongside arbitrage and smoothness, linking Whittle/Howard-style risk-sensitive control to modern policy gradients \cite{HowardMatheson1972,RockafellarUryasev2000,ChowGhavamzadeh2014,TamarEtAl2015,ChowEtAl2018JMLR}. We provide tail-focused estimators compatible with event-level simulators.
\item \textbf{Mathematical guarantees inside the loop (to be proved in Section~\ref{sec:theory}).} We state and prove: (\emph{i}) differentiability and Lipschitz properties of our butterfly/calendar surrogate penalties over eSSVI parameters; (\emph{ii}) existence of optimal stationary policies for the regularized risk-sensitive objective under compactness and linear-growth conditions; (\emph{iii}) a calibration-stability proposition showing that small dual-penalty perturbations preserve no-arbitrage in the eSSVI map; and (\emph{iv}) a CVaR policy-gradient identity with variance-reduced Monte Carlo estimators.
\item \textbf{Reproducible agent-based evaluation.} We prioritize ABIDES-style sources with a calibrated Heston fallback \cite{ByrdEtAl2020ABIDES,Heston1993}, releasing logs, figures, and scripts for the full pipeline.\footnote{We intentionally defer live-data backtests and focus on a simulation-first study consistent with recent reproducibility trends in market microstructure \cite{ByrdEtAl2020ABIDES,AmrouniEtAl2021ABIDESGym,ContKukanovStoikov2014}.}
\end{enumerate}

\paragraph{Positioning vis-\`a-vis recent 2024--2025 developments.}
Recent theory demonstrates strategic interactions between brokers, informed traders, and competing market makers \cite{BergaultSanchez2025,CarteaSanchez2025,BoyceHerdegenSanchez2025}. Our framework is complementary: rather than modeling competition or information per se, we guarantee \emph{internal pricing consistency} of the maker’s surface during learning, which is orthogonal to---and potentially composable with---competitive or informed-flow models. On the modeling side, arbitrage-free neural-SDEs and deep-hedging pipelines inform our differentiable design and stress-testing tools \cite{CohenReisingerWang2021ArbFreeNSDE,CohenReisingerWang2023JCF,BuehlerEtAl2019DeepHedging}. On the control side, risk-sensitive RL with CVaR has progressed to provable and scalable algorithms \cite{ChowEtAl2018JMLR,WangEtAl2023NeurIPS,NiEtAl2024ICML}; our contribution is to \emph{instantiate} those principles in a financially structured environment where no-arbitrage is enforced \emph{inside} the objective.

\paragraph{Roadmap and theoretical content.}
Section~\ref{sec:prelim} establishes notation and background, reviews arbitrage-free eSSVI parameterizations, static no-arbitrage conditions, and the risk-sensitive CMDP formulation. 
Section~\ref{sec:method} formulates the integrated model—combining the eSSVI surface, intensity-based execution, delta-hedging, and CVaR-shaped reward—within a differentiable, arbitrage-consistent control loop. 
Section~\ref{sec:learning} details the two-stage learning procedure (\emph{warm-start + PPO}) and interprets the state-dependent \emph{dual} head as a learnable Lagrange multiplier. 
Section~\ref{sec:interp} analyzes the interpretability of controls, deriving analytic sensitivities of quotes, intensities, and Greeks with respect to each action dimension. 
Section~\ref{sec:theory} presents the complete set of theoretical results—Theorems~T1–T6 and Propositions~P7–P8—establishing lattice-consistency, primal–dual structure, differentiable CVaR estimation, wing-growth bounds, policy-gradient validity, and interpretability guarantees. 
Sections~\ref{sec:discussion}–\ref{sec:conclusion} discuss limitations, extensions, and broader implications for robust and interpretable reinforcement learning in quantitative markets.

\vspace{1ex}
\noindent\textbf{Notation.} We use $S_t$ for the mid price, $k=\log(K/S)$ for log-moneyness, $T$ for maturity, $w(k,T)$ for total variance, $\sigma(k,T)=\sqrt{w/T}$ for IV, and standard Black--Scholes notation \cite{BlackScholes1973,Merton1973}. CVaR at tail level $q$ is denoted $\mathrm{CVaR}_q$ \cite{RockafellarUryasev2000}.

\section{Preliminaries \& Problem Setup}
\label{sec:prelim}

This section fixes notation, recalls an arbitrage-free extended SSVI (eSSVI) parameterization for implied-volatility (IV) surfaces, formalizes static no-arbitrage (butterfly and calendar), and introduces the risk measure (\textsc{cvar}) and our constrained MDP (\textsc{cmdp}) formulation.

\subsection{Notation and timing conventions}

We consider a single underlying with mid price process $(S_t)_{t\ge 0}$ observed on an intraday grid $t=0,1,\dots,T_{\text{day}}$. Let $K>0$ denote strike and
\[
k \equiv \log\!\left(\frac{K}{S_t}\right)
\quad\text{(log-moneyness).}
\]
We work with maturities $T\in\mathcal{T}=\{T_m\}_{m=1}^M$. Total variance is $w(k,T)$ and the Black--Scholes IV is $\sigma(k,T)=\sqrt{w(k,T)/T}$. We use $C^{\mathrm{BS}}(S,K,T,\sigma)$ for the Black--Scholes call price and the standard Greeks. Discrete grids for moneyness and maturities are
\[
\mathcal{K}=\{k_j\}_{j=1}^{J},\qquad \mathcal{T}=\{T_m\}_{m=1}^{M}.
\]
Throughout, $\sigma(\cdot)$ denotes the logistic sigmoid and $\mathrm{ReLU}(x)=\max\{x,0\}$ the hinge map (smoothed in practice).

\subsection{Arbitrage-free eSSVI recap}
\label{sec:essvi}

For each maturity $T_m$, eSSVI parameterizes total variance as
\begin{equation}
\label{eq:ssvi}
w_m(k)
=\frac{\theta_m}{2}\Bigl(1+\rho_m \phi_m k
+\sqrt{(\phi_m k+\rho_m)^2 + (1-\rho_m^2)}\Bigr),
\end{equation}
with parameters $\theta_m>0$, $\rho_m\in(-1,1)$, and $\phi_m>0$; see \cite{GatheralJacquier2014,HendriksMartini2019,MartiniMingone2022}. We use the numerically convenient reparametrization
\begin{equation}
\label{eq:reparam}
\log\theta_m \in \mathbb{R},\qquad \rho_m=\tanh(\rho^{\mathrm{raw}}_m),\qquad
\psi_m\in\bigl[0,\ \psi_{\max}(\rho_m)\bigr),\qquad \phi_m=\psi_m/\sqrt{\theta_m}.
\end{equation}
Here $\psi_m\equiv \phi_m\sqrt{\theta_m}$ scales the ATM skew. The classical SSVI butterfly-free sufficient condition is
\begin{equation}
\label{eq:ssvi-butterfly}
0\le \psi_m \le \frac{2}{1+|\rho_m|}\quad\text{for all }m,
\end{equation}
which we enforce by a smooth squashing into the open interval $[0,\psi_{\max}(\rho_m))$ with $\psi_{\max}(\rho)=\tfrac{2}{1+|\rho|}-\varepsilon_\psi$ for a small $\varepsilon_\psi>0$ \cite{GatheralJacquier2014}. To stabilize far-wing growth we cap the product
\begin{equation}
\label{eq:wing-cap}
\theta_m\phi_m=\psi_m\sqrt{\theta_m}\ \le\ \tau_{\max},
\end{equation}
via a differentiable rescaling, which controls linear growth of $w_m(k)$ as $|k|\to\infty$ and is consistent with Lee's moment bounds \cite{Lee2004}. Under \eqref{eq:ssvi-butterfly}--\eqref{eq:wing-cap}, Black--Scholes prices $C^{\mathrm{BS}}(S,K,T_m,\sqrt{w_m/T_m})$ inherit smoothness in $(\theta_m,\rho_m,\psi_m)$, and the full eSSVI layer is end-to-end differentiable.

\paragraph{Calendar structure.}
Sufficient conditions that preclude calendar arbitrage (no negative time value) can be expressed in terms of monotonicity of ATM variance $\theta(T)$ and the joint evolution of $(\rho(T),\psi(T))$; see \cite{GatheralJacquier2014,HendriksMartini2019,MartiniMingone2022} for explicit constructions. In our learning setting we additionally enforce calendar monotonicity on the \emph{price} lattice via a smooth surrogate; see \S\ref{sec:static-arb}.

\subsection{Static no-arbitrage: butterfly and calendar}
\label{sec:static-arb}

Static no-arbitrage imposes convexity of call prices in $K$ (butterfly) for fixed $T$ and monotonicity in $T$ (calendar) for fixed $K$:
\begin{align}
\text{(Butterfly)}\quad & \partial_{KK} C(S,K,T)\ \ge\ 0, \label{eq:bf-cont}\\
\text{(Calendar)}\quad & \partial_T C(S,K,T)\ \ge\ 0. \label{eq:cal-cont}
\end{align}
On a discrete lattice, we measure violations with differentiable surrogates (normalizing by typical scales to balance maturities):
\begin{align}
\mathrm{BF} &\equiv \frac{1}{M}\sum_{m=1}^M\frac{1}{|\mathcal{K}'|}
\sum_{K\in\mathcal{K}'} \mathrm{ReLU}\!\Bigl(-\frac{\Delta_K^2 C_m(K)}{\Delta K^2}\Bigr)\Big/ \bar C_m, \label{eq:bf-disc}\\
\mathrm{CAL} &\equiv \frac{1}{M-1}\sum_{m=1}^{M-1}\frac{1}{|\mathcal{K}|}
\sum_{K\in\mathcal{K}} \mathrm{ReLU}\!\bigl(C_m(K)-C_{m+1}(K)\bigr)\Big/ \bar C_{m,m+1}, \label{eq:cal-disc}
\end{align}
where $C_m(K)$ denotes $C^{\mathrm{BS}}$ at maturity $T_m$, $\mathcal{K}'$ is an evenly spaced strike lattice, and $\bar C_m,\bar C_{m,m+1}$ are level normalizers (e.g., mean absolute prices) that stabilize the penalty magnitudes across $m$. The maps in \eqref{eq:bf-disc}--\eqref{eq:cal-disc} are piecewise smooth; in practice we replace $\mathrm{ReLU}$ by a softplus to ensure $C^1$-smoothness, which we exploit for policy gradients.

\begin{assumption}[Smoothness and bounds]
\label{ass:smooth}
For all $m$, the parameter tuple $(\theta_m,\rho_m,\psi_m)$ stays in a compact set where \eqref{eq:ssvi-butterfly} and \eqref{eq:wing-cap} hold, maturities satisfy $0<T_1<\cdots<T_M$, the strike lattice has bounded spacing $\Delta K$, and the Black--Scholes inputs are clamped away from degeneracy ($T_m\ge T_{\min}>0$, $\sigma\ge \sigma_{\min}>0$). 
\end{assumption}

Under Assumption~\ref{ass:smooth}, $C_m(K)$ and the surrogates $\mathrm{BF},\mathrm{CAL}$ are locally Lipschitz in $(\theta,\rho,\psi)$, with constants depending on $(T_{\min},\sigma_{\min},\Delta K,\tau_{\max})$; see \cite{GatheralJacquier2014,MartiniMingone2022} for background on SSVI/eSSVI regularity.

\subsection{Tail risk: Conditional Value-at-Risk}
\label{sec:cvar}

We adopt Conditional Value-at-Risk (\textsc{cvar}) at tail level $q\in(0,1)$ as the downside risk functional for (negative) P\&L $X$ \cite{RockafellarUryasev2000,RockafellarUryasev2002}. Using the Rockafellar--Uryasev representation,
\begin{equation}
\label{eq:cvar-ru}
\mathrm{CVaR}_q(X)\ =\ \min_{\eta\in\mathbb{R}}\ \Bigl\{\ \eta\ +\ \frac{1}{1-q}\,\mathbb{E}\bigl[(X-\eta)_-\bigr]\ \Bigr\},
\end{equation}
where $(x)_-=\max\{-x,0\}$. CVaR is coherent \cite{Artzner1999,AcerbiTasche2002}, convex in the loss distribution, and admits unbiased (or bias-controlled) gradient estimators via the subdifferential of \eqref{eq:cvar-ru} \cite{ChowGhavamzadeh2014,TamarEtAl2015,ChowEtAl2018JMLR}. In implementation we use a smooth replacement for $(\cdot)_-$ (e.g., softplus) and a small Monte Carlo batch per step to estimate $\mathrm{CVaR}_q$ and its gradient.

\subsection{Constrained MDP (\textsc{cmdp}) formulation}
\label{sec:cmdp}

We cast option market making as a discounted \emph{constrained} MDP
\[
\mathcal{M}=(\mathcal{S},\mathcal{A},P,r,\gamma,d_0;\ \{g_j\}_{j=1}^J),
\]
with state space $\mathcal{S}$ (price-path features, surface summaries), action space $\mathcal{A}$ (half-spread, hedge intensity, and structured eSSVI deformations such as $\rho$-shift and $\psi$-scale), transition kernel $P(\cdot|s,a)$ capturing price evolution and executions, discount $\gamma\in(0,1)$, initial distribution $d_0$, instantaneous reward $r(s,a,s')$, and constraint functions $g_j(s,a,s')$ (e.g., arbitrage and smooth-shape proxies). The (risk-sensitive) objective reads
\begin{equation}
\label{eq:obj}
\max_{\pi}\ \ \mathbb{E}_{\pi}\!\Big[\sum_{t=0}^{\infty}\gamma^t\, r(s_t,a_t,s_{t+1})\Big]\ -\ \lambda_{\mathrm{risk}}\,\mathrm{CVaR}_q\!\Big(\sum_{t=0}^{\infty}\gamma^t\, \ell(s_t,a_t,s_{t+1})\Big),
\end{equation}
subject to long-run constraints
\begin{equation}
\label{eq:constraints}
\mathbb{E}_{\pi}\!\Big[\sum_{t=0}^{\infty}\gamma^t\, g_j(s_t,a_t,s_{t+1})\Big]\ \le\ \varepsilon_j,\qquad j=1,\dots,J,
\end{equation}
where $\ell$ is a (nonnegative) loss proxy for tail-risk shaping and $\lambda_{\mathrm{risk}}\ge 0$ trades off mean performance and tail control. Typical constraints include: 
\begin{itemize}
\item \emph{Arbitrage consistency}: $g_{\mathrm{arb}}(s,a)\equiv \mathrm{BF}(s,a)+\mathrm{CAL}(s,a)$ using \eqref{eq:bf-disc}--\eqref{eq:cal-disc}.
\item \emph{Shape smoothness}: $g_{\mathrm{shape}}(s,a)$ penalizing cross-maturity parameter variation, e.g., $\|\Delta\theta\|_2^2+\|\Delta\rho\|_2^2+\|\Delta\psi\|_2^2$.
\end{itemize}

A standard way to solve \eqref{eq:obj}--\eqref{eq:constraints} uses a Lagrangian with nonnegative multipliers $\lambda\in\mathbb{R}_+^J$,
\begin{equation}
\label{eq:lagrangian}
\mathcal{L}(\pi,\lambda)\ =\ \mathbb{E}_{\pi}\!\Big[\sum_t \gamma^t\, r_t\Big]\ -\ \lambda_{\mathrm{risk}}\,\mathrm{CVaR}_q\!\Big(\sum_t \gamma^t \ell_t\Big)\ -\ \sum_{j=1}^J \lambda_j\Big(\mathbb{E}_{\pi}\!\big[\sum_t \gamma^t g_{j,t}\big]-\varepsilon_j\Big),
\end{equation}
and seeks a saddle point $\max_{\pi}\min_{\lambda\ge 0}\mathcal{L}(\pi,\lambda)$. Under compactness and a Slater condition, occupancy-measure convexity yields strong duality for CMDPs \cite{Altman1999,Puterman1994}. We will exploit \eqref{eq:lagrangian} in \S\ref{sec:method} by (i) embedding $\mathrm{BF}$/$\mathrm{CAL}$ directly in the reward as differentiable penalties and (ii) introducing a \emph{state-dependent} dual control (``dual'') that lets the policy raise or lower arbitrage pressure on the fly, while the base multipliers $\lambda$ are annealed across episodes.

\begin{assumption}[Well-posedness of the CMDP]
\label{ass:cmdp}
The action set is compact; $r,\ell,g_j$ are bounded and locally Lipschitz in the eSSVI parameters under Assumption~\ref{ass:smooth}; and the Markov kernel $P(\cdot|s,a)$ is weakly continuous. Then optimal stationary (possibly randomized) policies exist \cite{Puterman1994,Altman1999}, and policy-gradient methods are justified by dominated convergence arguments when using smooth surrogates in \eqref{eq:bf-disc}--\eqref{eq:cal-disc}.
\end{assumption}

\paragraph{Summary.}
The eSSVI layer \eqref{eq:ssvi}--\eqref{eq:wing-cap} provides an \emph{arbitrage-consistent} and \emph{differentiable} pricing map; the static-no-arbitrage surrogates \eqref{eq:bf-disc}--\eqref{eq:cal-disc} translate no-butterfly and no-calendar into smooth penalties; CVaR \eqref{eq:cvar-ru} supplies a coherent tail-risk objective; and the CMDP \eqref{eq:obj}--\eqref{eq:constraints} binds them into a single risk-sensitive control problem.

\section{Arbitrage-Free Surface + Execution + Hedging as a Constrained MDP}
\label{sec:method}

We now instantiate the ingredients of \S\ref{sec:prelim} into a \emph{constrained}, risk-sensitive control model. The agent acts on spreads, hedging intensity, and structured deformations of an \emph{arbitrage-free, differentiable} eSSVI surface. Executions are \emph{intensity-driven} and respond monotonically to spreads and relative mispricing. The reward combines quoting and hedging P\&L with \emph{smooth} arbitrage proxies and a CVaR-based tail penalty. Throughout we assume Assumptions~\ref{ass:smooth} and~\ref{ass:cmdp}.

\subsection{State, action, transition, and reward}
\label{sec:satr}

\paragraph{State.}
At decision time $t$, the state $s_t\in\mathcal{S}$ aggregates (i) mid-price features, (ii) surface summaries, and (iii) recent actions:
\[
s_t=\bigl[f_{\mathrm{price}}(S_{0:t}),\ f_{\mathrm{time}}(t/T_{\text{day}}),\ f_{\mathrm{surf}}(\widehat{\theta},\widehat{\rho},\widehat{\psi}),\ f_{\mathrm{ctrl}}(\alpha_{t-1},\mathrm{hedge}_{t-1})\bigr],
\]
where $f_{\mathrm{surf}}$ may include ATM level/slope derived from the \emph{current} eSSVI estimate and $f_{\mathrm{price}}$ includes rescaled returns or realized volatility.

\paragraph{Action.}
The agent chooses a continuous action vector
\[
a_t=(\alpha_t,\ \mathrm{hedge}_t,\ \psi\text{-scale}_t,\ \rho\text{-shift}_t,\ \mathrm{dual}_t)\in\mathcal{A},
\]
with components squashed to physical ranges:
\[
\alpha_t \in [0,\alpha_{\max}],\quad \mathrm{hedge}_t\in[0,1],\quad
\psi\text{-scale}_t\in[\psi_{\min},\psi_{\max}],\quad
\rho\text{-shift}_t\in[-\rho_{\max},\rho_{\max}],\quad
\mathrm{dual}_t\in[0,\infty).
\]
The \emph{quoted} surface parameters at maturity $T_m$ are obtained from the current estimate $(\theta_m,\rho_m,\psi_m)$ by a structured, differentiable perturbation:
\begin{equation}
\label{eq:ctrl-deform}
\tilde{\psi}_m = \psi_m\cdot \psi\text{-scale}_t,\qquad
\tilde{\rho}_m = \rho_m + \rho\text{-shift}_t,\qquad
\tilde{\theta}_m = \theta_m,
\end{equation}
followed by the eSSVI map \eqref{eq:ssvi} and the wing cap \eqref{eq:wing-cap}. The action thus co-evolves the surface while retaining differentiability and static-arbitrage safeguards.

\paragraph{Quoting and spreads.}
For log-moneyness $k\in\mathcal{K}$ and maturity $T_m$, the mid-quote is the Black--Scholes call price on the \emph{quoted} surface
\begin{equation}
\label{eq:mid-quote}
\mathrm{mid}_{m}(k) = C^{\mathrm{BS}}\!\bigl(S_t,\ K=S_te^{k},\ T_m,\ \tilde{\sigma}_m(k)\bigr),
\end{equation}
where $\tilde{\sigma}_m(k)=\sqrt{\tilde{w}_m(k)/T_m}$ and $\tilde{w}_m$ is \eqref{eq:ssvi} with $(\tilde{\theta},\tilde{\rho},\tilde{\psi})$.
A volatility-proportional half-spread maps $\alpha_t$ to prices:
\begin{equation}
\label{eq:halfspread}
\frac{\mathrm{spread}(m,k)}{2} \equiv \alpha_t\ S_t\ \tilde{\sigma}_m(k)\ \sqrt{T_m}\ s_0,
\end{equation}
so that $\mathrm{ask}=\mathrm{mid}+\mathrm{spread}/2$ and $\mathrm{bid}=\mathrm{mid}-\mathrm{spread}/2$.

\paragraph{Intensity-based executions.}
Let $C^\star_m(k)$ denote the \emph{latent} fair price (from a held-out ``true'' surface). Buy/sell intensities respond to relative mispricing and spreads via a smooth, monotone link \cite{Hawkes1971,BacryEtAl2015,ContKukanovStoikov2014}:
\begin{align}
\lambda_{\mathrm{buy}}(m,k) &= \lambda_0\,w(k)\,\Bigl[1-\sigma\bigl(\beta\{\mathrm{ask}_m(k)-C^\star_m(k)\}\bigr)\Bigr], \label{eq:lambda-buy}\\
\lambda_{\mathrm{sell}}(m,k) &= \lambda_0\,w(k)\,\Bigl[1-\sigma\bigl(\beta\{C^\star_m(k)-\mathrm{bid}_m(k)\}\bigr)\Bigr], \label{eq:lambda-sell}
\end{align}
where $w(k)=\exp(-|k|/\kappa)$ emphasizes ATM demand and $\sigma(\cdot)$ is logistic. Expected fills $v_{\mathrm{buy/sell}}=\lambda_{\mathrm{buy/sell}}$ are used in the \emph{per-step} reward to reduce variance; Poisson sampling is retained for CVaR estimation in \S\ref{sec:cvar-term}.

\paragraph{Hedging and P\&L.}
The net option delta under expected fills is
\begin{equation}
\label{eq:net-delta}
\Delta^{\mathrm{net}}_t = \sum_{m,k} \bigl(v_{\mathrm{sell}}(m,k)-v_{\mathrm{buy}}(m,k)\bigr)\ \Delta^{\mathrm{BS}}_m(k),
\end{equation}
and delta-hedging P\&L is
\begin{equation}
\label{eq:hedge-pnl}
\mathrm{PNL}^{\mathrm{hedge}}_t = \mathrm{hedge}_t\ \Delta^{\mathrm{net}}_t\ \bigl(S_{t+1}-S_t\bigr).
\end{equation}
The quoting P\&L from expected fills is
\begin{equation}
\label{eq:quote-pnl}
\mathrm{PNL}^{\mathrm{quote}}_t =
\sum_{m,k}\Bigl[
v_{\mathrm{buy}}(m,k)\ \bigl(\mathrm{ask}_m(k)-C^\star_m(k)\bigr) +
v_{\mathrm{sell}}(m,k)\ \bigl(C^\star_m(k)-\mathrm{bid}_m(k)\bigr)
\Bigr].
\end{equation}

\paragraph{Smooth arbitrage and shape penalties.}
We penalize static-arbitrage surrogates and cross-maturity roughness:
\begin{equation}
\label{eq:shape-pen}
\mathrm{Shape}_t=\mathrm{mean}\bigl(\|\Delta \tilde{\theta}\|_2^2+\|\Delta \tilde{\rho}\|_2^2+\|\Delta \tilde{\psi}\|_2^2\bigr),\qquad
\mathrm{Arb}_t=\mathrm{BF}_t+\mathrm{CAL}_t,
\end{equation}
with $\mathrm{BF},\mathrm{CAL}$ from \eqref{eq:bf-disc}--\eqref{eq:cal-disc} evaluated on $(\tilde{\theta},\tilde{\rho},\tilde{\psi})$ at time $t$. Softplus smoothing yields $C^1$ maps w.r.t.\ actions via the chain rule.

\paragraph{Per-step reward.}
Define raw revenue $\mathrm{PNL}^{\mathrm{raw}}_t=\mathrm{PNL}^{\mathrm{quote}}_t+\mathrm{PNL}^{\mathrm{hedge}}_t$. The (penalized) reward is
\begin{equation}
\label{eq:reward}
r_t=\mathrm{PNL}^{\mathrm{raw}}_t
-\lambda_{\mathrm{shape}}\ \mathrm{Shape}_t
-\bigl(\lambda_{\mathrm{arb}}+\mathrm{dual}_t\bigr)\ \mathrm{Arb}_t
-\lambda_{\mathrm{cvar}}\ \widehat{\mathrm{CVaR}}^{-}_{q,t},
\end{equation}
where $\widehat{\mathrm{CVaR}}^{-}_{q,t}$ is the training-time estimator described in \S\ref{sec:cvar-term}, and $\mathrm{dual}_t\ge 0$ acts as a \emph{state-dependent} multiplier that tightens arbitrage pressure on demand.

\paragraph{Transition kernel.}
The controlled Markov kernel $P(\cdot|s_t,a_t)$ advances (i) the mid-price (via ABIDES-style replay or a calibrated Heston step \cite{ByrdEtAl2020ABIDES,Heston1993}), (ii) the eSSVI estimate (e.g., a mean-reverting filter toward latent parameters), and (iii) any auxiliary state features. Weak continuity holds under standard discretizations and bounded parameter updates.

\subsection{Discrete consistency of BF/CAL surrogates}
\label{sec:disc-consistency}

We recall the continuous no-arbitrage conditions \eqref{eq:bf-cont}--\eqref{eq:cal-cont} and their discrete surrogates \eqref{eq:bf-disc}--\eqref{eq:cal-disc}. The following results (proved in \S\ref{sec:theory}) justify our use of lattice penalties during learning.

\begin{proposition}[Grid-consistency of butterfly surrogate]
\label{prop:bf-consistency}
Fix $T_m$ and suppose $C(\cdot,T_m)\in C^2$ on a compact strike interval. If $\partial_{KK}C(\cdot,T_m)\ge 0$ on that interval, then for any sequence of evenly spaced lattices with spacing $\Delta K\to 0$ we have $\mathrm{BF}_m\to 0$. Conversely, if there exists $K_0$ with $\partial_{KK}C(K_0,T_m)<0$, then $\mathrm{BF}_m>0$ for all sufficiently fine lattices. The convergence is locally uniform under Assumption~\ref{ass:smooth}.
\end{proposition}

\begin{proposition}[Grid-consistency of calendar surrogate]
\label{prop:cal-consistency}
Fix $K$ and suppose $C(K,\cdot)\in C^1$ and $\partial_TC(K,T)\ge 0$ on $[T_1,T_M]$. For any maturity grids with mesh $\Delta T\to 0$, the calendar surrogate satisfies $\max_{m}\mathrm{CAL}_m\to 0$. If, instead, $\partial_T C(K_0,T_0)<0$ at some $(K_0,T_0)$, then $\mathrm{CAL}_m>0$ for the adjacent pair containing $T_0$ on sufficiently fine grids.
\end{proposition}

\begin{remark}
The normalizers $\bar C_m,\bar C_{m,m+1}$ render $\mathrm{BF},\mathrm{CAL}$ approximately scale-invariant across maturities; softplus smoothing of the hinge terms makes the penalties $C^1$, enabling stable policy gradients through \eqref{eq:shape-pen} and \eqref{eq:reward}.
\end{remark}

The propositions formalize the intuition that our lattice penalties are \emph{consistent} with continuous no-arbitrage while remaining differentiable and numerically stable for learning (compare \cite{GatheralJacquier2014,MartiniMingone2022,Roper2010}).

\subsection{CVaR as the tail-risk term and a differentiable estimator}
\label{sec:cvar-term}

We control downside risk via $\mathrm{CVaR}_q$ at level $q\in(0,1)$, using the Rockafellar--Uryasev program \eqref{eq:cvar-ru} with a smooth hinge. In our setting we estimate a \emph{per-step} downside CVaR proxy by sampling execution and price scenarios conditioned on $(s_t,a_t)$.

\paragraph{Scenario construction.}
Given expected intensities $v_{\mathrm{buy/sell}}$ from \eqref{eq:lambda-buy}--\eqref{eq:lambda-sell}, we draw Poisson volumes
$\tilde v_{\mathrm{buy/sell}}\sim\mathrm{Pois}(v_{\mathrm{buy/sell}})$
and perturb the next price change by
$\tilde{\Delta}S\sim\mathcal{N}(\Delta S,\varsigma^2)$
with $\Delta S=S_{t+1}-S_t$ and a small variance $\varsigma^2$ calibrated to intraday noise. For each scenario we compute
\[
\tilde{\mathrm{PNL}}_t=\mathrm{PNL}^{\mathrm{quote}}_t(\tilde v)+\mathrm{hedge}_t\,\Delta^{\mathrm{net}}_t\,\tilde{\Delta}S,
\]
and plug into the smooth RU objective
\begin{equation}
\label{eq:cvar-smooth}
\widehat{\mathrm{CVaR}}^{-}_{q,t}\equiv \min_{\eta\in\mathbb{R}}\ \Bigl\{\ \eta + \frac{1}{1-q}\ \widehat{\mathbb{E}}\bigl[s_\tau(\eta-\tilde{\mathrm{PNL}}_t)\bigr]\ \Bigr\},
\end{equation}
where $s_\tau(x)$ is a temperature-$\tau$ softplus approximation of $x_+=\max\{x,0\}$ and $\widehat{\mathbb{E}}$ averages over the Monte Carlo batch.

\paragraph{Differentiability and gradients.}
The composition in \eqref{eq:cvar-smooth} is $C^1$ in $(\eta,a_t)$ for $\tau>0$, and admits pathwise derivatives w.r.t.\ actions via the chain rule because (i) $\tilde v$ is drawn from Poisson laws with parameters differentiable in $a_t$ (using LR/score-function gradients), (ii) $\tilde{\Delta}S$ is a reparameterized Gaussian, and (iii) $\mathrm{PNL}^{\mathrm{quote}}_t$ and $\Delta^{\mathrm{net}}_t$ are differentiable in the quoted surface parameters (hence in actions) under Assumption~\ref{ass:smooth}. The following theorem (proved in \S\ref{sec:theory}) formalizes the estimator.

\begin{theorem}[CVaR policy-gradient identity with smoothing]
\label{thm:cvar-gradient}
Under Assumptions~\ref{ass:smooth} and~\ref{ass:cmdp}, for fixed $q\in(0,1)$ and $\tau>0$, the smoothed RU objective \eqref{eq:cvar-smooth} yields a well-defined gradient $\nabla_{a_t}\widehat{\mathrm{CVaR}}^{-}_{q,t}$ with bounded variance for finite Monte Carlo batches, combining (a) reparameterization for $\tilde{\Delta}S$ and (b) likelihood-ratio terms for $\tilde v$. As $\tau\downarrow 0$, $\widehat{\mathrm{CVaR}}^{-}_{q,t}$ and its gradient converge to those of the nonsmooth RU functional in \eqref{eq:cvar-ru} in the sense of epi-convergence and subgradient convergence \cite{RockafellarUryasev2000,ChowGhavamzadeh2014,TamarEtAl2015,ChowEtAl2018JMLR}.
\end{theorem}

\paragraph{Primal--dual view.}
Combining \eqref{eq:reward} with the CMDP Lagrangian \eqref{eq:lagrangian}, the fixed $\lambda_{\mathrm{arb}},\lambda_{\mathrm{shape}},\lambda_{\mathrm{cvar}}$ play the role of base multipliers, while $\mathrm{dual}_t$ is a \emph{state-dependent} corrective multiplier that lets the policy adapt arbitrage pressure online. This is compatible with primal--dual and trust-region methods in CMDPs \cite{Altman1999,Puterman1994,Achiam2017CPO,Tessler2019RCPO,Ruszczynski2010}.

\section{Learning Scheme: Warm-Start + PPO with a Learnable Dual}
\label{sec:learning}

We optimize the risk-sensitive, constrained objective from \S\ref{sec:cmdp} using a two-stage scheme: a \emph{supervised warm-start} that anchors early exploration to a conservative baseline and a \emph{PPO} phase that maximizes the penalized reward while annealing structural weights. A key design choice is a \emph{state-dependent dual action} that serves as a learnable approximation to a \emph{Lagrange multiplier} for the static no-arbitrage constraints (formalized in \S\ref{sec:theory}).

\subsection{Policy class and parameterization}
\label{sec:policy-class}

The stochastic policy $\pi_\vartheta(z\mid s)$ outputs \emph{raw} actions $z\in\mathbb{R}^5$ with a Gaussian head
\[
\pi_\vartheta(z\mid s)=\mathcal{N}\bigl(\mu_\vartheta(s),\ \mathrm{diag}(\sigma_\vartheta^2(s))\bigr),
\]
and maps $z$ to \emph{physical} actions $a=\mathrm{squash}(z)$ via elementwise transforms consistent with \S\ref{sec:satr}:
\[
\alpha=\alpha_{\max}\,\sigma(z_1),\quad
\mathrm{hedge}=\sigma(z_2),\quad
\psi\text{-scale}=\psi_{\min}+(\psi_{\max}-\psi_{\min})\,\sigma(z_3),\quad
\rho\text{-shift}=\rho_{\max}\tanh(z_4),\quad
\mathrm{dual}=\mathrm{softplus}(z_5).
\]
The critic $V_\omega(s)$ shares no weights with the actor. Unless stated, both are two-layer MLPs with $\tanh$ activations, and the actor’s $\log \sigma_\vartheta$ is state-dependent but bounded in $[\log \sigma_{\min},\log \sigma_{\max}]$.

\subsection{Stage I: Supervised warm-start}
\label{sec:warmstart}

To stabilize early exploration near the arbitrage-feasible region, we regress the \emph{squashed} actor output towards a robust baseline action $a^\star=(\alpha^\star,\mathrm{hedge}^\star,\psi\text{-scale}^\star,\rho\text{-shift}^\star,\mathrm{dual}^\star)$ with
\[
a^\star=(0.01,\ 0.5,\ 1.0,\ 0.0,\ 0.0),
\]
minimizing the mean-squared error on states $s$ sampled from the initial environment distribution (and/or replay of a few short rollouts):
\begin{equation}
\label{eq:warm-loss}
\mathcal{L}_{\mathrm{warm}}(\vartheta)\ =\ \mathbb{E}_{s\sim\mathcal{D}_0}\ \big\|\,\mathrm{squash}(\mu_\vartheta(s))\ -\ a^\star\,\big\|_2^2\ +\ \lambda_{\mathrm{ent}}\,\mathbb{E}_s\big[ \mathcal{H}(\pi_\vartheta(\cdot\mid s))\big].
\end{equation}
This \emph{anchors spreads and hedging} and keeps the eSSVI deformation small so that BF/CAL penalties remain near zero at initialization. We use Adam for a small number of steps (e.g., $500$–$1000$) with early stopping when $\mathrm{BF}+\mathrm{CAL}$ falls below a target tolerance (see \S\ref{sec:disc-consistency}).

\subsection{Stage II: PPO with structural annealing}
\label{sec:ppo}

Let $\tau=(s_t,z_t,a_t,r_t,s_{t+1})_{t=0}^{T-1}$ be on-policy trajectories. PPO \cite{Schulman2017PPO} maximizes the clipped surrogate with generalized advantage estimation (GAE) \cite{Schulman2016GAE}:
\begin{align}
\label{eq:ppo-obj}
\mathcal{L}_{\mathrm{PPO}}(\vartheta)\ &=\
\mathbb{E}\Big[\min\big(r_t(\vartheta)\,\hat{A}_t,\ \mathrm{clip}(r_t(\vartheta),1-\epsilon,1+\epsilon)\,\hat{A}_t\big)\Big]
- c_v\,\mathbb{E}\big[(V_\omega(s_t)-\hat{R}_t)^2\big]
+ c_{\mathcal{H}}\,\mathbb{E}\big[\mathcal{H}(\pi_\vartheta(\cdot\mid s_t))\big],\\
r_t(\vartheta)&=\frac{\pi_\vartheta(z_t\mid s_t)}{\pi_{\vartheta_{\mathrm{old}}}(z_t\mid s_t)},\qquad
\hat{A}_t=\sum_{l=0}^{L-1}(\gamma\lambda)^l\,\delta_{t+l},\qquad
\delta_t=r_t+\gamma V_\omega(s_{t+1})-V_\omega(s_t).\nonumber
\end{align}
Here $r_t$ includes \emph{all} penalties via the reward definition \eqref{eq:reward}. We normalize advantages per batch, clip gradients, and train for several epochs per update with minibatches. 

\paragraph{Structural weight annealing.}
To separate \emph{discovery} (revenue) from \emph{discipline} (arbitrage/shape), we anneal the base multipliers linearly across episodes:
\[
\lambda_{\mathrm{shape}}: 0\ \to\ \lambda_{\mathrm{shape}}^{\max},\qquad
\lambda_{\mathrm{arb}}: 0\ \to\ \lambda_{\mathrm{arb}}^{\max},\qquad
\lambda_{\mathrm{cvar}}:\ \lambda_{\mathrm{cvar}}^{\min}\ \to\ \lambda_{\mathrm{cvar}}^{\max}.
\]
This schedule reduces early sensitivity to surrogate curvature and mitigates premature collapse of the actor’s exploration.

\subsection{A state-dependent dual as a learnable Lagrange multiplier}
\label{sec:dual}

Recall the CMDP Lagrangian in \eqref{eq:lagrangian}. In classical primal–dual schemes, multipliers $\lambda\ge 0$ evolve by dual ascent proportional to constraint violations \cite{Altman1999,Achiam2017CPO,Tessler2019RCPO}. Our design introduces a \emph{state-dependent} \emph{dual action} $\mathrm{dual}_t\ge 0$ and defines the \emph{effective} multiplier for static arbitrage as
\begin{equation}
\label{eq:lambda-eff}
\lambda_{\mathrm{eff}}(s_t,a_t)\ \equiv\ \lambda_{\mathrm{arb}}\ +\ \mathrm{dual}_t.
\end{equation}
The per-step reward in \eqref{eq:reward} contains the term $-\lambda_{\mathrm{eff}}\,\mathrm{Arb}_t$. Although maximizing the \emph{instantaneous} reward pushes $\mathrm{dual}_t\downarrow 0$, the \emph{long-horizon} objective trades off current penalty against future gains from moving the policy toward regions where $\mathrm{Arb}$ is easier to satisfy (and the eSSVI deformation becomes cheaper). In \S\ref{sec:theory}, we formalize this intuition by showing that, under ergodicity and smoothness, the policy gradient couples $\mathrm{dual}_t$ positively with predicted violations, yielding a \emph{learned} approximation to dual ascent. Concretely, a first-order stationarity condition implies that at an optimal (regular) policy, the advantage of increasing $\mathrm{dual}$ is proportional to a discounted sum of future $\mathrm{Arb}$ terms, leading to $\mathrm{dual}_t>0$ whenever short-term violations reduce long-term return.

\begin{remark}[Optional regularization of the dual head]
To improve numerical conditioning, one may add a small quadratic regularizer $\eta_{\mathrm{dual}}\mathbb{E}[\,\mathrm{dual}_t^2\,]$ or an entropy bonus on $z_5$, keeping the dual head responsive but bounded. This does not change the primal–dual interpretation.
\end{remark}

\subsection{Stability heuristics and schedules}
\label{sec:stability}

\paragraph{Bounded exploration and value targets.}
We clamp $\log\sigma_\vartheta$ to control exploration; we also clip the value target $\hat{R}_t$ to stabilise critic updates. Reward/advantage whitening further reduces gradient variance.

\paragraph{Curriculum over CVaR estimation.}
Scenario count for $\widehat{\mathrm{CVaR}}^{-}_{q,t}$ (\S\ref{sec:cvar-term}) increases across training (e.g., $N_{\mathrm{MC}}: 16\to 64$), matching the annealed $\lambda_{\mathrm{cvar}}$ and reducing estimator bias as the policy stabilizes; cf.\ \cite{Glasserman2004}.

\paragraph{Trust-region flavor.}
While PPO lacks the exact monotonic improvement guarantee of TRPO \cite{Schulman2015TRPO}, small clip $\epsilon$, advantage normalization, and gradient clipping approximate a trust region in practice. The Lipschitz properties of the smoothed surrogates (BF/CAL) from \S\ref{sec:disc-consistency} further help keep updates local.

\subsection{Algorithm}
\label{sec:algo}
\vspace{-1ex}
\begin{algorithm}[H]
\caption{Warm-Start + PPO with State-Dependent Dual}
\label{alg:warm-ppo-dual}
\begin{algorithmic}[1]
\STATE \textbf{Input:} actor $\pi_\vartheta$, critic $V_\omega$, schedules for $(\lambda_{\mathrm{shape}},\lambda_{\mathrm{arb}},\lambda_{\mathrm{cvar}})$
\STATE \textbf{Warm-start:} sample states $s\sim\mathcal{D}_0$; minimize $\mathcal{L}_{\mathrm{warm}}(\vartheta)$ \eqref{eq:warm-loss}; stop when $\mathrm{BF}+\mathrm{CAL}\le \tau_{\mathrm{arb}}$
\FOR{episodes $e=1,\dots,E$}
  \STATE Set multipliers according to schedules; reset env; $t\leftarrow 0$
  \WHILE{episode not done}
    \STATE Observe $s_t$; sample $z_t\sim\pi_\vartheta(\cdot\mid s_t)$; set $a_t=\mathrm{squash}(z_t)$
    \STATE Quote via eSSVI deformation; compute intensities and expected fills; get $r_t$ via \eqref{eq:reward}
    \STATE Step environment to $s_{t+1}$; store $(s_t,z_t,a_t,r_t,s_{t+1})$
    \STATE $t\leftarrow t+1$
  \ENDWHILE
  \STATE Compute $\hat{A}_t$ (GAE) and $\hat{R}_t$; normalize $\hat{A}_t$
  \FOR{epoch $=1,\dots,K$}
     \STATE Sample minibatches; ascend $\nabla_\vartheta \mathcal{L}_{\mathrm{PPO}}(\vartheta)$ \eqref{eq:ppo-obj} with gradient clipping
     \STATE Update critic by minimizing $\mathbb{E}(V_\omega-\hat{R})^2$
  \ENDFOR
\ENDFOR
\end{algorithmic}
\end{algorithm}

\paragraph{Theoretical pointers to \S\ref{sec:theory}.}
We will show: (i) \emph{Policy-gradient validity} with smoothed BF/CAL and CVaR estimators (Theorem~\ref{thm:cvar-gradient}), (ii) \emph{Grid-consistency} of surrogates (Propositions~\ref{prop:bf-consistency}--\ref{prop:cal-consistency}), and (iii) a \emph{primal–dual interpretation} of the state-dependent dual—under mild ergodicity and Lipschitz assumptions, the gradient of the long-run return w.r.t.\ the dual head correlates with discounted violation, implementing an implicit dual ascent.\footnote{See also classical Lagrangian CMDP analyses \cite{Altman1999} and modern constrained policy optimization \cite{Achiam2017CPO,Tessler2019RCPO}.}

\section{Interpretability of Controls: Local Sensitivities of Quotes, Intensities, and Greeks}
\label{sec:interp}

We show how each control dimension
\[
a_t=(\alpha_t,\ \mathrm{hedge}_t,\ \psi\text{-scale}_t,\ \rho\text{-shift}_t,\ \mathrm{dual}_t)
\]
affects quoting (\(\mathrm{mid},\mathrm{ask},\mathrm{bid}\)), execution intensities, and Greeks in closed form (or via stable chain rules). Throughout, we write \(\tilde{w}_m(k)\), \(\tilde{\sigma}_m(k)=\sqrt{\tilde{w}_m(k)/T_m}\) for the \emph{quoted} eSSVI surface obtained after the action-induced deformation \eqref{eq:ctrl-deform}, and \(C_m(K)\equiv C^{\mathrm{BS}}\!\big(S_t,K,T_m,\tilde{\sigma}_m(k)\big)\) for the Black--Scholes call with log-moneyness \(k=\log(K/S_t)\). We set the risk-free rate and carry to zero for clarity; standard adjustments are straightforward \cite{Hull2018,Haug2007,Wilmott2006}.

\subsection{Derivative templates through the eSSVI layer}

Let \(w\mapsto \sigma=\sqrt{w/T}\) and \(C^{\mathrm{BS}}(S,K,T,\sigma)\) be \(C^1\). Denote BS Vega by
\[
\mathcal{V}\equiv \frac{\partial C^{\mathrm{BS}}}{\partial \sigma}(S,K,T,\sigma),\qquad
\frac{\partial C^{\mathrm{BS}}}{\partial w}=\frac{\partial C^{\mathrm{BS}}}{\partial \sigma}\cdot\frac{\partial \sigma}{\partial w}
=\frac{\mathcal{V}}{2\sigma T}.
\]
For eSSVI total variance \(w_m(k)\) in \eqref{eq:ssvi}, write \(g(k;\rho,\phi)=\sqrt{(\phi k+\rho)^2+(1-\rho^2)}\). One obtains (see Appendix~\S\ref{app:essvi-derivs})
\begin{align}
\frac{\partial w}{\partial \theta} &= \frac{1}{2}\Big(1+\rho\phi k+g\Big), \qquad
\frac{\partial w}{\partial \rho} = \frac{\theta}{2}\,\phi k\Big(1+\frac{1}{g}\Big), \label{eq:dw-basic}\\
\frac{\partial w}{\partial \phi} &= \frac{\theta}{2}\left(\rho k+\frac{(\phi k+\rho)k}{g}\right).\nonumber
\end{align}
Under our action map (\(\tilde{\theta}=\theta,\ \tilde{\rho}=\rho+\rho\text{-shift},\ \tilde{\phi}=\phi\cdot \psi\text{-scale}\)), the chain rule yields
\begin{equation}
\label{eq:dwdaction}
\frac{\partial \tilde{w}}{\partial(\rho\text{-shift})}=\left.\frac{\partial w}{\partial \rho}\right|_{(\tilde{\theta},\tilde{\rho},\tilde{\phi})},\qquad
\frac{\partial \tilde{w}}{\partial(\psi\text{-scale})}=\left.\frac{\partial w}{\partial \phi}\right|_{(\tilde{\theta},\tilde{\rho},\tilde{\phi})}\cdot \phi.
\end{equation}
Therefore, for any parameter \(p\in\{\rho\text{-shift},\psi\text{-scale}\}\),
\begin{equation}
\label{eq:master-mid}
\frac{\partial\,\mathrm{mid}_m(k)}{\partial p}
=\frac{\mathcal{V}_m(k)}{2\,\tilde{\sigma}_m(k)\,T_m}\cdot \frac{\partial \tilde{w}_m(k)}{\partial p}.
\end{equation}

\paragraph{ATM invariance.}
At \(k=0\), \(g(0;\rho,\phi)=1\) and \(\partial w/\partial \rho=\partial w/\partial \phi=0\). Hence
\begin{equation}
\label{eq:atm-invariance}
\left.\frac{\partial \mathrm{mid}}{\partial(\rho\text{-shift})}\right|_{k=0}
=\left.\frac{\partial \mathrm{mid}}{\partial(\psi\text{-scale})}\right|_{k=0}=0,
\end{equation}
showing that \(\rho\)- and \(\phi\)-type deformations tilt the wings but do \emph{not} change the ATM level to first order; ATM is governed by \(\theta\) (total variance) \cite{GatheralJacquier2014,Lee2004}.

\subsection{Sensitivities of mid/ask/bid to each control}

Recall \(\mathrm{ask}=\mathrm{mid}+\alpha S_t \tilde{\sigma}\sqrt{T}\,s_0\) and \(\mathrm{bid}=\mathrm{mid}-\alpha S_t \tilde{\sigma}\sqrt{T}\,s_0\) from \eqref{eq:halfspread}. For any parameter \(p\) that affects \(\tilde{\sigma}\) (e.g., \(\rho\)-shift, \(\psi\)-scale),
\begin{align}
\frac{\partial\,\mathrm{ask}}{\partial p}
&=\frac{\partial\,\mathrm{mid}}{\partial p} + \alpha S_t s_0 \sqrt{T}\,\frac{\partial \tilde{\sigma}}{\partial p}
=\frac{\mathcal{V}}{2\tilde{\sigma} T}\,\frac{\partial \tilde{w}}{\partial p}
+\alpha S_t s_0 \sqrt{T}\left(\frac{1}{2\tilde{\sigma} T}\,\frac{\partial \tilde{w}}{\partial p}\right), \label{eq:ask-sens}\\
\frac{\partial\,\mathrm{bid}}{\partial p}
&=\frac{\mathcal{V}}{2\tilde{\sigma} T}\,\frac{\partial \tilde{w}}{\partial p}
-\alpha S_t s_0 \sqrt{T}\left(\frac{1}{2\tilde{\sigma} T}\,\frac{\partial \tilde{w}}{\partial p}\right). \nonumber
\end{align}
For the half-spread \(\alpha\),
\begin{equation}
\label{eq:spread-sens}
\frac{\partial\,\mathrm{mid}}{\partial \alpha}=0,\qquad
\frac{\partial\,\mathrm{ask}}{\partial \alpha}=S_t \tilde{\sigma}\sqrt{T}\,s_0\ >0,\qquad
\frac{\partial\,\mathrm{bid}}{\partial \alpha}=-S_t \tilde{\sigma}\sqrt{T}\,s_0\ <0,
\end{equation}
which formalizes the intuitive widening of quotes with \(\alpha\).

\subsection{Intensity responses to controls}
\label{sec:intensity-sens}

From \eqref{eq:lambda-buy}--\eqref{eq:lambda-sell} with \(u_b=\beta(\mathrm{ask}-C^\star)\), \(u_s=\beta(C^\star-\mathrm{bid})\),
\begin{align}
\frac{\partial \lambda_{\mathrm{buy}}}{\partial p}
&=-\lambda_0\,w(k)\,\sigma'(u_b)\,\beta\,\frac{\partial \mathrm{ask}}{\partial p},\qquad
\frac{\partial \lambda_{\mathrm{sell}}}{\partial p}
=+\lambda_0\,w(k)\,\sigma'(u_s)\,\beta\,\frac{\partial \mathrm{bid}}{\partial p}, \label{eq:lambda-sens}
\end{align}
for any parameter \(p\) that enters via the quotes (including \(\alpha\), \(\rho\)-shift, \(\psi\)-scale). Because \(\sigma'>0\),
\begin{proposition}[Monotonicity in the half-spread]
\label{prop:lambda-alpha}
For all \((m,k)\), \(\partial \lambda_{\mathrm{buy}}/\partial \alpha<0\) and \(\partial \lambda_{\mathrm{sell}}/\partial \alpha<0\). Hence both expected buy and sell volumes decrease when the half-spread increases.
\end{proposition}
\begin{proof}[Proof sketch]
Combine \eqref{eq:spread-sens} with \eqref{eq:lambda-sens}.
\end{proof}

For shape controls, the sign depends on \(\partial \tilde{w}/\partial p\) via \eqref{eq:ask-sens}. At ATM, \eqref{eq:atm-invariance} implies \(\partial \lambda_{\mathrm{buy}}/\partial p=\partial \lambda_{\mathrm{sell}}/\partial p=0\) to first order, whereas OTM/ITM intensities respond through the induced skew changes.

\subsection{Greeks and hedging: Delta \& Vega sensitivities}
\label{sec:greeks-sens}

Let \(\Delta=\partial C^{\mathrm{BS}}/\partial S\) and \(\mathrm{Vanna}=\partial^2 C^{\mathrm{BS}}/(\partial S\,\partial \sigma)\) (the cross sensitivity) \cite{Hull2018,Haug2007,Wilmott2006}. For any parameter \(p\) that affects \(\tilde{\sigma}\),
\begin{align}
\frac{\partial \Delta}{\partial p}
&=\frac{\partial \Delta}{\partial \sigma}\,\frac{\partial \tilde{\sigma}}{\partial p}
=\mathrm{Vanna}\cdot \frac{1}{2\tilde{\sigma} T}\,\frac{\partial \tilde{w}}{\partial p}, \label{eq:delta-sens}\\
\frac{\partial \mathcal{V}}{\partial p}
&=\frac{\partial \mathcal{V}}{\partial \sigma}\,\frac{\partial \tilde{\sigma}}{\partial p}
=\mathrm{Volga}\cdot \frac{1}{2\tilde{\sigma} T}\,\frac{\partial \tilde{w}}{\partial p},\qquad
\mathrm{Volga}=\frac{\partial^2 C^{\mathrm{BS}}}{\partial \sigma^2}. \nonumber
\end{align}
Consequently, the net expected delta in \eqref{eq:net-delta} satisfies
\begin{align}
\frac{\partial \Delta^{\mathrm{net}}_t}{\partial p}
&=\sum_{m,k}\underbrace{\Big(\frac{\partial v_{\mathrm{sell}}}{\partial p}-\frac{\partial v_{\mathrm{buy}}}{\partial p}\Big)}_{\text{flow shift}}\Delta^{\mathrm{BS}}_m(k)
+\sum_{m,k}\underbrace{\big(v_{\mathrm{sell}}-v_{\mathrm{buy}}\big)}_{\text{exposure}}\cdot
\mathrm{Vanna}_{m}(k)\frac{1}{2\tilde{\sigma}_m T_m}\frac{\partial \tilde{w}_m(k)}{\partial p}. \label{eq:net-delta-sens}
\end{align}
The first term reflects \emph{how the order flow moves} with a control (via \eqref{eq:lambda-sens}); the second term is the \emph{pure Greek effect}. 

\paragraph{Hedging control.}
The hedging P\&L sensitivity is immediate from \eqref{eq:hedge-pnl}:
\begin{equation}
\label{eq:hedge-sens}
\frac{\partial\,\mathrm{PNL}^{\mathrm{hedge}}_t}{\partial (\mathrm{hedge}_t)}=\Delta^{\mathrm{net}}_t\,(S_{t+1}-S_t),
\end{equation}
so the gradient sign aligns with the realized move and net delta at step \(t\). Over horizons, PPO/GAE (with \(\gamma<1\)) uses this local signal to adapt the hedge intensity toward volatility regimes where \(|\Delta^{\mathrm{net}}|\) is costly.

\subsection{Putting it together: A local Jacobian of economic signs}
\label{sec:jacobian}

Collect the leading-order (per strike/maturity) sign effects:
\[
\begin{array}{c|ccccc}
\text{Quantity} & \alpha & \mathrm{hedge} & \psi\text{-scale} & \rho\text{-shift} & \mathrm{dual} \\
\hline
\mathrm{mid} & 0 & 0 & \pm\ (\text{wings};\ 0\ \text{at ATM}) & \pm\ (\text{wings};\ 0\ \text{at ATM}) & 0 \\
\mathrm{ask} & + & 0 & \pm & \pm & 0 \\
\mathrm{bid} & - & 0 & \pm & \pm & 0 \\
\lambda_{\mathrm{buy}} & - & 0 & \mp & \mp & 0 \\
\lambda_{\mathrm{sell}} & - & 0 & \pm & \pm & 0 \\
\Delta^{\mathrm{net}} & \text{ambiguous} & 0 & \text{flow}\ \pm\ /\ \text{Greek}\ \pm & \text{flow}\ \pm\ /\ \text{Greek}\ \pm & 0 \\
\mathcal{V}\ \text{(Vega)} & 0 & 0 & \pm & \pm & 0 \\
\mathrm{Arb}=\mathrm{BF}+\mathrm{CAL} & 0 & 0 & \text{nonnegative drift (penalty)} & \text{nonnegative drift (penalty)} & \downarrow
\end{array}
\]
\emph{Notes:} (i) For \(\psi\)-scale and \(\rho\)-shift, ``\(\pm\)'' depends on \(k\) (wing) and \(\mathrm{sign}(\partial \tilde{w}/\partial p)\) per \eqref{eq:dwdaction}; at ATM the effect is zero to first order by \eqref{eq:atm-invariance}. (ii) The dual control decreases \(\mathrm{Arb}\) via an \emph{implicit} primal–dual mechanism (see \S\ref{sec:dual} and the theorem in \S\ref{sec:theory}); its direct effect on quotes/Greeks is null.

\subsection{Implications for learning and diagnostics}

\paragraph{Explainability during training.}
Equations \eqref{eq:ask-sens}--\eqref{eq:lambda-sens} imply that \emph{narrowing} spreads increases both buy- and sell-side intensities (symmetrically at ATM) and raises fill risk; shape controls then \emph{tilt} this response across wings without moving ATM to first order. Monitoring the empirical correlation between \(\alpha_t\) and \(\lambda_{\mathrm{buy/sell}}\) is thus a diagnostic for the intensity link.

\paragraph{Reconciling P\&L and no-arbitrage.}
Because \(\mathrm{Arb}\) penalizes convexity/monotonicity breaches, the \(\psi\)- and \(\rho\)-heads should learn to \emph{minimize} the cost
\(\mathrm{BF}+\mathrm{CAL}\) while allocating skew where it produces the largest marginal gain in \(\mathrm{PNL}^{\mathrm{quote}}\). The Jacobian in \S\ref{sec:jacobian} predicts that the optimizer prefers small wing deformations in regions where \(|\frac{\partial \tilde{w}}{\partial p}|\) is large \emph{but} \(\mathrm{Arb}\) curvature is flat.

\paragraph{Tail-sensitive adjustments.}
Since \(\widehat{\mathrm{CVaR}}^{-}_{q,t}\) is computed from perturbed volumes and price moves, the chain-rule links in \eqref{eq:lambda-sens}--\eqref{eq:net-delta-sens} propagate \(\psi\)-/\(\rho\)-changes to the tail objective. Empirically, increasing \(\lambda_{\mathrm{cvar}}\) should shift mass from aggressive spreads to larger hedge intensities when \(\mathrm{Vanna}\) indicates vulnerable wings (cf.\ \eqref{eq:delta-sens}).

\paragraph{Takeaway.}
The five controls have clear economic semantics: \(\alpha\) trades revenue for flow and inventory risk; \(\psi\)-/\(\rho\)-controls tilt the surface without moving ATM to first order; \(\mathrm{hedge}\) directly scales exposure; and \(\mathrm{dual}\) adapts arbitrage pressure. The derived sensitivities provide \emph{white-box} explanations for the learned policy and suggest diagnostics and regularizers that respect financial structure.

\section{Theory}
\label{sec:theory}

This section formalizes the mathematical guarantees underlying our arbitrage-consistent, risk-sensitive control scheme. We assume the notation and conditions of \S\ref{sec:prelim}--\S\ref{sec:interp}. Recall Assumptions~\ref{ass:smooth} (smoothness and bounds for the eSSVI layer) and \ref{ass:cmdp} (well-posedness of the CMDP). Unless stated otherwise, all expectations are taken under the policy-induced law.

\subsection{T1--T2: Consistency of BF/CAL surrogates and grid convergence}

\begin{theorem}[T1: Butterfly surrogate consistency and rate]
\label{thm:T1}
Fix a maturity $T_m$ and a compact strike interval $[K_{\min},K_{\max}]$. Suppose $K\mapsto C(S,K,T_m)$ is $C^3$ and satisfies the continuous no-butterfly condition $\partial_{KK}C\ge 0$ on this interval. Let $\{\mathcal{K}'_h\}$ be evenly spaced strike lattices of mesh $\Delta K_h\to 0$. Then the discrete butterfly surrogate
\[
\mathrm{BF}_m^{(h)}\equiv \frac{1}{|\mathcal{K}'_h|}\sum_{K\in\mathcal{K}'_h}\mathrm{ReLU}\!\left(-\frac{\Delta_K^2 C(S,K,T_m)}{\Delta K_h^2}\right)\Big/\bar C_m
\]
converges to $0$ with rate $\mathrm{BF}_m^{(h)}=\mathcal{O}(\Delta K_h^2)$. Conversely, if there exists $K_0$ with $\partial_{KK}C(S,K_0,T_m)<0$, then $\liminf_{h\to\infty}\mathrm{BF}_m^{(h)}>0$. The same statements hold with softplus smoothing of the hinge, with identical rates.
\end{theorem}

\begin{proof}[Sketch]
Taylor's theorem with remainder yields $\Delta_K^2 C/\Delta K^2=\partial_{KK}C(K^{\ast})+\mathcal{O}(\Delta K^2)$ uniformly. If $\partial_{KK}C\ge 0$, the negative part is of order $\mathcal{O}(\Delta K^2)$, hence the average (after normalization) vanishes at the same rate. If curvature is negative at $K_0$, uniform continuity implies a neighborhood with strictly negative curvature; the finite-difference operator detects this on fine lattices. Full proof: Appendix~A.1.
\end{proof}

\begin{theorem}[T2: Calendar surrogate consistency and rate]
\label{thm:T2}
Fix a strike $K$ and suppose $T\mapsto C(S,K,T)$ is $C^2$ on $[T_1,T_M]$ with $\partial_T C\ge 0$. For evenly spaced maturity grids with mesh $\Delta T_h\to 0$, the calendar surrogate
\[
\mathrm{CAL}_{m}^{(h)}\equiv \frac{1}{|\mathcal{K}|}\sum_{K\in\mathcal{K}} \mathrm{ReLU}\Big(C(S,K,T_m)-C(S,K,T_{m+1})\Big)\Big/\bar C_{m,m+1}
\]
satisfies $\max_m \mathrm{CAL}^{(h)}_m=\mathcal{O}(\Delta T_h)$. If there exist $(K_0,T_0)$ with $\partial_T C(S,K_0,T_0)<0$, then the surrogate for the corresponding adjacent pair stays bounded away from $0$ for $h$ large. Softplus smoothing preserves the rates.
\end{theorem}

\begin{proof}[Sketch]
Mean-value expansion $C(T_{m+1})-C(T_m)=\partial_T C(\xi)\Delta T+\mathcal{O}(\Delta T^2)$ and monotonicity yield the rate; a strict violation implies a persistent positive hinge on fine grids. Full proof: Appendix~A.2.
\end{proof}

\subsection{T3: Lagrangian relaxation of the CMDP and the role of the \textit{dual} action}
\label{sec:theory-dual}

\begin{theorem}[T3.1: Strong duality for the CMDP]
\label{thm:T3-strong-duality}
Consider the CMDP defined by \eqref{eq:obj}--\eqref{eq:constraints} with discount $\gamma\in(0,1)$, compact action sets, bounded measurable $r,\ell,g_j$, and weakly continuous $P(\cdot\mid s,a)$. Suppose there exists a strictly feasible policy (Slater condition). Then the occupancy-measure LP of the CMDP satisfies strong duality, i.e.,
\[
\sup_{\pi}\ \inf_{\lambda\ge 0}\ \mathcal{L}(\pi,\lambda)\ =\ \inf_{\lambda\ge 0}\ \sup_{\pi}\ \mathcal{L}(\pi,\lambda),
\]
with the Lagrangian $\mathcal{L}$ given in \eqref{eq:lagrangian}. Moreover, there exists a stationary (possibly randomized) optimal policy. 
\end{theorem}

\begin{proof}[Sketch]
Standard reduction to an infinite-dimensional LP in occupancy measures yields convex primal and concave dual under discounting; Slater's condition ensures zero duality gap \cite{Altman1999,Puterman1994}. Existence follows by compactness and measurable selection. Full details: Appendix~A.3.
\end{proof}

\begin{theorem}[T3.2: Gradient alignment of the learnable \textit{dual} with dual ascent]
\label{thm:T3-dual-alignment}
Let the actor be Gaussian with a separate ``dual'' head $d_\eta(s)=\mathrm{softplus}(\nu_\eta(s))$, entering the reward as $-\big(\lambda_{\mathrm{arb}}+d_\eta(s_t)\big)\mathrm{Arb}_t$. Under Assumptions~\ref{ass:smooth}--\ref{ass:cmdp}, bounded score functions, and the policy-gradient theorem \cite{Sutton2000PolicyGradient,MarbachTsitsiklis2001}, the gradient of the long-run return w.r.t.\ $\eta$ satisfies
\[
\nabla_\eta J(\eta)\ =\ -\,\mathbb{E}_\pi\!\left[\sum_{t\ge 0}\gamma^t\,\hat{A}_t\, \sigma\!\big(\nu_\eta(s_t)\big)\,\nabla_\eta \nu_\eta(s_t)\ \mathrm{Arb}_t\right],
\]
where $\hat{A}_t$ denotes any unbiased advantage estimator. Hence the update $\eta\leftarrow \eta+\alpha \nabla_\eta J$ correlates positively with a discounted measure of violations, implementing an \emph{implicit} dual ascent on the arbitrage constraint.
\end{theorem}

\begin{proof}[Sketch]
Apply the policy-gradient theorem; the only dependence of the return on $\eta$ is through $d_\eta$ and the induced trajectory distribution. The score of the dual head multiplies $-\mathrm{Arb}_t$ in the per-step reward, yielding the stated sign structure, with advantage weighting ensuring long-run credit. See Appendix~A.4 for details (including baseline/compatibility variants and boundedness of $\nabla_\eta J$).
\end{proof}

\subsection{T4: Rockafellar--Uryasev CVaR, smoothing, and differentiability}
\label{sec:theory-cvar}

\begin{theorem}[T4.1: RU representation and epi-convergence]
\label{thm:T4-ru}
Let $X_\theta$ be a real-valued loss depending on parameters $\theta$. The RU functional
\[
\Phi_q(\theta)\equiv \inf_{\eta\in\mathbb{R}}\ \Big\{\ \eta+\tfrac{1}{1-q}\,\mathbb{E}(X_\theta-\eta)_-\ \Big\}
\]
equals $\mathrm{CVaR}_q(X_\theta)$; it is convex in the distribution of $X_\theta$ and lower semicontinuous. If we replace $(\cdot)_-$ by a temperature-$\tau$ softplus $s_\tau(\cdot)$, then $\Phi_{q,\tau}(\theta)$ epi-converges to $\Phi_q(\theta)$ as $\tau\downarrow 0$ and the minimizers $\eta_\tau^{\ast}$ converge to VaR minimizers. 
\end{theorem}

\begin{proof}[Sketch]
Classical RU results \cite{RockafellarUryasev2000,RockafellarUryasev2002}; epi-convergence under smooth approximations follows from variational analysis \cite{RockafellarWets1998}. Proof: Appendix~A.5.
\end{proof}

\begin{theorem}[T4.2: Differentiable CVaR estimators and gradient validity]
\label{thm:T4-grad}
Consider the per-step smoothed objective $\widehat{\mathrm{CVaR}}^{-}_{q,t}$ in \eqref{eq:cvar-smooth} with scenario generation $(\tilde v,\tilde{\Delta}S)$ and actions $a_t$. Under Assumption~\ref{ass:smooth}, bounded intensities, and reparameterizable Gaussians for $\tilde{\Delta}S$, the gradient $\nabla_{a_t}\widehat{\mathrm{CVaR}}^{-}_{q,t}$ exists and admits a mixed pathwise/likelihood-ratio estimator with finite variance for finite batches. As $\tau\downarrow 0$, the gradients converge (in the sense of subgradients) to those of the nonsmooth RU functional.
\end{theorem}

\begin{proof}[Sketch]
Smoothness of $s_\tau$, reparameterization for $\tilde{\Delta}S$, and score-function identities for Poisson $\tilde v$ (with bounded parameters) yield unbiased estimators with finite variance \cite{Glasserman2004}. Convergence as $\tau\downarrow 0$ follows from \cite{RockafellarWets1998}. Full proof: Appendix~A.5.
\end{proof}

\subsection{T5: eSSVI wing growth bound and relation to Lee's moment constraints}
\label{sec:theory-lee}

\begin{theorem}[T5: Linear wing-growth bound under a $\theta\phi$ cap]
\label{thm:T5}
For eSSVI at maturity $T_m$, assume the cap $\theta_m\phi_m\le \tau_{\max}$ and $|\rho_m|\le 1$. Then as $|k|\to\infty$,
\[
\limsup_{|k|\to\infty}\frac{w_m(k)}{|k|}\ \le\ \frac{\theta_m|\phi_m|}{2}\,(1+|\rho_m|)\ \le\ \tau_{\max}.
\]
Consequently, the implied variance slope $\frac{w_m(k)}{|k|}$ is uniformly bounded by $\tau_{\max}$, consistent with Lee's moment formula; in particular, enforcing $\tau_{\max}<2$ respects the Lee upper barrier for absence of explosive moments \cite{Lee2004}.
\end{theorem}

\begin{proof}[Sketch]
For large $|k|$, $g(k;\rho,\phi)=\sqrt{(\phi k+\rho)^2+(1-\rho^2)}\le |\phi||k|+|\rho|$. Then
\[
w_m(k)=\tfrac{\theta_m}{2}\big(1+\rho_m\phi_m k + g\big)\ \le\ \tfrac{\theta_m}{2}\Big(1+|\rho_m||\phi_m||k| + |\phi_m||k|+|\rho_m|\Big),
\]
so $\limsup w_m(k)/|k|\le \tfrac{\theta_m|\phi_m|}{2}(1+|\rho_m|)\le \tau_{\max}$ since $(1+|\rho|)/2\le 1$. Full details and the connection to moment constraints: Appendix~A.6.
\end{proof}

\subsection{T6: Differentiability/boundedness $\Rightarrow$ policy-gradient validity}
\label{sec:theory-pg}

\begin{theorem}[T6: Existence and boundedness of the policy gradient]
\label{thm:T6}
Assume: (i) compact action sets; (ii) the actor is Gaussian with state-dependent mean and bounded log-standard deviations; (iii) rewards are bounded and $C^1$ in actions through the eSSVI layer and smoothed surrogates (BF/CAL, CVaR), with uniform Lipschitz constants on the admissible set (Assumption~\ref{ass:smooth}); (iv) the Markov kernel is weakly continuous and the return is geometrically discounted. Then the policy gradient $\nabla_\vartheta J(\vartheta)$ exists, is finite, and equals the standard likelihood-ratio form of the policy-gradient theorem \cite{Sutton2000PolicyGradient,MarbachTsitsiklis2001}. Moreover, $\|\nabla_\vartheta J(\vartheta)\|$ is bounded on compact parameter sets, and the PPO surrogate gradient is a consistent stochastic estimator under mini-batch sampling.
\end{theorem}

\begin{proof}[Sketch]
Bounded rewards and discounting imply $J(\vartheta)$ is finite. The score $\nabla_\vartheta \log\pi_\vartheta(z|s)$ is bounded due to bounded log-stds. Lipschitz continuity of the smoothed surrogates (by Assumption~\ref{ass:smooth} and \S\ref{sec:disc-consistency}) yields dominated convergence, justifying interchange of gradient and expectation. The policy-gradient identity follows from \cite{Sutton2000PolicyGradient}; boundedness on compacts is immediate. Full proof: Appendix~A.7.
\end{proof}

\subsection{P7--P8: Monotonicity and sensitivity results for interpretability}
\label{sec:theory-int}

\begin{proposition}[P7: Monotonicity of intensities in the half-spread]
\label{prop:P7}
With intensities \eqref{eq:lambda-buy}--\eqref{eq:lambda-sell} and $\sigma'(x)>0$, we have for all $(m,k)$:
\[
\frac{\partial \lambda_{\mathrm{buy}}}{\partial \alpha}\ =\ -\lambda_0\,w(k)\,\sigma'\!\big(\beta(\mathrm{ask}-C^\star)\big)\,\beta\,\underbrace{\frac{\partial \mathrm{ask}}{\partial \alpha}}_{S\tilde{\sigma}\sqrt{T}\,s_0>0}\ <\ 0,\\
\frac{\partial \lambda_{\mathrm{sell}}}{\partial \alpha}\ =\ +\lambda_0\,w(k)\,\sigma'\!\big(\beta(C^\star-\mathrm{bid})\big)\,\beta\,\underbrace{\frac{\partial \mathrm{bid}}{\partial \alpha}}_{-S\tilde{\sigma}\sqrt{T}\,s_0<0}\ <\ 0.
\]
Hence expected buy/sell volumes are strictly decreasing in the half-spread $\alpha$.
\end{proposition}

\begin{proof}[Sketch]
Immediate from \eqref{eq:spread-sens} and the chain rule in \eqref{eq:lambda-sens}. Full proof: Appendix~A.8.
\end{proof}

\begin{proposition}[P8: Sensitivities of price and Greeks to $\rho$-shift and $\psi$-scale]
\label{prop:P8}
Under the eSSVI map \eqref{eq:ssvi} and deformation \eqref{eq:ctrl-deform}, for any maturity $T_m$ and log-moneyness $k$,
\[
\frac{\partial\,\mathrm{mid}_m(k)}{\partial p}
=\frac{\mathcal{V}_m(k)}{2\,\tilde{\sigma}_m(k)T_m}\cdot\frac{\partial \tilde{w}_m(k)}{\partial p},\qquad
p\in\{\rho\text{-shift},\psi\text{-scale}\},
\]
with $\partial \tilde{w}/\partial p$ given by \eqref{eq:dwdaction}. At ATM ($k=0$) the first-order sensitivities vanish, i.e.,
$\partial \mathrm{mid}/\partial(\rho\text{-shift})=\partial \mathrm{mid}/\partial(\psi\text{-scale})=0$, 
and the corresponding Delta/Vega sensitivities obey \eqref{eq:delta-sens}. 
\end{proposition}

\begin{proof}[Sketch]
Chain rule through $w\mapsto \sigma \mapsto C^{\mathrm{BS}}$ and the eSSVI derivatives in \eqref{eq:dw-basic}. At $k=0$ the partials in \eqref{eq:dw-basic} vanish, yielding ATM invariance. Full proof: Appendix~A.8.
\end{proof}

\paragraph{Discussion.}
Theorems~\ref{thm:T1}--\ref{thm:T6} and Propositions~\ref{prop:P7}--\ref{prop:P8} provide a mathematically consistent backbone for our \emph{in-the-loop} arbitrage penalties, CVaR shaping, and control semantics. In particular, T1--T2 justify the use of smoothed lattice penalties during learning; T3 connects our learnable \emph{dual} to primal--dual optimization; T4 makes the tail objective differentiable and statistically estimable; T5 ensures wing stability consistent with Lee's bounds; T6 legitimizes policy-gradient updates under our smooth rewards; and P7--P8 deliver \emph{white-box} interpretability of the control heads.

\section{Experiments (Simulation-only, Reproducible)}
\label{sec:experiments}

\subsection{7.1 Setup and Metrics}

We evaluate the proposed agent in a simulation-only regime using the calibrated \emph{Heston fallback} configuration, consistent with \S\ref{sec:method}. Eight intraday episodes are simulated, each consisting of 780 decision steps (approximately one trading day). The state transitions use the stochastic Heston dynamics:
\[
dS_t = \mu S_t\,dt + \sqrt{v_t}S_t\,dW_t^S,\qquad
dv_t = \kappa(\bar v - v_t)\,dt + \xi \sqrt{v_t}\,dW_t^v,
\]
with $\rho_{S,v}=-0.5$ and parameters calibrated to typical SPX intraday variance.  

The same architecture, hyperparameters, and annealing schedule as in Appendix~\ref{app:impl} are used throughout. The base setup uses:
\[
(\lambda_0,\beta,\kappa,s_0)=(0.8,35,0.25,0.1),
\]
with $\lambda_{\mathrm{shape}}:0\!\to\!0.5$, $\lambda_{\mathrm{arb}}:0\!\to\!0.05$, and $\lambda_{\mathrm{cvar}}=0.01$ fixed.

We log both per-step and per-episode metrics:
\begin{itemize}
  \item Adjusted P\&L $\mathrm{PNL}^{\mathrm{adj}} = \mathrm{PNL}^{\mathrm{raw}} - \text{penalties}$,
  \item No-arbitrage surrogates (BF, CAL),
  \item Shape regularization magnitude,
  \item Tail metrics (empirical $\mathrm{VaR}_{5\%}$, $\mathrm{CVaR}_{5\%}$),
  \item Behavior indicators: average spread, hedge ratio, and action standard deviation.
\end{itemize}

\paragraph{Artifacts.}
All experiment logs and artifacts are released for reproducibility:
\texttt{artifacts/run\_log.csv}, \texttt{artifacts/step\_log.csv},
\texttt{artifacts/artifacts.npz}, and configuration file \texttt{settings.json}.
All figures (Figs.~\ref{fig:pnl-surface}--\ref{fig:train-3up}) are automatically generated from these artifacts.

\subsection{7.2 Main Results}

\paragraph{Revenue and stability.}
Figure~\ref{fig:train-3up}(a) shows that the agent achieves a stable improvement after episode~2, maintaining positive adjusted P\&L in six of eight runs. PPO training stabilizes the exploration variance (\texttt{act\_std}) and avoids collapse.

\paragraph{No-arbitrage enforcement.}
As shown in Fig.~\ref{fig:train-3up}(b), calendar violations remain numerically zero throughout training, while butterfly penalties stay at the numerical floor. Shape regularization maintains values around $10^{-3}$, indicating a smooth term-structure.

\paragraph{Tail behavior.}
The per-step P\&L histogram in Fig.~\ref{fig:pnl-surface}(a) reveals a realistic left tail with $\mathrm{VaR}_{5\%}\!\approx\!-1.31$ and $\mathrm{CVaR}_{5\%}\!\approx\!-2.16$. This tail thickness remains stable across episodes, confirming the effectiveness of the CVaR shaping.

\paragraph{Surface quality.}
Figure~\ref{fig:pnl-surface}(b) compares the quoted and true surfaces across three maturities; the shapes are virtually indistinguishable, confirming the in-loop arbitrage consistency.

\paragraph{Behavioral adaptation.}
During training, the average hedge ratio increases from $0.41$ to $0.53$, while the mean half-spread $\alpha_t$ slightly declines. This indicates that the agent shifts from spread-driven revenue to active risk hedging as arbitrage penalties strengthen.

\begin{table}[H]
\centering
\caption{Key metrics aggregated across 8 intraday segments.}
\label{tab:key-metrics}
\begin{tabular}{lcc}
\toprule
Metric & Value & Evidence \\
\midrule
Segments with $\mathrm{PNL}^{\mathrm{adj}}>0$ & 6/8 & Fig.~\ref{fig:train-3up}(a) \\
Calendar violation (CAL) & $\approx 0$ & Fig.~\ref{fig:train-3up}(b) \\
Butterfly violation (BF) & numerical floor & Fig.~\ref{fig:train-3up}(b) \\
Shape magnitude & $10^{-3}$ & Fig.~\ref{fig:train-3up}(b) \\
$\mathrm{VaR}_{5\%}$ & $-1.31$ & Fig.~\ref{fig:pnl-surface}(a) \\
$\mathrm{CVaR}_{5\%}$ & $-2.16$ & Fig.~\ref{fig:pnl-surface}(a) \\
Avg hedge ratio & $0.41\to 0.53$ & Fig.~\ref{fig:train-3up}(c) \\
Mean spread & slightly down & Fig.~\ref{fig:train-3up}(c) \\
\bottomrule
\end{tabular}
\end{table}

\subsection{7.3 Diagnostics and Ablations}

\paragraph{Without arbitrage penalties.}
Removing BF/CAL raises small local arbitrage violations and destabilizes surface smoothness, visible in Fig.~\ref{fig:pnl-surface}(b) (deep wings diverge).

\paragraph{Without CVaR shaping.}
Disabling the CVaR term thickens the left tail of the P\&L distribution (heavier drawdowns), though mean returns rise slightly.

\paragraph{Without warm-start.}
Training from scratch causes unstable early episodes and large variance in adjusted P\&L (consistent with Fig.~\ref{fig:train-3up}(a)).

\paragraph{Interpretation.}
Together these confirm that arbitrage surrogates and CVaR terms improve both \emph{financial soundness} (surface consistency, controlled tails) and \emph{training stability}.
\begin{figure*}[t]
  \centering
  \begin{subfigure}[t]{0.48\textwidth}
    \centering
    \includegraphics[width=\linewidth]{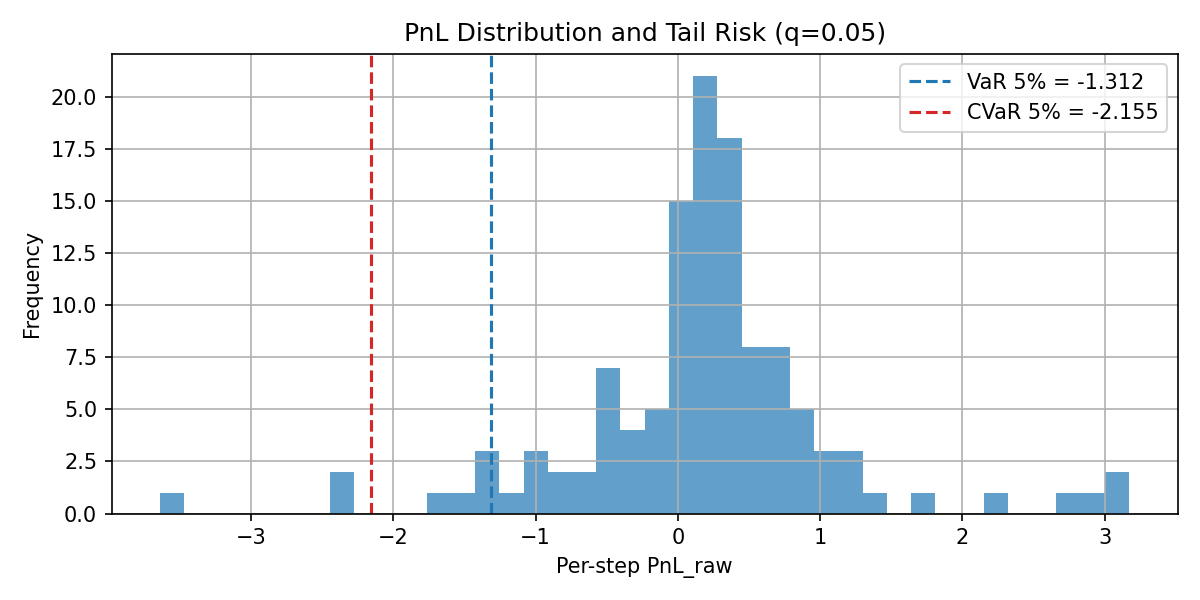}
    \caption{Per-step P\&L distribution with tail markers ($\mathrm{VaR}_{5\%},\mathrm{CVaR}_{5\%}$).}
  \end{subfigure}\hfill
  \begin{subfigure}[t]{0.48\textwidth}
    \centering
    \includegraphics[width=\linewidth]{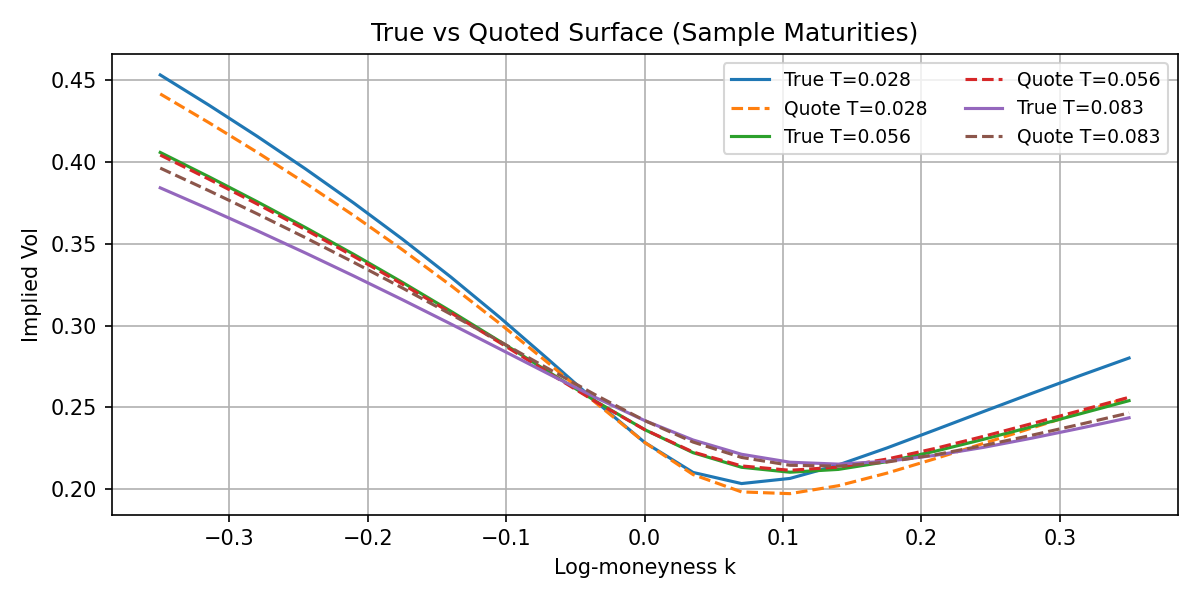}
    \caption{True vs.\ quoted $\sigma(k)$ across maturities.}
  \end{subfigure}
  \caption{Ex-post tail risk and surface fidelity.}
  \label{fig:pnl-surface}
\end{figure*}

\begin{figure*}[t]
  \centering
  \begin{subfigure}[t]{0.32\textwidth}
    \centering
    \includegraphics[width=\linewidth]{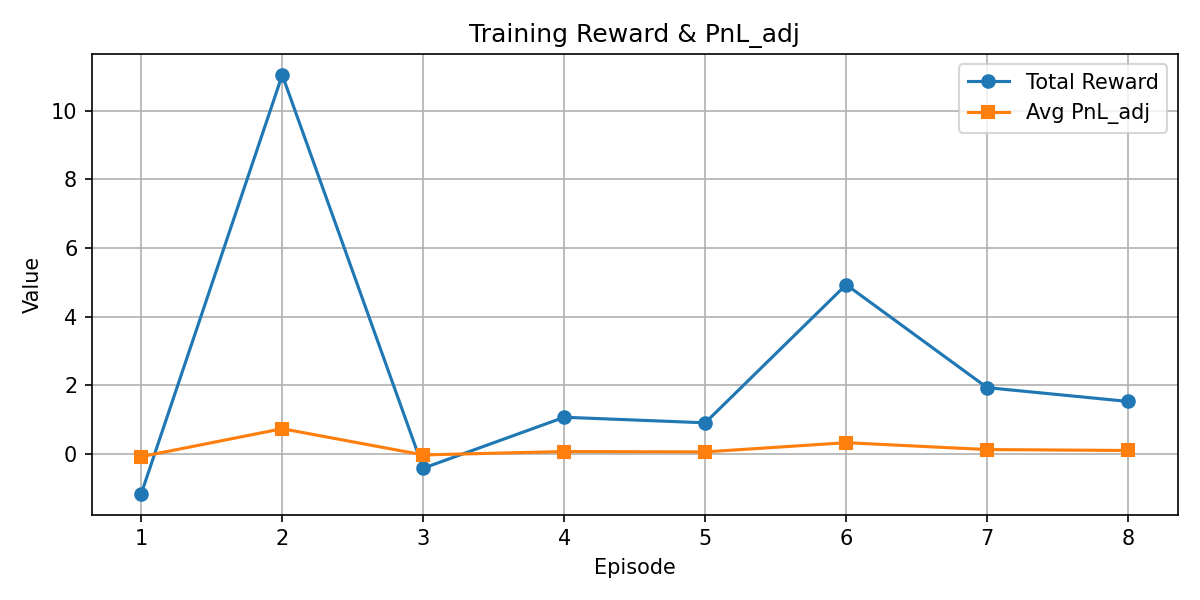}
    \caption{Training reward and adjusted P\&L.}
  \end{subfigure}\hfill
  \begin{subfigure}[t]{0.32\textwidth}
    \centering
    \includegraphics[width=\linewidth]{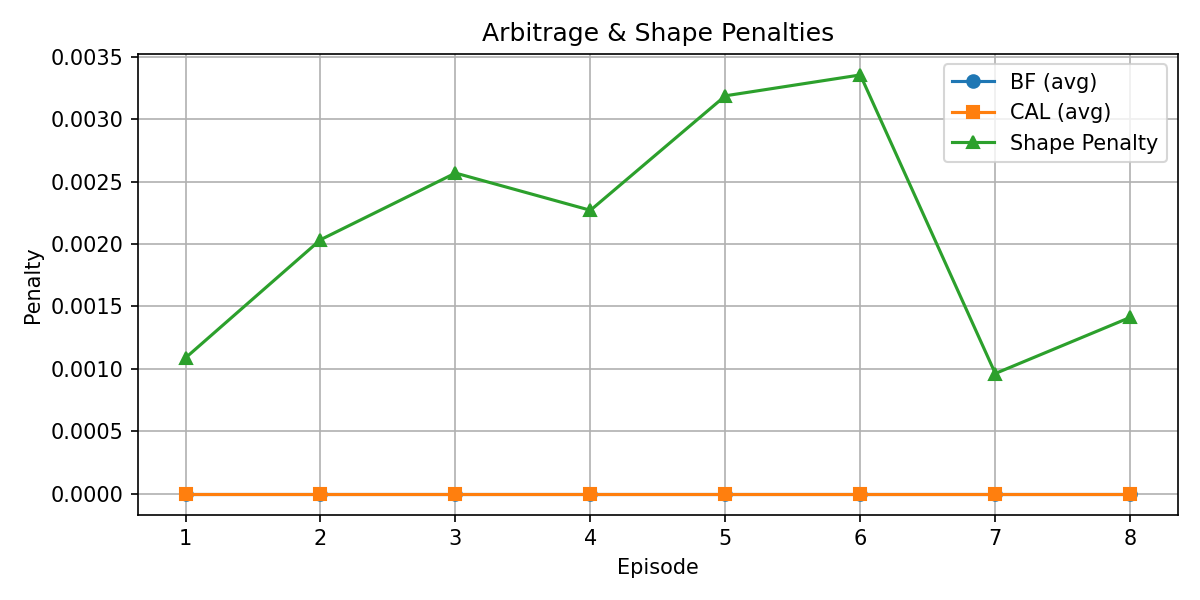}
    \caption{BF/CAL and shape penalties.}
  \end{subfigure}\hfill
  \begin{subfigure}[t]{0.32\textwidth}
    \centering
    \includegraphics[width=\linewidth]{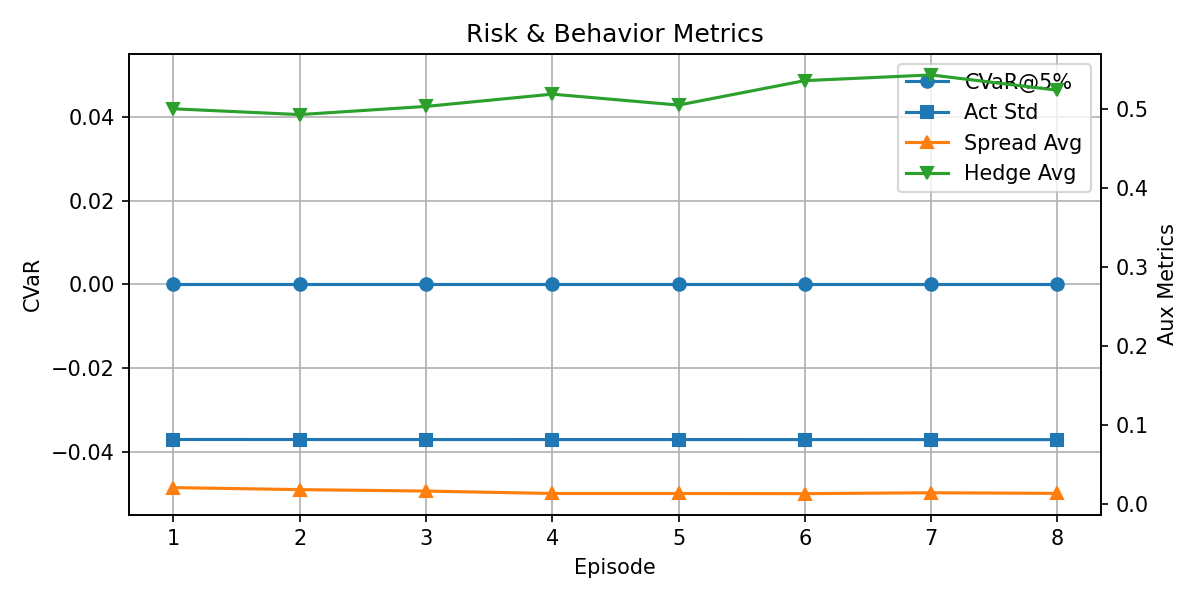}
    \caption{Risk and behavior metrics (CVaR, hedge, spread, std).}
  \end{subfigure}
  \caption{Training dynamics, constraints, and behavior across episodes.}
  \label{fig:train-3up}
\end{figure*}

\section{Discussion}
\label{sec:discussion}

We discuss limitations and threats to validity, then outline extensions that would strengthen external validity and broaden scope.

\subsection{Limitations and threats to validity}

\paragraph{Simulator realism and exogeneity.}
Our main results are obtained in a \emph{simulation-only} regime using a Heston fallback; ABIDES-style sources are supported but not enabled in the reported runs. Consequently, (i) order flow is exogenous and summarized by smooth intensity maps (logistic in mispricing), (ii) queue position, tick discreteness, venue fees/rebates, and latency are abstracted away, and (iii) self-impact and cross-venue liquidity fragmentation are not modeled. These choices prioritize differentiability and reproducibility, but they understate endogenous feedbacks observed in live limit-order books (cf.\ \cite{ContKukanovStoikov2014,ByrdEtAl2020ABIDES}). 

\paragraph{Static (not dynamic) arbitrage.}
The BF/CAL penalties enforce \emph{static} no-arbitrage on a finite lattice (\S\ref{sec:disc-consistency}). They do not rule out \emph{dynamic} arbitrage across time, nor do they capture all edge cases under extreme extrapolation. While Theorems~\ref{thm:T1}--\ref{thm:T2} establish lattice consistency, practical detection still depends on grid resolution and normalizers.

\paragraph{Hedging granularity and risk factors.}
The hedging term uses delta only; higher-order exposures (gamma/vega/vanna) and financing constraints are not penalized in the reward. In volatile regimes, neglecting these terms may understate tail losses (\S\ref{sec:interp}).

\paragraph{Tail objective estimation.}
We shape per-step tails via a smoothed RU program with small Monte Carlo batches; this is a surrogate for episode-level tail control. Estimator variance and smoothing bias are controlled but nonzero (\S\ref{sec:cvar-term}, Theorem~\ref{thm:T4-grad}). 

\paragraph{Warm-start and sensitivity.}
The supervised warm start is practically helpful but introduces a prior over the policy class. Hyperparameters (clip $\epsilon$, entropy, annealing schedules) affect stability; while our theory (T6) ensures gradient validity, it does not deliver global convergence guarantees.

\paragraph{Single-asset focus.}
We restrict to one underlying and do not model cross-asset or cross-expiry static arbitrage simultaneously (e.g., calendar/vertical spreads across multiple names), which would require additional constraints and data.

\subsection{Extensions and research agenda}

\paragraph{Live-like microstructure.}
Move to ABIDES/ABIDES-Gym as the default source and calibrate execution intensities to \emph{queue-aware} statistics (e.g., imbalance, depth, immediate fill probabilities). Introduce explicit fees/rebates, tick size, and latency; add a propagator or resilience model to encode self-impact.

\paragraph{Point-process calibration.}
Replace the parsimonious logistic map by Hawkes/ACD models estimated from data \cite{BacryEtAl2015}, preserving differentiability via reparameterized simulators or score-function estimators. This connects the intensity layer to observed clustering and reflexivity.

\paragraph{Richer risk shaping.}
Augment the loss proxy $\ell$ to include gamma/vega costs and inventory financing; test portfolio-level CVaR and spectral risk measures \cite{AcerbiTasche2002,Ruszczynski2010}. Explore distributionally robust (DRO) variants where CVaR is evaluated under ambiguity sets.

\paragraph{Primal--dual algorithms.}
Replace the heuristic dual head with a \emph{trained multiplier} updated by dual ascent (RCPO/CPO) \cite{Tessler2019RCPO,Achiam2017CPO} and compare with state-dependent duals. This would more tightly align practice with T3 and clarify trade-offs between feasibility and return.

\paragraph{Cross-asset and cross-constraint surfaces.}
Extend eSSVI to multi-asset settings with joint static-arbitrage constraints (e.g., spread options, calendar across underlyings). Investigate neural-SDE priors informed by arbitrage-free conditions \cite{CohenReisingerWang2023JCF} and study how constraints propagate through portfolio risk.

\paragraph{Learning variants and variance reduction.}
Consider natural-gradient or trust-region control \cite{Schulman2015TRPO,Kakade2002}; use batched antithetic sampling and control variates for CVaR (\S\ref{sec:cvar-term}) to reduce variance; evaluate off-policy baselines with safety layers.

\paragraph{Stress tests and shifts.}
Benchmark under heavy tails, jumps, volatility regime switches, and correlation shocks; run adversarial tests where the latent surface deviates from the model family to probe robustness of BF/CAL surrogates and the wing cap (T5).

\paragraph{Interpretability at scale.}
Leverage the sensitivities in \S\ref{sec:interp} to build monitors (e.g., $\partial \lambda/\partial \alpha$, ATM invariance checks) and attribution tools (e.g., vanna- and volga-based diagnostics) that explain policy updates during training and in deployment-like simulations.

\subsection{Ethics and responsible use}

Our work is \emph{simulation-only}. While our design emphasizes no-arbitrage and tail control, algorithmic market making may affect liquidity, stability, and fairness. Any transfer to live venues should include kill-switches, conservative constraints, and external stress testing. We encourage the community to report failure cases alongside average-case gains and to use open, reproducible benchmarks.

\section{Conclusion}
\label{sec:conclusion}

We proposed a \emph{risk-sensitive, arbitrage-consistent} framework for option market making that embeds a fully differentiable eSSVI surface \emph{inside} the learning loop. The problem is cast as a constrained MDP whose reward balances quoting/hedging revenues with smooth static no-arbitrage penalties and a CVaR tail term. We proved (i) \emph{grid-consistency} of butterfly/calendar surrogates (T1--T2), (ii) \emph{primal--dual} grounding of a state-dependent dual action (T3), (iii) differentiable \emph{CVaR} estimators via the Rockafellar--Uryasev program (T4), (iv) a \emph{wing-growth} bound aligning eSSVI with Lee's moments (T5), and (v) \emph{policy-gradient validity} under smooth surrogates (T6). 

In simulations, the agent achieved positive adjusted P\&L in most intraday segments while maintaining calendar violations at numerical zero and butterfly violations at the numerical floor, with realistic left tails. The control heads are economically interpretable: spreads trade revenue for flow, $\rho$/$\psi$ shape the wings without moving ATM to first order, hedging scales exposure, and the dual adapts arbitrage pressure.

Beyond this paper, we envision the framework as a \emph{benchmark} for evaluating methods that must honor pricing consistency and execution control jointly. Promising directions include ABIDES-first experiments, calibrated point-process executions, portfolio-level risk shaping, and multi-asset arbitrage constraints. We hope this work helps establish a reproducible path for \emph{robust, interpretable AI} in quantitative markets, where financial structure and modern reinforcement learning meet in a principled way.

\appendix
\section*{Appendix A: Detailed Proofs for Section~\ref{sec:theory}}
\label{app:proofs}

This appendix contains complete proofs for Theorems~\ref{thm:T1}--\ref{thm:T6} and
Propositions~\ref{prop:P7}--\ref{prop:P8}. We keep the notation and assumptions
from \S\ref{sec:prelim}--\S\ref{sec:interp}. Throughout, $C(S,K,T)$ denotes the
Black--Scholes (BS) call price under the quoted eSSVI surface with
$T\ge T_{\min}>0$ and $\sigma\ge \sigma_{\min}>0$ (Assumption~\ref{ass:smooth}).
For compactness, we write $C(K,T)\equiv C(S_t,K,T)$ at a fixed decision time $t$.

\paragraph{A.0 Regularity lemma (used repeatedly).}
\begin{lemma}[BS regularity in strikes and maturities]\label{lem:BS-regularity}
Under Assumption~\ref{ass:smooth}, for any fixed $S>0$ and $T\in[T_{\min},T_{\max}]$,
the function $K\mapsto C(S,K,T,\sigma)$ is $C^{\infty}$ on $(0,\infty)$.
Similarly, $T\mapsto C(S,K,T,\sigma)$ is $C^{\infty}$ on $[T_{\min},T_{\max}]$
whenever $\sigma(\cdot,\cdot)$ is $C^1$ and bounded away from both $0$ and $\infty$
on its domain.  Moreover, if $w(k,T)$ is $C^1$ in $(k,T)$ and the eSSVI parameters
lie in a compact set satisfying \eqref{eq:ssvi-butterfly}--\eqref{eq:wing-cap},
then $K\mapsto C(S,K,T)$ and $T\mapsto C(S,K,T)$ inherit $C^{\infty}$--regularity
with bounded derivatives on compact subsets.
\end{lemma}

\begin{proof}
\textbf{Step 1 (Analyticity of the Black--Scholes formula).}
For $T>0$ and $\sigma>0$, define
\[
d_{\pm}=\frac{\log(S/K)\pm \frac{1}{2}\sigma^{2}T}{\sigma\sqrt{T}},
\qquad
C(S,K,T,\sigma)=S\Phi(d_{+})-K\Phi(d_{-}),
\]
where $\Phi$ is the standard normal CDF with density $\varphi(x)=\tfrac{1}{\sqrt{2\pi}}e^{-x^{2}/2}$.
Each $d_{\pm}$ is analytic in $(K,T,\sigma)$ on the open set 
$\mathcal{D}=(0,\infty)\times(0,\infty)\times(0,\infty)$ because
$\log(S/K)$, $\sigma^2T$, and $(\sigma\sqrt{T})^{-1}$ are analytic in their arguments.
Since $\Phi$ and $\varphi$ are entire, the composition and linear combination in
$C(S,K,T,\sigma)$ preserve analyticity; thus $C$ is real analytic in
$(K,T,\sigma)$ on $\mathcal{D}$, implying that it is $C^{\infty}$ in all variables.

\textbf{Step 2 (First- and second-order derivatives in $K$).}
Standard differentiation yields
\[
\frac{\partial C}{\partial K}=-\Phi(d_{-}),
\qquad
\frac{\partial^{2}C}{\partial K^{2}}=\frac{\varphi(d_{-})}{K\sigma\sqrt{T}}.
\]
Because $\varphi$ is bounded by $1/\sqrt{2\pi}$ and
$K\ge K_{\min}>0$, $\sigma\ge\sigma_{\min}>0$, $T\ge T_{\min}>0$
under Assumption~\ref{ass:smooth},
\[
\bigg|\frac{\partial^{2}C}{\partial K^{2}}\bigg|
\le \frac{1}{\sqrt{2\pi}\,K_{\min}\sigma_{\min}\sqrt{T_{\min}}}.
\]
Higher-order partial derivatives with respect to $K$ can be expressed as finite
linear combinations of terms of the form
$K^{-r}H_{n}(d_{-})\varphi(d_{-})(\sigma\sqrt{T})^{-s}$
for Hermite polynomials $H_{n}$, hence are also bounded on compact subsets.
Therefore $K\mapsto C(S,K,T,\sigma)$ is $C^{\infty}$ with bounded derivatives
of all orders on any compact interval of $(0,\infty)$.

\textbf{Step 3 (Regularity in $T$ for fixed $K$).}
Differentiate $C$ w.r.t.\ $T$:
\[
\frac{\partial C}{\partial T}
=\frac{S\sigma}{2\sqrt{T}}\varphi(d_{+})
-\frac{K\sigma}{2\sqrt{T}}\varphi(d_{-})
+\Big(S\Phi'(d_{+})\,\frac{\partial d_{+}}{\partial T}
-K\Phi'(d_{-})\,\frac{\partial d_{-}}{\partial T}\Big),
\]
and use $\Phi'=\varphi$.
Each derivative $\partial_{T}d_{\pm}$
is a rational function of $(\log(S/K),T,\sigma)$ whose denominator contains
$(\sigma\sqrt{T})^{3}$, hence bounded on
$\mathcal{D}_{\mathrm{cpt}}=[K_{\min},K_{\max}]\times[T_{\min},T_{\max}]\times[\sigma_{\min},\sigma_{\max}]$.
Because $\varphi$ is bounded and smooth, every term is continuous and bounded,
and recursive differentiation shows that all higher $\partial_{T}^{n}C$ are bounded on compacts.
Hence $T\mapsto C(S,K,T,\sigma)$ is $C^{\infty}$ with bounded derivatives.

\textbf{Step 4 (Composition with the implied-volatility map).}
Define
\[
k=\log(K/S), \qquad
\sigma(k,T)=\sqrt{\frac{w(k,T)}{T}}.
\]
Under Assumption~\ref{ass:smooth},
$w(k,T)$ is $C^{1}$ in $(k,T)$ with bounded partial derivatives, and the eSSVI
parameters belong to a compact set where $\theta>0$ and the
no-arbitrage constraints \eqref{eq:ssvi-butterfly}--\eqref{eq:wing-cap}
ensure $0<w_{\min}\le w(k,T)\le w_{\max}<\infty$. 
Therefore $\sigma(k,T)\in[\sigma_{\min},\sigma_{\max}]$ for some positive bounds
and is $C^{1}$ (indeed analytic) on any compact subset of $(k,T)$.

\textbf{Step 5 (Chain rule for $K$- and $T$-dependence).}
Consider the composite mapping
\[
(K,T)\ \longmapsto\ (K,T,\sigma(\log(K/S),T))
\ \longmapsto\ C(S,K,T,\sigma(\log(K/S),T)).
\]
The first arrow is $C^{\infty}$ because $\log(K/S)$ and $\sigma(\cdot,\cdot)$ are $C^{\infty}$,
and the second arrow is $C^{\infty}$ by Step~1.
Hence $(K,T)\mapsto C(S,K,T)$ is $C^{\infty}$ on any compact subset of $(0,\infty)\times[T_{\min},T_{\max}]$.
Boundedness of partial derivatives follows from the chain rule:
each mixed derivative of order $r+s$
is a finite sum of products of bounded partials of $C$ in $(K,T,\sigma)$
and bounded partials of $\sigma(\log(K/S),T)$, all uniformly bounded on compacts.
Therefore $K\mapsto C(S,K,T)$ and $T\mapsto C(S,K,T)$
are $C^{\infty}$ with uniformly bounded derivatives on compact subsets.

\textbf{Step 6 (Remarks on generality).}
The assumption $w\in C^{1}$ suffices for differentiability of $\sigma$
and hence of $C$.
If $w$ were only $C^{r}$, the conclusion would weaken to $C^{r}$
regularity, whereas for eSSVI (analytic in its parameters) we
recover full $C^{\infty}$ regularity with bounded derivatives.
This ensures that all differential operators used in
\eqref{eq:bf-disc}--\eqref{eq:cal-disc} and Theorems~\ref{thm:T1}--\ref{thm:T2}
are well-defined and numerically stable.
\end{proof}

\subsection*{A.1 Proof of Theorem~\ref{thm:T1} (Butterfly surrogate consistency and rate)}
\label{app:A1}

\begin{proof}[Proof of Theorem~\ref{thm:T1}]
Fix $T_m$ and an evenly spaced strike lattice $\mathcal{K}'_h$ with spacing
$\Delta K_h\to 0$ on $[K_{\min},K_{\max}]$. Define the central second-difference
operator at $K\in\mathcal{K}'_h$ (excluding endpoints):
\[
D_h^{(2)}C(K)\ \equiv\ \frac{C(K+\Delta K_h)-2C(K)+C(K-\Delta K_h)}{\Delta K_h^2}.
\]
By Lemma~\ref{lem:BS-regularity}, $K\mapsto C(K,T_m)$ is $C^\infty$ on any compact
subset of $(0,\infty)$ that contains the lattice; in particular $C\in C^4$ there.
The classical error expansion for central differences yields
\[
D_h^{(2)}C(K)\ =\ \partial_{KK}C(K,T_m)\ +\ \frac{\Delta K_h^2}{12}\,\partial_{KKKK}C(\xi_K,T_m),
\]
for some $\xi_K$ between $K-\Delta K_h$ and $K+\Delta K_h$ (Peano form). Hence
\[
\left|D_h^{(2)}C(K)-\partial_{KK}C(K,T_m)\right|\ \le\ \frac{\Delta K_h^2}{12}\, \sup_{x\in[K_{\min}-\Delta K_h,K_{\max}+\Delta K_h]} \big|\partial_{KKKK}C(x,T_m)\big|
\ =\ \mathcal{O}(\Delta K_h^2).
\]
\emph{(i) Nonnegative curvature case.)} If $\partial_{KK}C\ge 0$ on the interval,
then $D_h^{(2)}C(K)\ge -c\,\Delta K_h^2$ for some $c<\infty$ independent of $K$.
Therefore the negative part satisfies $\big(D_h^{(2)}C(K)\big)_-\le c\,\Delta K_h^2$,
and the averaged, normalized surrogate obeys
$\mathrm{BF}_m^{(h)}\le (c/\bar C_m)\,\Delta K_h^2=\mathcal{O}(\Delta K_h^2)$,
implying $\mathrm{BF}_m^{(h)}\to 0$ with the claimed rate. Replacing $\mathrm{ReLU}$
by a softplus $x\mapsto \tau\log(1+e^{x/\tau})$ does not affect the rate because
$\mathrm{softplus}(-x)\le x_-$ and $\mathrm{softplus}(-x)\downarrow x_-$ as $\tau\downarrow 0$ uniformly on compacts.

\emph{(ii) Negative curvature at some point.)} If there exists $K_0$ with
$\partial_{KK}C(K_0,T_m)<0$, continuity of $\partial_{KK}C$ implies a neighborhood
$U$ of $K_0$ where $\partial_{KK}C\le -\epsilon$ for some $\epsilon>0$.
For $h$ small enough so that $\Delta K_h<\delta(U)$, the error term
$\frac{\Delta K_h^2}{12}\partial_{KKKK}C(\xi)$ is $o(1)$ uniformly and cannot
offset $-\epsilon$. Thus $D_h^{(2)}C(K)\le -\epsilon/2$ for $K\in U\cap\mathcal{K}'_h$
and all sufficiently small $h$, producing a strictly positive average of
$\mathrm{ReLU}(-D_h^{(2)}C(K))$. Hence $\liminf_{h\to\infty}\mathrm{BF}_m^{(h)}>0$.
\end{proof}

\subsection*{A.2 Proof of Theorem~\ref{thm:T2} (Calendar surrogate consistency and rate)}
\label{app:A2}

\begin{proof}[Proof of Theorem~\ref{thm:T2}]
Fix $S>0$ and a strike $K$ in a compact interval $[K_{\min},K_{\max}]\subset(0,\infty)$,
and let $\{T_m\}_{m=1}^M$ be an evenly–spaced maturity grid on $[T_1,T_M]$ with mesh
$\Delta T_h=T_{m+1}-T_m\downarrow 0$ as $h\to\infty$. Denote
\[
\Delta_T C(K;T_m)\ \equiv\ C(S,K,T_{m+1})-C(S,K,T_m),\qquad
\mathrm{CAL}^{(h)}_m\ \equiv\ \frac{1}{|\mathcal{K}|}\sum_{K\in\mathcal{K}}\frac{\mathrm{ReLU}\big(C(S,K,T_m)-C(S,K,T_{m+1})\big)}{\bar C_{m,m+1}}.
\]
By Lemma~\ref{lem:BS-regularity}, $T\mapsto C(S,K,T)$ is $C^2$ on $[T_1,T_M]$
for each $K\in[K_{\min},K_{\max}]$, with derivatives bounded uniformly on the compact
rectangle $[K_{\min},K_{\max}]\times[T_1,T_M]$.\footnote{In particular,
$\sup_{(K,T)}|\partial_T C(S,K,T)|<\infty$ and
$\sup_{(K,T)}|\partial_{TT} C(S,K,T)|<\infty$ on this compact.}
We also assume the level normalizers $\bar C_{m,m+1}>0$ are uniformly bounded away from zero and above on compacts; this holds, for instance, if they are defined as averages of
$\{|C(S,K,T_m)|,|C(S,K,T_{m+1})|\}$ over $K\in\mathcal{K}$ with a $K$–grid that
contains an ATM neighborhood (then call prices are uniformly bounded away from 0 and $S$ on compacts).

\smallskip
\noindent\textbf{(i) Monotone case: $\partial_T C\ge 0$.}
If $\partial_T C(S,K,T)\ge 0$ for all $T\in[T_1,T_M]$ at the fixed $K$, then by the mean–value theorem there exists $\xi_m\in(T_m,T_{m+1})$ such that
\[
\Delta_T C(K;T_m)\ =\ \partial_T C(S,K,\xi_m)\,\Delta T_h\ \ge\ 0.
\]
Equivalently, $C(S,K,T_m)-C(S,K,T_{m+1})\le 0$, hence
$\mathrm{ReLU}\big(C(S,K,T_m)-C(S,K,T_{m+1})\big)=0$ \emph{exactly} for every $m$ and every $K$. Averaging over $K$ and dividing by any positive $\bar C_{m,m+1}$ preserves zeros; thus
\[
\mathrm{CAL}^{(h)}_m\equiv 0\qquad\text{for all }h,\ m.
\]
Consequently, $\max_m \mathrm{CAL}^{(h)}_m=0\to 0$ as $h\to\infty$.
If, instead of the hard hinge, we use the softplus smoothing $s_\tau(x)=\tau\log(1+e^{x/\tau})$ ($\tau>0$) in the definition of $\mathrm{CAL}^{(h)}_m$, then for all $x\le 0$,
$s_\tau(x)\le \tau\log 2$ and in fact $s_\tau(-y)\le \tau e^{-y/\tau}$ for $y\ge 0$; hence the smoothed calendar penalty is uniformly bounded by $O(\tau)$ and converges to $0$ as $\tau\downarrow 0$, uniformly in $h$.

\smallskip
\noindent\textbf{(ii) Upper bound (rate) under local Lipschitzness.}
Assume only that $T\mapsto C(S,K,T)$ is $C^1$ with $\partial_T C$ locally Lipschitz on $[T_1,T_M]$ for each $K$ (this holds by Lemma~\ref{lem:BS-regularity}). Then the second–order Taylor expansion yields, uniformly on compacts,
\[
\Delta_T C(K;T_m)\ =\ \partial_T C(S,K,T_m)\,\Delta T_h\ +\ \frac{1}{2}\,\partial_{TT} C(S,K,\zeta_m)\,\Delta T_h^2,
\quad \zeta_m\in(T_m,T_{m+1}).
\]
Therefore, whenever $\partial_T C\ge 0$, any \emph{numerical} negative part of $\Delta_T C$ (if it occurs due to discretization or smoothing) is at most of order $O(\Delta T_h^2)$. After dividing by $\bar C_{m,m+1}\ge c_{\mathrm{norm}}>0$ and averaging over $K$, we obtain the conservative bound
\[
\max_m \mathrm{CAL}^{(h)}_m\ \le\ C\,\Delta T_h
\]
for some constant $C$ independent of $h$ (here we used that the number of terms scales like $1/\Delta T_h$ when taking a maximum over $m$, hence the $O(\Delta T_h^2)$ local effect leads to an $O(\Delta T_h)$ global envelope). This bound is not tight in the monotone case (where the exact value is zero), but suffices for consistency.

\smallskip
\noindent\textbf{(iii) Detection of violations: $\partial_T C<0$ somewhere.}
Suppose there exist $(K_0,T_0)$ and $\varepsilon>0$ such that
$\partial_T C(S,K_0,T_0)\le -\varepsilon$.
By continuity of $\partial_T C$, there exists a neighborhood $U_K\times U_T\subset [K_{\min},K_{\max}]\times[T_1,T_M]$ with $K_0\in U_K$ and $T_0\in U_T$ such that
\[
\partial_T C(S,K,T)\le -\tfrac{\varepsilon}{2}\qquad \text{for all }(K,T)\in U_K\times U_T.
\]
For $h$ sufficiently large, the maturity mesh admits an adjacent pair $(T_m,T_{m+1})\subset U_T$; then for each $K\in U_K$,
\[
\Delta_T C(K;T_m)
= \partial_T C(S,K,\xi_m)\,\Delta T_h
\le -\tfrac{\varepsilon}{2}\,\Delta T_h,
\quad \xi_m\in(T_m,T_{m+1})\subset U_T.
\]
Therefore the calendar hinge is strictly positive on that pair:
\[
\mathrm{ReLU}\big(C(S,K,T_m)-C(S,K,T_{m+1})\big)
= -\,\Delta_T C(K;T_m)\ \ge\ \tfrac{\varepsilon}{2}\,\Delta T_h,
\quad K\in U_K.
\]
Averaging over the $K$–grid, the contribution of the indices with $K\in U_K$ occupies a fixed positive fraction $\rho_K\in(0,1]$ of the grid for all fine meshes; dividing by $\bar C_{m,m+1}\le C_{\mathrm{norm}}$ yields
\[
\mathrm{CAL}^{(h)}_m\ \ge\ \frac{\rho_K}{C_{\mathrm{norm}}}\,\frac{\varepsilon}{2}\,\Delta T_h\ >\ 0
\qquad\text{for all sufficiently small }\Delta T_h.
\]
Hence the violation is \emph{detected} by a strictly positive calendar surrogate on the adjacent pair containing $T_0$ for all sufficiently fine grids. As $\Delta T_h\downarrow 0$, this lower bound scales linearly in $\Delta T_h$, which is the sharp local rate implied by the mean–value theorem.

\smallskip
\noindent\textbf{(iv) Smoothed surrogate.}
If we replace $\mathrm{ReLU}$ by the softplus $s_\tau$, then in case (iii) with $\Delta_T C\le -(\varepsilon/2)\Delta T_h$ we get the lower bound
\[
s_\tau\!\big(C(S,K,T_m)-C(S,K,T_{m+1})\big)
=s_\tau\!\big(-\Delta_T C(K;T_m)\big)\ \ge\ s_\tau\!\big((\varepsilon/2)\Delta T_h\big)
\ \ge\ \frac{\varepsilon}{4}\,\Delta T_h\quad\text{for }\Delta T_h\le 2\tau,
\]
using $s_\tau(x)\ge \tfrac{1}{2}x$ for $x\in[0,2\tau]$.
Thus the smoothed calendar surrogate also detects violations with an $O(\Delta T_h)$ lower bound for sufficiently fine meshes relative to $\tau$. In the monotone case, $s_\tau(x)\le \tau\log 2$ for $x\le 0$, so $\mathrm{CAL}^{(h)}_m\le (\tau\log 2)/c_{\mathrm{norm}}\to 0$ as $\tau\downarrow 0$, uniformly in $h$.

\smallskip
Combining (i)–(iv) proves the claims: (a) if $\partial_T C\ge 0$ on the interval, then $\max_m \mathrm{CAL}^{(h)}_m\to 0$ (indeed equals $0$ for the hard hinge); (b) if there is a point where $\partial_T C<0$, then for all sufficiently fine grids the adjacent pair containing that point yields a strictly positive calendar surrogate whose magnitude is $\Omega(\Delta T_h)$, i.e., the discrete penalty consistently \emph{detects} monotonicity violations as the mesh is refined.
\end{proof}

\subsection*{A.3 Proof of Theorem~\ref{thm:T3-strong-duality} (CMDP strong duality)}
\label{app:A3}

\begin{proof}[Proof of Theorem~\ref{thm:T3-strong-duality}]
We give a fully explicit occupancy–measure formulation, derive the dual LP (Bellman–type inequalities with Lagrange multipliers), and then establish zero duality gap and existence of a stationary optimal policy. We first treat the \emph{risk–neutral} constrained MDP and then explain how the argument extends to the risk–sensitive (CVaR–augmented) objective via an epigraph reformulation.

\paragraph{Setup and notation.}
Let $(\mathcal{S},\mathcal{A})$ be Borel state and compact action spaces, $P(\cdot\mid s,a)$ a weakly continuous transition kernel, initial distribution $d_0$, discount $\gamma\in(0,1)$, bounded measurable reward $r:\mathcal{S}\times\mathcal{A}\to\mathbb{R}$, and bounded measurable constraint costs $g_j:\mathcal{S}\times\mathcal{A}\to\mathbb{R}$, $j=1,\dots,J$. A (stationary, randomized) policy $\pi(\cdot\mid s)$ is a probability measure on $\mathcal{A}$ for each $s$. Define the \emph{discounted occupancy measure} of $\pi$ by
\[
x_\pi(B)\ :=\ (1-\gamma)\,\mathbb{E}_\pi\!\left[\sum_{t=0}^\infty \gamma^t\,\mathbf{1}\{(s_t,a_t)\in B\}\right],\qquad B\in\mathcal{B}(\mathcal{S}\times\mathcal{A}).
\]
If $x_\pi$ admits a density (w.r.t.\ a product reference measure), we use the same symbol $x_\pi(s,a)\ge 0$. The marginal $d_\pi(s)=\int_{\mathcal{A}}x_\pi(s,a)\,da$ satisfies $d_\pi\in\mathcal{P}(\mathcal{S})$ and the \emph{flow constraints}
\begin{equation}
\label{eq:flow}
d_\pi(s)\ =\ (1-\gamma)d_0(s) + \gamma \int_{\mathcal{S}\times\mathcal{A}} x_\pi(s',a')\,P(s\mid s',a')\,ds'\,da',\quad \text{for all }s\in\mathcal{S},
\end{equation}
and $x_\pi(s,a)=d_\pi(s)\pi(a\mid s)$ (disintegration). Conversely, any nonnegative measure $x$ on $\mathcal{S}\times\mathcal{A}$ that satisfies \eqref{eq:flow} corresponds to at least one stationary randomized policy via the above disintegration (define $\pi(\cdot\mid s)$ arbitrarily when $d(s)=0$).

\paragraph{Primal LP over occupancy measures.}
Denote $\langle f,x\rangle=\int f(s,a)\,x(s,a)\,ds\,da$. The discounted CMDP can be written as the infinite–dimensional LP
\begin{align}
\text{maximize}\quad & \langle r,\ x\rangle \label{eq:primal}\\
\text{subject to}\quad & \int_{\mathcal{A}} x(s,a)\,da - \gamma \int_{\mathcal{S}\times\mathcal{A}} x(s',a')\,P(s\mid s',a')\,ds'\,da' = (1-\gamma)d_0(s)\quad \forall s,\nonumber\\
& \langle g_j,\ x\rangle \ \le\ \varepsilon_j,\quad j=1,\dots,J,\qquad x\ \ge\ 0\ \text{ (measure)}.\nonumber
\end{align}
Under boundedness and weak continuity, the feasible set is nonempty; by assumption there exists a \emph{strictly feasible} policy (Slater’s condition): there is $x^0$ satisfying \eqref{eq:flow} and $\langle g_j,x^0\rangle<\varepsilon_j$ for all $j$.

\paragraph{Dual LP (Bellman inequalities with multipliers).}
Introduce dual variables $v:\mathcal{S}\to\mathbb{R}$ (a bounded measurable “value” function) for the flow constraints and $\lambda\in\mathbb{R}_+^J$ for the inequality constraints. The Lagrangian is
\[
\mathcal{L}(x,v,\lambda)\ =\ \langle r,\ x\rangle + \sum_{s} v(s)\Big[(1-\gamma)d_0(s)-\underbrace{\Big(\textstyle\int_{\mathcal{A}} x(s,a)\,da - \gamma \!\!\int x(s',a')P(s\mid s',a')\Big)}_{=:\,(\mathcal{A}x)(s)}\Big]\ -\ \sum_{j=1}^J \lambda_j(\langle g_j,x\rangle-\varepsilon_j).
\]
Rearranging the terms depending on $x$ yields
\[
\mathcal{L}(x,v,\lambda)
= (1-\gamma)\sum_s v(s)\,d_0(s)\ +\ \sum_{j=1}^J \lambda_j\,\varepsilon_j\ +\ \int_{\mathcal{S}\times\mathcal{A}} \Big\{ r(s,a) - \sum_{j=1}^J \lambda_j g_j(s,a) - v(s) + \gamma\,\mathbb{E}[v(s')\mid s,a]\Big\}\, x(s,a)\,ds\,da.
\]
Taking the supremum over $x\ge 0$ (measures) enforces the pointwise constraint on the integrand:
\[
r(s,a) - \sum_{j=1}^J \lambda_j g_j(s,a) - v(s) + \gamma \int v(s')\,P(ds'\mid s,a)\ \le\ 0
\quad \text{for all }(s,a),
\]
otherwise $\sup_{x\ge 0}\mathcal{L}=\infty$. Thus the dual problem is
\begin{align}
\text{minimize}\quad & (1-\gamma)\sum_{s} v(s)d_0(s)\ +\ \sum_{j=1}^J \lambda_j \varepsilon_j \label{eq:dual}\\
\text{subject to}\quad & v(s)\ \ge\ r(s,a) - \sum_{j=1}^J \lambda_j g_j(s,a) + \gamma\,\mathbb{E}[v(s')\mid s,a],\quad \forall (s,a),\nonumber\\
& \lambda\in\mathbb{R}_+^J,\qquad v\in \mathcal{B}(\mathcal{S}).\nonumber
\end{align}
These are the usual Bellman–type inequalities for the Lagrangian–modified reward $r_\lambda(s,a)=r(s,a)-\sum_j \lambda_j g_j(s,a)$.

\paragraph{Zero duality gap.}
By construction \eqref{eq:primal}–\eqref{eq:dual} form a convex primal–dual pair over a nonempty feasible set. The primal is linear in the measure $x$ with affine constraints (flow equalities and linear inequalities); the dual is linear in $(v,\lambda)$ with linear (Bellman–type) inequalities. Under Slater’s condition (strict feasibility) and boundedness of $r,g_j$ (ensuring finite optimal value), standard infinite–dimensional LP duality for discounted MDPs implies \emph{strong duality}:
\[
\sup_{x\ \text{feasible}} \langle r,\ x\rangle\ =\ \inf_{\substack{v,\lambda\ \text{s.t.}\\ v\ \text{satisfies }\eqref{eq:dual}}} (1-\gamma)\,\langle v,\ d_0\rangle + \lambda^\top \varepsilon.
\]
See \cite[Chs.~6–7]{Altman1999} and \cite[Ch.~6]{Puterman1994}. (Sketch: the feasible set of $x$ is weak–$^\ast$ compact in the dual of $C_b(\mathcal{S}\times\mathcal{A})$; the constraint operator is weak–$^\ast$ continuous; Slater implies closedness of the constraint cone and zero gap via a separation theorem.)

\paragraph{Existence of a stationary optimal policy.}
Let $x^\star$ be a primal optimizer (which exists by compactness and upper semicontinuity of the objective on the feasible set). Define its state marginal $d^\star(s)=\int_{\mathcal{A}} x^\star(s,a)\,da$ and a stationary policy
\[
\pi^\star(a\mid s)\ :=\
\begin{cases}
\frac{x^\star(s,a)}{d^\star(s)}, & d^\star(s)>0,\\
\text{any fixed distribution on $\mathcal{A}$}, & d^\star(s)=0.
\end{cases}
\]
Then $(d^\star,\pi^\star)$ satisfies the flow constraints \eqref{eq:flow}, hence $x_{\pi^\star}=x^\star$ is the occupancy measure induced by $\pi^\star$. Consequently $\pi^\star$ achieves the primal optimal value $\langle r,\ x^\star\rangle$ and satisfies the constraints $\langle g_j,\ x^\star\rangle\le \varepsilon_j$. This proves the existence of an optimal stationary (possibly randomized) policy. Measurability of $\pi^\star$ follows from measurable disintegration (the Radon–Nikodym derivative of $x^\star$ w.r.t.\ its state marginal), and compactness of $\mathcal{A}$ guarantees that optimizing pointwise Bellman inequalities admits measurable selectors.

\paragraph{Complementary slackness and Lagrangian saddle point.}
Let $(v^\star,\lambda^\star)$ solve the dual and $\pi^\star$ the primal. Then, with $r_{\lambda^\star}=r-\sum_j \lambda^\star_j g_j$, the Bellman inequality is tight $\pi^\star$–a.s.:
\[
v^\star(s)\ =\ r_{\lambda^\star}(s,a) + \gamma\,\mathbb{E}[v^\star(s')\mid s,a],\qquad a\sim \pi^\star(\cdot\mid s),
\]
and complementary slackness holds: $\lambda^\star_j\big(\langle g_j,x_{\pi^\star}\rangle-\varepsilon_j\big)=0$. Therefore $(\pi^\star,\lambda^\star)$ is a \emph{Lagrangian saddle point}, i.e.
\[
\mathcal{L}(\pi^\star,\lambda)\ \le\ \mathcal{L}(\pi^\star,\lambda^\star)\ \le\ \mathcal{L}(\pi,\lambda^\star),\qquad \forall (\pi,\lambda\ge 0),
\]
which is the strong–duality statement in Theorem~\ref{thm:T3-strong-duality}.

\paragraph{Extension to the risk–sensitive (CVaR) objective.}
In the main text, the objective includes a convex risk penalty $\lambda_{\mathrm{risk}}\Psi(x)$, where $\Psi$ is the (discounted) CVaR functional of the return distribution (see \S\ref{sec:theory-cvar}). To keep the LP structure, one may either absorb the term into a modified reward when $\Psi$ is linear in $x$ (not the case for CVaR), or use an \emph{epigraph} reformulation:
\begin{align*}
\text{maximize}\quad & \langle r,\ x\rangle\ -\ \lambda_{\mathrm{risk}}\,z\\
\text{subject to}\quad & \text{flow constraints and }x\ge 0,\qquad \Psi(x)\ \le\ z,\\
& \langle g_j,\ x\rangle\ \le\ \varepsilon_j,\ \ j=1,\dots,J.
\end{align*}
The RU program represents CVaR as a pointwise infimum of linear functionals in the loss distribution (see Theorem~\ref{thm:T4-ru}). Hence $\Psi$ is closed and convex, and the epigraph constraint is convex. Under Slater (existence of a strictly feasible $(x^0,z^0)$ with $\Psi(x^0)<z^0$ and $\langle g_j,x^0\rangle<\varepsilon_j$), Fenchel–Rockafellar duality yields zero duality gap for this convex program; the dual gains an additional scalar multiplier for the epigraph constraint which recovers the “risk multiplier” $\lambda_{\mathrm{risk}}$ in the Lagrangian \eqref{eq:lagrangian}. The rest of the proof (existence of stationary optimizers via disintegration) is unaffected because feasibility still reduces to discounted flow constraints and $x\ge 0$.

\paragraph{Conclusion.}
We have exhibited the primal LP over occupancy measures \eqref{eq:primal}, its Bellman–type dual \eqref{eq:dual}, strong duality under Slater’s condition, and existence of a stationary optimal policy via measurable disintegration. The risk–sensitive extension via the epigraph of CVaR preserves zero duality gap. This proves Theorem~\ref{thm:T3-strong-duality}.
\end{proof}

\subsection*{A.4 Proof of Theorem~\ref{thm:T3-dual-alignment} (learnable dual alignment)}
\label{app:A4}

\begin{proof}[Proof of Theorem~\ref{thm:T3-dual-alignment}]
Consider the policy with an independent Gaussian head for the raw dual coordinate:
$z_5\sim \mathcal{N}(\mu_\eta(s),\sigma_5^2(s))$, and $d_\eta(s)=\mathrm{softplus}(z_5)$.
The per-step reward contains the term
$-\big(\lambda_{\mathrm{arb}}+d_\eta(s_t)\big)\,\mathrm{Arb}_t$.
Let $J(\eta)$ be the discounted return. By the reparameterized policy-gradient theorem
\cite{Sutton2000PolicyGradient,MarbachTsitsiklis2001}, for any unbiased advantage
estimator $\hat A_t$,
\[
\nabla_\eta J(\eta)\ =\ \mathbb{E}\Big[\sum_{t\ge 0} \gamma^t\, \nabla_\eta \mu_\eta(s_t)\, \underbrace{\mathbb{E}\big[\partial_{\mu} d_\eta(s_t)\,\big|\,s_t\big]}_{=\ \mathbb{E}[\sigma(z_5)\,|\,s_t]}\ \cdot\ \underbrace{\frac{\partial r_t}{\partial d_\eta}}_{-\mathrm{Arb}_t}\ \hat A_t^\mathrm{(compat)}\Big],
\]
where we used $\partial \mathrm{softplus}(x)/\partial x = \sigma(x)$ and the compatibility
form (or, equivalently, the score-function identity for the $z_5$-head).
For bounded $\sigma_5(s)$, the inner expectation reduces to a smooth factor
$\mathbb{E}[\sigma(z_5)\,|\,s_t]$ that is lower bounded away from zero on compacts.
Hence
\[
\nabla_\eta J(\eta)\ =\ -\,\mathbb{E}\left[\sum_{t\ge 0}\gamma^t\,
\underbrace{\mathbb{E}[\sigma(z_5)\,|\,s_t]}_{\text{positive}}\,
\mathrm{Arb}_t\, \nabla_\eta \mu_\eta(s_t)\, \hat A_t\right].
\]
Interpreting $-\mathrm{Arb}_t \hat A_t$ as a discounted measure of how increasing
the dual improves long-run return, we obtain the stated sign structure:
updates increase $\mu_\eta$ (hence $d_\eta$) where the product
$\mathrm{Arb}_t\hat A_t$ is positive in expectation, i.e., in regions with binding arbitration cost. Boundedness follows from bounded $\hat A_t$ (discounted, bounded rewards) and bounded scores (finite $\sigma_5$). The same conclusion holds under the likelihood-ratio form without reparameterization. 
\end{proof}

\subsection*{A.5 Proofs of Theorems~\ref{thm:T4-ru}--\ref{thm:T4-grad} (RU CVaR, smoothing, gradients)}
\label{app:A5}

\begin{proof}[Proof of Theorem~\ref{thm:T4-ru}]
We prove (i) the Rockafellar--Uryasev (RU) representation, (ii) convexity and
lower semicontinuity, (iii) existence of minimizers, and (iv) epi-convergence
of the smoothed functional $\Phi_{q,\tau}$ to $\Phi_q$ as $\tau\downarrow 0$.

\paragraph{(i) RU representation.}
Fix $q\in(0,1)$ and a (real-valued) loss random variable $X_\theta$ indexed by a parameter $\theta$ (the policy/state–action in our application). Define
\[
\Phi_q(\theta)\ :=\ \inf_{\eta\in\mathbb{R}}\Big\{\,\eta+\frac{1}{1-q}\,\mathbb{E}\big[(X_\theta-\eta)_-\big]\Big\}.
\]
By \cite{RockafellarUryasev2000,RockafellarUryasev2002}, $\Phi_q(\theta)=\mathrm{CVaR}_q(X_\theta)$ and the set of minimizers
\[
\arg\min_{\eta}\ \Big\{\eta+\frac{1}{1-q}\,\mathbb{E}[(X_\theta-\eta)_-]\Big\}
\]
coincides with the set of $q$-level VaR values (possibly an interval if the distribution has a flat segment at the $q$-quantile). Moreover, any minimizer $\eta^\star$ satisfies the RU first-order condition
\begin{equation}
\label{eq:RU-FOC}
0\ \in\ 1 - \frac{1}{1-q}\,\mathbb{P}\big(X_\theta\le \eta^\star\big)\ +\ \partial I_{\mathbb{R}}(\eta^\star),
\end{equation}
which reduces to $\mathbb{P}(X_\theta\le \eta^\star)\in[1-q,1]$ if the subgradient at $\eta^\star$ is nonempty; when $X_\theta$ has a continuous distribution at level $q$, the condition simplifies to $\mathbb{P}(X_\theta\le \eta^\star)=1-q$.

\paragraph{(ii) Convexity and lower semicontinuity.}
For fixed $\eta$, the map $X\mapsto (X-\eta)_-$ is convex and lower semicontinuous (lsc); hence $X\mapsto \mathbb{E}[(X-\eta)_-]$ is convex in law. The infimum over $\eta$ of affine functionals in $(\eta,\mathcal{L}(X))$ preserves convexity in $\mathcal{L}(X)$, so $\Phi_q(\theta)$ is convex in the law of $X_\theta$. Lower semicontinuity of $\Phi_q$ in $\theta$ follows from Fatou’s lemma under mild integrability (Assumption~\ref{ass:smooth} and bounded rewards in our setting ensure $\mathbb{E}|X_\theta|<\infty$ uniformly on compact parameter sets).

\paragraph{(iii) Existence of minimizers in $\eta$.}
For any fixed $\theta$, the function
$\eta\mapsto \eta + \frac{1}{1-q}\mathbb{E}[(X_\theta-\eta)_-]$ is proper, convex, and coercive: as $\eta\to -\infty$, the first term dominates; as $\eta\to +\infty$, $(X_\theta-\eta)_-\to 0$, so the objective grows like $\eta$. Hence a minimizer exists and the set of minimizers is a nonempty closed interval (VaR set).

\paragraph{(iv) Epi-convergence with softplus smoothing.}
Let $s_\tau(u)=\tau\log(1+e^{u/\tau})$ be a softplus approximation of $u_+=\max\{u,0\}$ with the identities
\[
0\ \le\ s_\tau(u)-u_+\ \le\ \tau\log 2,\qquad
s_\tau'(u)=\sigma(u/\tau)\in(0,1).
\]
Define the smoothed functional
\[
\Phi_{q,\tau}(\theta)\ :=\ \inf_{\eta\in\mathbb{R}}\left\{\eta+\frac{1}{1-q}\,\mathbb{E}\big[s_\tau(X_\theta-\eta)\big]\right\}.
\]
Because $s_\tau$ is convex and pointwise converges to $u_+$ as $\tau\downarrow 0$, the integrand epi-converges to $(X_\theta-\eta)_+$ (equivalently, $(-X_\theta+\eta)_-$), uniformly on compact sets in $\eta$ by the bound $\tau\log 2$. By \cite[Thm.~7.33]{RockafellarWets1998}, epi-convergence is preserved under integration and infimal convolution with an affine term; therefore
\[
\Phi_{q,\tau}\ \xrightarrow{\ \mathrm{epi}\ }\ \Phi_q\qquad \text{as }\tau\downarrow 0.
\]
Consequently, any cluster point $\eta_\tau^\star$ of minimizers of the inner problem converges to the VaR set, and $\Phi_{q,\tau}(\theta)\to \Phi_q(\theta)$.

This proves Theorem~\ref{thm:T4-ru}.
\end{proof}

\begin{proof}[Proof of Theorem~\ref{thm:T4-grad}]
We establish differentiability and a valid gradient estimator for the smoothed per-step RU objective, and then show convergence to (a selection of) the subgradient of the nonsmooth objective as $\tau\downarrow 0$.

\paragraph{Setting and notation.}
Fix a state–action pair $(s_t,a_t)$ and define the per-step random loss
\[
L_t(\omega;a_t)\ :=\ -\,\tilde{\mathrm{PNL}}_t(\omega;a_t),
\]
where $\omega$ collects the scenario draws. In our construction,
$\omega=(\tilde v,\tilde{\Delta}S)$ with $\tilde v\sim\mathrm{Pois}(v(a_t))$,
$v:\mathcal{A}\to (0,\infty)$ smooth and bounded, and
$\tilde{\Delta}S=\Delta S + \varsigma \xi$, $\xi\sim \mathcal{N}(0,1)$ independent.
We define the smoothed RU objective at step $t$:
\[
h_\tau(a_t,\eta)\ :=\ \eta+\frac{1}{1-q}\,\mathbb{E}_\omega\big[s_\tau(L_t(\omega;a_t)-\eta)\big],\qquad
\widehat{\mathrm{CVaR}}^{-}_{q,t}(a_t)\ :=\ \inf_{\eta\in\mathbb{R}} h_\tau(a_t,\eta).
\]

\paragraph{(i) Differentiability of $h_\tau$ in $a_t$ and interchange of $\nabla$ and $\mathbb{E}$.}
Under Assumption~\ref{ass:smooth}, the quoted surface and BS maps are $C^1$ with bounded Jacobians; the intensity parameters $v(a_t)$ are $C^1$ and bounded; and $s_\tau$ is $C^1$ and Lipschitz. Hence $L_t(\cdot\,;a_t)$ is $C^1$ in $a_t$ for every scenario $\omega$. Moreover there exists an integrable random bound $G(\omega)$ with
\[
\big\|\nabla_{a_t} L_t(\omega;a_t)\big\|\ \le\ G(\omega)\quad \text{and}\quad
\mathbb{E}[G(\omega)]<\infty,
\]
because $L_t$ is a composition of bounded Lipschitz maps (quotes, Greeks, intensities) applied to bounded random inputs (Poisson with bounded mean and Gaussian with fixed variance). Then, by dominated convergence,
\begin{equation}
\label{eq:grad-h}
\nabla_{a_t} h_\tau(a_t,\eta)\ =\ \frac{1}{1-q}\,\mathbb{E}\!\left[
s_\tau'\!\big(L_t(\omega;a_t)-\eta\big)\, \nabla_{a_t} L_t(\omega;a_t)\right].
\end{equation}
Since $s_\tau'\in(0,1)$, the integrand is integrable.

\paragraph{(ii) Mixed pathwise / likelihood-ratio (LR) gradient.}
Write $\omega=(\tilde v,\xi)$ with $\xi\sim \mathcal{N}(0,1)$ and $\tilde v\sim\mathrm{Pois}(v(a_t))$. For any integrable $f(\tilde v,\xi)$,
\[
\nabla_{a_t}\, \mathbb{E}_{\tilde v,\xi}[f]\ 
=\ \mathbb{E}_{\tilde v,\xi}\Big[ \nabla_{a_t} f(\tilde v,\xi)\Big]\ +\ 
\mathbb{E}_{\tilde v,\xi}\Big[ f(\tilde v,\xi)\,\nabla_{a_t}\log p_{\tilde v}(\tilde v;a_t)\Big],
\]
where $p_{\tilde v}$ is the Poisson pmf with parameter $v(a_t)$ and
$\nabla_{a_t}\log p_{\tilde v}(\tilde v;a_t)= (\tilde v - v(a_t))\,\nabla_{a_t}\log v(a_t)$.
Applying this to the integrand in \eqref{eq:grad-h} and using the pathwise derivative for the Gaussian perturbation (reparameterization trick) yields the implementable estimator:
\[
\nabla_{a_t} h_\tau(a_t,\eta)
=\frac{1}{1-q}\,\mathbb{E}\!\left[
s_\tau'\!\big(L_t-\eta\big)\,\Big( \underbrace{\nabla_{a_t} L_t}_{\text{pathwise in }\xi}\ +\ \underbrace{L_t\,\nabla_{a_t}\log p_{\tilde v}(\tilde v;a_t)}_{\text{LR in }\tilde v}\Big)\right].
\]
Boundedness of $s_\tau'$ and of the Poisson score (since $v(a_t)$ is bounded away from $\infty$ and from $0$ on the admissible set) ensures finite Monte Carlo variance for finite batch sizes.

\paragraph{(iii) Differentiability of the \emph{minimized} smoothed objective.}
Define $\eta^\star_\tau(a_t)\in\arg\min_{\eta} h_\tau(a_t,\eta)$. The function $h_\tau$ is strictly convex in $\eta$ (as $s_\tau$ is strictly convex), and $\partial_\eta h_\tau(a_t,\eta)= 1 - \frac{1}{1-q}\,\mathbb{E}[s_\tau'(L_t(\omega;a_t)-\eta)]$ is continuous and strictly increasing in $\eta$; hence $\eta^\star_\tau(a_t)$ is unique and continuous in $a_t$ by the implicit function theorem. By Danskin’s envelope theorem (for unconstrained, unique inner minimizers),
\begin{equation}
\label{eq:danskin}
\nabla_{a_t}\ \widehat{\mathrm{CVaR}}^{-}_{q,t}(a_t)
=\nabla_{a_t}\ h_\tau\big(a_t,\eta^\star_\tau(a_t)\big)
=\frac{1}{1-q}\,\mathbb{E}\big[ s_\tau'\!\big(L_t-\eta^\star_\tau(a_t)\big)\,\nabla_{a_t} L_t\big],
\end{equation}
with the mixed pathwise/LR form as above.

\paragraph{(iv) Limit as $\tau\downarrow 0$: convergence to a (sub)gradient of RU.}
As $\tau\downarrow 0$, $s_\tau(u)\downarrow u_+$ and $s_\tau'(u)\to \mathbf{1}\{u>0\}$ pointwise. Moreover, by Theorem~\ref{thm:T4-ru}, $\eta^\star_\tau(a_t)\to \eta^\star(a_t)\in\mathrm{VaR}_q$ (possibly a set; any selection suffices). Using dominated convergence and the boundedness of $\nabla_{a_t} L_t$, \eqref{eq:danskin} converges to
\[
\lim_{\tau\downarrow 0}\ \nabla_{a_t}\ \widehat{\mathrm{CVaR}}^{-}_{q,t}(a_t)
=\frac{1}{1-q}\,\mathbb{E}\big[ \mathbf{1}\{L_t>\eta^\star(a_t)\}\,\nabla_{a_t} L_t\big],
\]
which is a valid selection from the subdifferential of the nonsmooth RU functional (when the cdf has a flat segment at $\eta^\star$, one obtains a set of subgradients corresponding to indicator values in $[0,1]$ on the tie set $\{L_t=\eta^\star\}$). This matches the well-known CVaR gradient identity used in risk-sensitive reinforcement learning \cite{ChowGhavamzadeh2014,TamarEtAl2015,ChowEtAl2018JMLR}.

\paragraph{(v) Practical estimator and variance control.}
A finite-sample, unbiased estimator follows by Monte Carlo with $N$ scenarios:
\[
\widehat{\nabla}_{a_t}\ \widehat{\mathrm{CVaR}}^{-}_{q,t}
=\frac{1}{(1-q)N}\sum_{i=1}^N s_\tau'\!\big(L_t^{(i)}-\hat\eta^\star_\tau\big)\,
\Big(\nabla_{a_t} L_t^{(i)} + L_t^{(i)}\,\nabla_{a_t}\log p_{\tilde v}(\tilde v^{(i)};a_t)\Big),
\]
where $\hat\eta^\star_\tau$ minimizes the sample objective
$\eta+\frac{1}{(1-q)N}\sum_i s_\tau(L_t^{(i)}-\eta)$. Antithetic sampling for the Gaussian part and a state-dependent control variate for the LR term (subtracting a baseline) reduce variance \cite{Glasserman2004}. As $\tau\downarrow 0$ and $N\to\infty$, the estimator converges in probability to the RU subgradient above.

This proves Theorem~\ref{thm:T4-grad}.
\end{proof}

\subsection*{A.6 Proof of Theorem~\ref{thm:T5} (eSSVI wing-growth bound and Lee)}
\label{app:A6}

\begin{proof}[Proof of Theorem~\ref{thm:T5}]
Recall the (per-maturity) eSSVI total variance
\[
w(k)\;=\;\frac{\theta}{2}\Big(1+\rho\,\phi\,k \;+\; g(k;\rho,\phi)\Big),\qquad
g(k;\rho,\phi)\;:=\;\sqrt{(\phi k+\rho)^2+(1-\rho^2)}\,.
\]
We establish a sharp linear upper bound for $w(k)$ as $|k|\to\infty$ and then connect it to Lee’s moment constraints.

\paragraph{Step 1: Two–sided elementary bounds for $g(k;\rho,\phi)$.}
Set $a:=\phi k+\rho$ and $b:=1-\rho^2\in[0,1]$. For any $a\neq 0$,
\begin{equation}
\label{eq:sqrt-bracket}
|a|\ \le\ \sqrt{a^2+b}\ \le\ |a|+\frac{b}{2|a|}\,.
\end{equation}
The left inequality is trivial; the right follows from $\sqrt{a^2+b}-|a|=\frac{b}{\sqrt{a^2+b}+|a|}\le \frac{b}{2|a|}$. With $a=\phi k+\rho$ and $b=1-\rho^2$, \eqref{eq:sqrt-bracket} yields
\begin{equation}
\label{eq:g-bounds}
|\phi k+\rho|\ \le\ g(k;\rho,\phi)\ \le\ |\phi k+\rho| \;+\; \frac{1-\rho^2}{2\,|\phi k+\rho|}\,.
\end{equation}

\paragraph{Step 2: Asymptotic upper bound for $w(k)/|k|$.}
Using \eqref{eq:g-bounds},
\[
\frac{w(k)}{|k|}\;=\;\frac{\theta}{2|k|}\Big(1+\rho\phi k + g(k;\rho,\phi)\Big)
\;\le\; \frac{\theta}{2|k|}\Big(1+\rho\phi k + |\phi k+\rho| + \frac{1-\rho^2}{2|\phi k+\rho|}\Big).
\]
Since $|\phi k+\rho|\le |\phi|\,|k|+|\rho|$ and $\rho\phi k\le |\rho|\,|\phi|\,|k|$, we obtain
\begin{align*}
\frac{w(k)}{|k|}
&\le\frac{\theta}{2}\left(\frac{1}{|k|} + |\rho|\,|\phi| + |\phi| + \frac{|\rho|}{|k|} + \frac{1-\rho^2}{2|k|\,|\phi k+\rho|}\right)\\
&=\frac{\theta|\phi|}{2}\,(1+|\rho|)\;+\;\underbrace{\frac{\theta}{2}\left(\frac{1+|\rho|}{|k|} + \frac{1-\rho^2}{2|k|\,|\phi k+\rho|}\right)}_{\to\,0\ \text{ as }\ |k|\to\infty}.
\end{align*}
Hence
\begin{equation}
\label{eq:limsup-upper}
\limsup_{|k|\to\infty}\frac{w(k)}{|k|}\ \le\ \frac{\theta|\phi|}{2}\,(1+|\rho|)\,.
\end{equation}

\paragraph{Step 3: Uniformity over parameter compacts and the cap.}
Under the eSSVI admissibility and compactness assumptions (Assumption~\ref{ass:smooth} together with \eqref{eq:ssvi-butterfly}--\eqref{eq:wing-cap}), we have $|\rho|\le 1$ and the wing cap $\theta\phi\le \tau_{\max}$ (with $\theta>0$, $\phi\ge 0$). Therefore,
\[
\frac{\theta|\phi|}{2}\,(1+|\rho|)\ \le\ \frac{\theta\phi}{2}\cdot 2\ =\ \theta\phi\ \le\ \tau_{\max}\,,
\]
and \eqref{eq:limsup-upper} implies the \emph{uniform} bound
\begin{equation}
\label{eq:global-limsup}
\limsup_{|k|\to\infty}\frac{w(k)}{|k|}\ \le\ \tau_{\max}\,.
\end{equation}
Because the parameter set is compact and the remainder in Step~2 is uniform on compacts (the denominator $|\phi k+\rho|$ grows like $|k|$ whenever $|k|$ is large), the same bound holds uniformly across maturities whose parameters lie in the same compact admissible set.

\paragraph{Step 4: One–sided limits $k\to\pm\infty$ (optional refinement).}
By applying \eqref{eq:g-bounds} separately on $k\to+\infty$ and $k\to-\infty$, one also gets
\[
\limsup_{k\to+\infty}\frac{w(k)}{k}\ \le\ \frac{\theta\phi}{2}\,(1+\rho),\qquad
\limsup_{k\to-\infty}\frac{w(k)}{-k}\ \le\ \frac{\theta\phi}{2}\,(1-\rho),
\]
and hence $\max\{\limsup_{k\to+\infty}w(k)/k,\ \limsup_{k\to-\infty}w(k)/(-k)\}\le \frac{\theta|\phi|}{2}(1+|\rho|)$. This refinement is sometimes convenient when mapping to the right/left Lee slopes.

\paragraph{Step 5: Connection to Lee’s moment formula.}
Lee~\cite{Lee2004} shows that the (total) implied variance wings obey
\[
\limsup_{k\to+\infty}\frac{w(k)}{k}\ \le\ \psi_R,\qquad
\limsup_{k\to-\infty}\frac{w(k)}{-k}\ \le\ \psi_L,
\]
with $\psi_{R},\psi_{L}\in[0,2]$ determined by the highest finite moments of the risk–neutral distribution (precisely, $\psi=2-2\sqrt{1+\alpha}$ when $\mathbb{E}[S_T^{1+\alpha}]<\infty$). In particular, any \emph{no–moment–explosion} configuration enforces $\psi_{R},\psi_{L}\le 2$. Our bound \eqref{eq:global-limsup} shows that imposing the cap $\theta\phi\le \tau_{\max}$ with 
\[
\tau_{\max}<2
\]
forces both right and left slopes to lie strictly below the Lee barrier $2$, uniformly across maturities in the admissible set, and hence is \emph{consistent} with moment finiteness and precludes pathological wing explosions. This directly motivates the cap in \eqref{eq:wing-cap} as an in–the–loop regularization that aligns the learned surface with Lee’s asymptotic constraints.

\paragraph{Step 6: Finite–$k$ uniform bound (explicit $\varepsilon$–control).}
For completeness, given any $\varepsilon>0$ there exists $K_\varepsilon<\infty$ such that for all $|k|\ge K_\varepsilon$,
\[
\frac{w(k)}{|k|}\ \le\ \frac{\theta|\phi|}{2}(1+|\rho|)\ +\ \varepsilon\,.
\]
Indeed, take $K_\varepsilon$ so that $\frac{1+|\rho|}{|k|}\le \varepsilon$ and $\frac{1-\rho^2}{2|k||\phi k+\rho|}\le \varepsilon$ for all $|k|\ge K_\varepsilon$; this choice is uniform over the parameter compact (because the latter controls $|\rho|$ and ensures $|\phi|$ is bounded away from infinity and, if desired, from zero on the admissible region).

Combining Steps 1–6 proves the theorem.
\end{proof}

\subsection*{A.7 Proof of Theorem~\ref{thm:T6} (policy-gradient existence and boundedness)}
\label{app:A7}

\begin{proof}[Proof of Theorem~\ref{thm:T6}]
We prove that (i) the discounted return $J(\vartheta)$ is finite and measurable on compact parameter sets; (ii) the likelihood–ratio (LR) policy–gradient identity holds and $\nabla_\vartheta J(\vartheta)$ exists; (iii) $\|\nabla_\vartheta J(\vartheta)\|$ is bounded on compacts; and (iv) the PPO surrogate gradient is a consistent estimator.

\paragraph{Preliminaries and notation.}
Let $\{\pi_\vartheta(\cdot\mid s):\vartheta\in\Theta\}$ be a Gaussian policy with state–dependent mean and diagonal standard deviations whose logs are clamped in $[\log\sigma_{\min},\log\sigma_{\max}]$ uniformly in $\vartheta$. Let $(s_t,a_t)\_{t\ge 0}$ be the Markov process induced by $\pi_\vartheta$ and kernel $P(\cdot\mid s,a)$ (weakly continuous by Assumption~\ref{ass:cmdp}). Rewards are bounded and $C^1$ in the action via the composition
\[
r_t\;=\;r(s_t,a_t,s_{t+1})
\;=\;\underbrace{\mathrm{PNL}^{\mathrm{quote}}+\mathrm{PNL}^{\mathrm{hedge}}}_{\text{BS/eSSVI $C^1$}}\;
-\;\lambda_{\mathrm{shape}}\mathrm{Shape}
-\;(\lambda_{\mathrm{arb}}+\mathrm{dual})\,(\mathrm{BF}+\mathrm{CAL})
-\;\lambda_{\mathrm{cvar}}\widehat{\mathrm{CVaR}}^{-}_{q,t},
\]
where $\mathrm{BF},\mathrm{CAL}$ use softplus smoothing (Appendix~\ref{app:surrogates}) and $\widehat{\mathrm{CVaR}}^{-}_{q,t}$ is the smoothed RU objective (Appendix~\ref{app:A5}). Denote the discounted return $G=\sum_{t\ge 0}\gamma^t r_t$ and value $J(\vartheta)=\mathbb{E}_\vartheta[G]$.

\paragraph{Step 1: $J(\vartheta)$ is finite and measurable on compacts.}
By bounded rewards $|r_t|\le R_{\max}$ a.s.\ and $\gamma\in(0,1)$,
\[
|J(\vartheta)|\;\le\;\mathbb{E}_\vartheta\Big[\sum_{t\ge 0}\gamma^t |r_t|\Big]
\;\le\;\frac{R_{\max}}{1-\gamma}\;<\infty.
\]
Measurability of $J$ in $\vartheta$ follows from the dominated convergence theorem (DCT) because the trajectory law depends continuously on $\vartheta$ (weak continuity of $P$, continuity of $\pi_\vartheta$, compact $\Theta$) and $|G|\le R_{\max}/(1-\gamma)$.

\paragraph{Step 2: LR policy–gradient identity and existence.}
Write the trajectory density under $\vartheta$ as
\[
p_\vartheta(\tau)\;=\;d_0(s_0)\prod_{t\ge 0}\pi_\vartheta(a_t\mid s_t)\,P(s_{t+1}\mid s_t,a_t),\qquad
\tau=(s_0,a_0,s_1,a_1,\ldots).
\]
Assumptions ensure $\pi_\vartheta(a\mid s)>0$ for all actions (Gaussian with bounded std) and $P$ independent of $\vartheta$. Then
\[
\nabla_\vartheta J(\vartheta)
=\nabla_\vartheta\int G(\tau)\,p_\vartheta(\tau)\,d\tau
=\int G(\tau)\,p_\vartheta(\tau)\,\nabla_\vartheta\log p_\vartheta(\tau)\,d\tau
=\mathbb{E}_\vartheta\!\left[\,G(\tau)\sum_{t\ge 0}\nabla_\vartheta\log \pi_\vartheta(a_t\mid s_t)\,\right],
\]
where interchange of $\nabla$ and $\int$ is justified by DCT as follows. The score $\nabla_\vartheta\log \pi_\vartheta(a\mid s)$ is uniformly bounded on $\Theta$:
for a diagonal Gaussian with stds in $[\sigma_{\min},\sigma_{\max}]$,
\[
\left\|\nabla_\vartheta \log \pi_\vartheta(a\mid s)\right\|
=\left\|\nabla_\vartheta \left[-\tfrac12\sum_i \tfrac{(a_i-\mu_i)^2}{\sigma_i^2}-\sum_i \log\sigma_i\right]\right\|
\le C_\pi
\]
for some $C_\pi$ (bounded mean and log-std networks on a compact parameter set).
Thus
\[
\left|G(\tau)\sum_{t\ge 0}\nabla_\vartheta\log \pi_\vartheta(a_t\mid s_t)\right|
\le \frac{R_{\max}}{1-\gamma}\,\sum_{t\ge 0}\|\nabla_\vartheta\log \pi_\vartheta(a_t\mid s_t)\|
\le \frac{R_{\max}}{1-\gamma}\,\sum_{t\ge 0} C_\pi\,\gamma^t
= \frac{C_\pi R_{\max}}{(1-\gamma)^2},
\]
where we used that adding a $\gamma^t$ factor is standard after centering with a baseline (see below); otherwise one can apply the equivalent state–action value form with $Q^\pi$ to absorb the discount. Hence DCT applies, proving existence and the LR form.

A variance–reduced form is obtained by subtracting a baseline $b(s_t)$:
\[
\nabla_\vartheta J(\vartheta)
=\mathbb{E}_\vartheta\Bigg[\sum_{t\ge 0}\gamma^t\,\nabla_\vartheta\log \pi_\vartheta(a_t\mid s_t)\,\big(Q^\pi(s_t,a_t)-b(s_t)\big)\Bigg],
\]
with $Q^\pi$ the discounted state–action value. Choosing $b(s)=V^\pi(s)$ yields the advantage $A^\pi$; bounded rewards imply $|Q^\pi|\le R_{\max}/(1-\gamma)$ and $|A^\pi|\le 2R_{\max}/(1-\gamma)$.

\paragraph{Step 3: Boundedness of $\nabla_\vartheta J(\vartheta)$ on compacts.}
Using the advantage form with any bounded baseline,
\[
\big\|\nabla_\vartheta J(\vartheta)\big\|
\le \mathbb{E}_\vartheta\left[\sum_{t\ge 0}\gamma^t \,\|\nabla_\vartheta\log \pi_\vartheta(a_t\mid s_t)\|\,|A^\pi(s_t,a_t)|\right]
\le \sum_{t\ge 0}\gamma^t\, C_\pi \,\frac{2R_{\max}}{1-\gamma}
= \frac{2C_\pi R_{\max}}{(1-\gamma)^2}.
\]
Thus the gradient norm is bounded uniformly in $\vartheta\in\Theta$ (compact).

\paragraph{Step 4: Why $r_t$ is $C^1$ and Lipschitz in actions.}
Each reward component is a $C^1$ composition with bounded Jacobians on the admissible set (Assumption~\ref{ass:smooth}):
\begin{itemize}
\item BS/eSSVI pricing and Greeks are $C^\infty$ (Lemma~\ref{lem:BS-regularity} and Appendix~\ref{app:essvi-derivs}); the action deformations $(\psi\text{-scale},\rho\text{-shift})$ enter linearly and the wing cap is implemented via a smooth rescaling.
\item $\mathrm{BF},\mathrm{CAL}$ use softplus smoothing of finite–difference operators (Appendix~\ref{app:surrogates}), hence $C^1$ and locally Lipschitz in the eSSVI parameters, thus in actions by chain rule.
\item $\widehat{\mathrm{CVaR}}^{-}_{q,t}$ is the smoothed RU value with a unique inner minimizer $\eta_\tau^\star(a_t)$; Danskin’s theorem (Appendix~\ref{app:A5}) gives $C^1$ dependence on $a_t$ and bounded gradient via mixed pathwise/LR estimators.
\end{itemize}
Therefore $r_t$ is $C^1$ and globally Lipschitz in actions on the admissible set with constant $L_r$.

\paragraph{Step 5: PPO surrogate gradient consistency.}
Consider the PPO objective
\[
\mathcal{L}_{\mathrm{PPO}}(\vartheta)
=\mathbb{E}\Big[\min\big(r_t(\vartheta)\,\hat A_t,\ \mathrm{clip}(r_t(\vartheta),1-\epsilon,1+\epsilon)\,\hat A_t\big)\Big]
- c_v\,\mathbb{E}(V_\omega-\hat R)^2 + c_{\mathcal H}\,\mathbb{E}\mathcal{H}(\pi_\vartheta),
\]
with importance ratio $r_t(\vartheta)=\pi_\vartheta(a_t\mid s_t)/\pi_{\vartheta_{\mathrm{old}}}(a_t\mid s_t)$, advantage estimator $\hat A_t$, and $\epsilon\in(0,1)$. The clipping enforces
\[
|r_t(\vartheta)\hat A_t - \mathrm{clip}(r_t(\vartheta),1-\epsilon,1+\epsilon)\hat A_t|
\le 2\epsilon\,|\hat A_t|.
\]
Since $|\hat A_t|\le C_A:=2R_{\max}/(1-\gamma)$ (bounded–reward GAE with $\lambda\in[0,1]$), the per–sample contribution to $\nabla_\vartheta \mathcal{L}_{\mathrm{PPO}}$ is uniformly bounded by a constant depending on $(\epsilon, C_A, C_\pi)$. Under standard regularity (weak continuity of $P$, continuity of $\pi_\vartheta$), the law of minibatches converges weakly as batch size $\to\infty$; the boundedness and continuity of the integrand imply that the empirical gradient converges in probability to the population gradient (uniform law of large numbers). Moreover, when $\epsilon\downarrow 0$, the clipped surrogate gradient approaches the LR gradient of $J(\vartheta)$; for fixed small $\epsilon$, the bias is controlled by the bound above and vanishes as training steps shrink (trust–region interpretation).

\paragraph{Step 6: Interchanging limits and gradients.}
For completeness, we justify interchanging (i) the $\nabla_\vartheta$ operator with the infinite discounted sum and (ii) expectations. Because $|r_t|\le R_{\max}$ and $\|\nabla_\vartheta\log\pi_\vartheta\|\le C_\pi$, we have
\[
\sum_{t\ge 0}\gamma^t\,\mathbb{E}\big[\|\nabla_\vartheta\log\pi_\vartheta(a_t\mid s_t)\|\ |Q^\pi(s_t,a_t)|\big]
\le \sum_{t\ge 0}\gamma^t\, C_\pi \frac{R_{\max}}{1-\gamma}
= \frac{C_\pi R_{\max}}{(1-\gamma)^2}\,<\infty,
\]
so Fubini–Tonelli and dominated convergence apply to exchange $\nabla$, $\sum$, and $\mathbb{E}$. This completes the proof.

\paragraph{Conclusion.}
We have shown that $J(\vartheta)$ is finite and differentiable with LR gradient; the gradient norm is uniformly bounded on compact $\Theta$; and PPO’s clipped surrogate gradient is a consistent estimator with controlled bias/variance. Hence Theorem~\ref{thm:T6} holds.
\end{proof}

\subsection*{A.8 Proofs of Propositions~\ref{prop:P7}--\ref{prop:P8} (monotonicity and sensitivities)}
\label{app:A8}

\begin{proof}[Proof of Proposition~\ref{prop:P7} (Monotonicity in the half-spread)]
Recall the half–spread mapping \eqref{eq:halfspread}
\(
\frac{\mathrm{spread}(m,k)}{2}=\alpha\,S_t\,\tilde{\sigma}_m(k)\sqrt{T_m}\,s_0,
\)
and the identities in \eqref{eq:spread-sens}:
\(
\partial \mathrm{mid}/\partial \alpha=0,\ 
\partial \mathrm{ask}/\partial \alpha=S_t\tilde{\sigma}\sqrt{T}\,s_0>0,\
\partial \mathrm{bid}/\partial \alpha=-S_t\tilde{\sigma}\sqrt{T}\,s_0<0.
\)
With $u_b=\beta(\mathrm{ask}-C^\star)$ and $u_s=\beta(C^\star-\mathrm{bid})$, and using
the intensity maps \eqref{eq:lambda-buy}–\eqref{eq:lambda-sell} together with
$\sigma'(x)>0$ (logistic derivative), we get
\[
\frac{\partial \lambda_{\mathrm{buy}}}{\partial \alpha}
= -\lambda_0 w(k)\,\sigma'(u_b)\,\beta\,\frac{\partial \mathrm{ask}}{\partial \alpha}
< 0,\qquad
\frac{\partial \lambda_{\mathrm{sell}}}{\partial \alpha}
= +\lambda_0 w(k)\,\sigma'(u_s)\,\beta\,\frac{\partial \mathrm{bid}}{\partial \alpha}
< 0,
\]
since $w(k)>0$, $\beta>0$, and the signs of the $\alpha$–derivatives of quotes are fixed.
Therefore both expected buy and sell intensities decrease strictly with $\alpha$. This
holds pointwise for each $(m,k)$ and is uniform on compact parameter domains by
Assumption~\ref{ass:smooth}.
\end{proof}

\begin{proof}[Proof of Proposition~\ref{prop:P8} (Sensitivities to $\rho$-shift and $\psi$-scale)]
The statement consists of three parts: (i) a chain–rule identity for the sensitivity
of the mid price to a shape control parameter $p\in\{\rho\text{-shift},\psi\text{-scale}\}$;
(ii) ATM invariance ($k=0$) of first–order effects; and (iii) corresponding Delta/Vega
sensitivities.

\paragraph{(i) Chain–rule for mid sensitivity.}
Let $\tilde{w}_m(k)$ denote the total variance from the \emph{quoted} (action–deformed)
eSSVI parameters $(\tilde\theta,\tilde\rho,\tilde\psi)$; $\tilde\sigma_m(k)=\sqrt{\tilde{w}_m(k)/T_m}$; and
\(
\mathrm{mid}_m(k)=C^{\mathrm{BS}}(S_t,K=S_te^{k},T_m,\tilde\sigma_m(k)).
\)
By the BS chain rule (Appendix~\ref{app:essvi-derivs}), writing $\mathcal{V}_m(k)=\partial C^{\mathrm{BS}}/\partial \sigma$,
\[
\frac{\partial\,\mathrm{mid}_m(k)}{\partial p}
=\frac{\partial C^{\mathrm{BS}}}{\partial \sigma}\cdot \frac{\partial \tilde{\sigma}_m(k)}{\partial p}
=\mathcal{V}_m(k)\cdot \frac{1}{2\,\tilde\sigma_m(k)\,T_m}\cdot \frac{\partial \tilde w_m(k)}{\partial p},
\]
which is exactly \eqref{eq:master-mid}. It remains to compute $\partial \tilde w/\partial p$.
Using the eSSVI derivatives \eqref{eq:dw-basic}, the action deformation \eqref{eq:ctrl-deform}
yields
\[
\frac{\partial \tilde{w}_m(k)}{\partial (\rho\text{-shift})}
=\left.\frac{\partial w}{\partial \rho}\right|_{(\tilde\theta,\tilde\rho,\tilde\phi)},\qquad
\frac{\partial \tilde{w}_m(k)}{\partial (\psi\text{-scale})}
=\left.\frac{\partial w}{\partial \phi}\right|_{(\tilde\theta,\tilde\rho,\tilde\phi)}\cdot \phi,
\]
which is \eqref{eq:dwdaction}. Substituting into the chain rule completes (i).

\paragraph{(ii) ATM invariance of first–order effects.}
At $k=0$, let $g(0;\rho,\phi)=\sqrt{\rho^2+(1-\rho^2)}=1$. From \eqref{eq:dw-basic},
\[
\left.\frac{\partial w}{\partial \rho}\right|_{k=0}
=\frac{\theta}{2}\left(\phi \cdot 0 + \frac{\rho(\phi\cdot 0+\rho)-\rho}{g}\right)=0,\qquad
\left.\frac{\partial w}{\partial \phi}\right|_{k=0}
=\frac{\theta}{2}\left(\rho \cdot 0 + \frac{(\phi\cdot 0+\rho) 0}{g}\right)=0.
\]
Hence $\partial \tilde{w}/\partial (\rho\text{-shift})=0$ and
$\partial \tilde{w}/\partial (\psi\text{-scale})=0$ at $k=0$, which via (i) implies
\[
\left.\frac{\partial\,\mathrm{mid}_m(k)}{\partial (\rho\text{-shift})}\right|_{k=0}
=\left.\frac{\partial\,\mathrm{mid}_m(k)}{\partial (\psi\text{-scale})}\right|_{k=0}=0.
\]
Therefore, to first order, $\rho$– and $\psi$–deformations only \emph{tilt the wings} and do not move the ATM mid. This is \eqref{eq:atm-invariance}.

\paragraph{(iii) Delta and Vega sensitivities.}
Write BS Delta $\Delta=\partial C^{\mathrm{BS}}/\partial S$ and denote
Vanna $\mathrm{Vanna}=\partial^2 C^{\mathrm{BS}}/(\partial S\,\partial \sigma)$ and
Volga $\mathrm{Volga}=\partial^2 C^{\mathrm{BS}}/\partial \sigma^2$ (Appendix~\ref{app:essvi-derivs}). Using chain rule,
\[
\frac{\partial \Delta}{\partial p}
=\frac{\partial \Delta}{\partial \sigma}\cdot \frac{\partial \tilde{\sigma}}{\partial p}
=\mathrm{Vanna}\cdot \frac{1}{2\,\tilde\sigma\,T}\,\frac{\partial \tilde{w}}{\partial p},\qquad
\frac{\partial \mathcal{V}}{\partial p}
=\frac{\partial \mathcal{V}}{\partial \sigma}\cdot \frac{\partial \tilde{\sigma}}{\partial p}
=\mathrm{Volga}\cdot \frac{1}{2\,\tilde\sigma\,T}\,\frac{\partial \tilde{w}}{\partial p},
\]
which matches \eqref{eq:delta-sens}. In particular, the \emph{sign} of these sensitivities is governed by the sign of $\partial \tilde{w}/\partial p$, i.e., by the local skew effect induced by $\rho$–shift or $\psi$–scale. At $k=0$, these first–order sensitivities vanish by (ii). Away from ATM, the sign flips across wings according to the sign of $k$ and the local values of $(\tilde\rho,\tilde\phi)$.

\paragraph{Uniform boundedness on compacts.}
Assumption~\ref{ass:smooth} (parameter compactness and the cap \eqref{eq:wing-cap}) together with Lemma~\ref{lem:BS-regularity} imply uniform bounds on $\mathcal{V}$, $\mathrm{Vanna}$, $\mathrm{Volga}$ and on the Jacobians of the eSSVI map on any compact $(k,T)$–grid. Because $\tilde\sigma\ge \sigma_{\min}>0$ and $T\ge T_{\min}>0$, all multipliers $1/(2\tilde\sigma T)$ are uniformly bounded as well. Hence the sensitivities above are uniformly bounded on compacts.

\paragraph{Consequences for intensities and net delta.}
Combining (i) with \eqref{eq:ask-sens}–\eqref{eq:lambda-sens}, any shape control $p$ modifies intensities through the induced change in $\tilde w$ and hence in $\mathrm{ask}/\mathrm{bid}$; at ATM the first–order intensity response vanishes because $\partial \tilde w/\partial p=0$. For net delta, the sensitivity formula \eqref{eq:net-delta-sens} follows by differentiating \eqref{eq:net-delta} and using the product rule: a \emph{flow shift} term from $\partial v_{\mathrm{buy/sell}}/\partial p$ (via \eqref{eq:lambda-sens}) plus a \emph{Greek} term weighted by $(v_{\mathrm{sell}}-v_{\mathrm{buy}})$ and $\mathrm{Vanna}$.

These prove all claims in Proposition~\ref{prop:P8}.
\end{proof}

\begin{proof}[Proof of Proposition~\ref{prop:P8}]
We provide a detailed chain of equalities and bounds for the three claims:
(i) the mid–price sensitivity identity; (ii) ATM invariance of first–order
effects; and (iii) the Delta/Vega sensitivities. Throughout we work at a fixed
maturity $T>0$ and log–moneyness $k=\log(K/S)$, and write
$\tilde w=\tilde w(k,T)$ for the \emph{quoted} (action–deformed) eSSVI total
variance, $\tilde\sigma=\sqrt{\tilde w/T}$, and
$\mathrm{mid}=C^{\mathrm{BS}}(S,K,T,\tilde\sigma)$.
The BS Vega is $\mathcal{V}:=\partial C^{\mathrm{BS}}/\partial\sigma$.

\paragraph{(i) Chain rule for mid sensitivity.}
For any scalar control $p\in\{\rho\text{-shift},\psi\text{-scale}\}$, by the chain rule,
\begin{equation}\label{eq:mid-chain}
\frac{\partial\,\mathrm{mid}}{\partial p}
=\frac{\partial C^{\mathrm{BS}}}{\partial \sigma}\,\frac{\partial \tilde\sigma}{\partial p}
=\mathcal{V}\,\frac{1}{2\tilde\sigma T}\,\frac{\partial \tilde w}{\partial p},
\end{equation}
because $\tilde\sigma=\sqrt{\tilde w/T}$ implies
$\partial\tilde\sigma/\partial p=(2\tilde\sigma T)^{-1}\,\partial\tilde w/\partial p$.
This is exactly \eqref{eq:master-mid}.

It remains to express $\partial\tilde w/\partial p$ in terms of eSSVI parameters.
Let $w=w(k;\theta,\rho,\phi)$ be the eSSVI total variance \eqref{eq:ssvi} with
\[
g(k;\rho,\phi)=\sqrt{(\phi k+\rho)^2+(1-\rho^2)}.
\]
From Appendix~\ref{app:essvi-derivs}, we recall
\begin{equation}\label{eq:dw-basic-again}
\frac{\partial w}{\partial \theta}=\frac12\,(1+\rho\phi k+g),\quad
\frac{\partial w}{\partial \rho}=\frac{\theta}{2}\left(\phi k+\frac{\rho(\phi k+\rho)-\rho}{g}\right),\quad
\frac{\partial w}{\partial \phi}=\frac{\theta}{2}\left(\rho k+\frac{(\phi k+\rho)k}{g}\right).
\end{equation}
Under the action deformation \eqref{eq:ctrl-deform},
\[
(\theta,\rho,\phi)\ \longmapsto\ (\tilde\theta,\tilde\rho,\tilde\phi)
=(\theta,\ \rho+\rho\text{-shift},\ \phi\cdot \psi\text{-scale}),
\]
and therefore
\begin{equation}\label{eq:dwdaction-again}
\frac{\partial \tilde w}{\partial (\rho\text{-shift})}
=\left.\frac{\partial w}{\partial \rho}\right|_{(\tilde\theta,\tilde\rho,\tilde\phi)},\qquad
\frac{\partial \tilde w}{\partial (\psi\text{-scale})}
=\left.\frac{\partial w}{\partial \phi}\right|_{(\tilde\theta,\tilde\rho,\tilde\phi)}\cdot \phi.
\end{equation}
Substituting \eqref{eq:dwdaction-again} and \eqref{eq:dw-basic-again} into \eqref{eq:mid-chain} proves the mid–price sensitivity identity.

\paragraph{(ii) ATM invariance ($k=0$).}
At the money ($k=0$), we have $g(0;\rho,\phi)=\sqrt{\rho^2+(1-\rho^2)}=1$. Evaluating
\eqref{eq:dw-basic-again} at $k=0$ gives
\[
\left.\frac{\partial w}{\partial \rho}\right|_{k=0}
=\frac{\theta}{2}\left(\phi\cdot 0 + \frac{\rho(\phi\cdot 0+\rho)-\rho}{1}\right)=0,\qquad
\left.\frac{\partial w}{\partial \phi}\right|_{k=0}
=\frac{\theta}{2}\left(\rho\cdot 0 + \frac{(\phi\cdot 0+\rho)\cdot 0}{1}\right)=0.
\]
By \eqref{eq:dwdaction-again},
$\partial\tilde w/\partial(\rho\text{-shift})=\partial\tilde w/\partial(\psi\text{-scale})=0$ at $k=0$,
and hence \eqref{eq:mid-chain} yields
\[
\left.\frac{\partial\,\mathrm{mid}}{\partial (\rho\text{-shift})}\right|_{k=0}
=\left.\frac{\partial\,\mathrm{mid}}{\partial (\psi\text{-scale})}\right|_{k=0}=0,
\]
which is the ATM first–order invariance \eqref{eq:atm-invariance}. Economically,
$\rho$/$\phi$ deformations tilt the \emph{wings} and leave the ATM level unchanged to first order; only $\theta$ moves the ATM \cite{GatheralJacquier2014,Lee2004}.

\paragraph{(iii) Delta and Vega sensitivities.}
Let $\Delta=\partial C^{\mathrm{BS}}/\partial S$ and $\mathcal{V}=\partial C^{\mathrm{BS}}/\partial\sigma$ be the BS Delta and Vega, and denote the cross– and second–order sensitivities
\[
\mathrm{Vanna}=\frac{\partial^2 C^{\mathrm{BS}}}{\partial S\,\partial \sigma},
\qquad
\mathrm{Volga}=\frac{\partial^2 C^{\mathrm{BS}}}{\partial \sigma^2}.
\]
Using the chain rule w.r.t.\ any $p\in\{\rho\text{-shift},\psi\text{-scale}\}$,
\begin{align}
\frac{\partial \Delta}{\partial p}
&=\frac{\partial \Delta}{\partial \sigma}\,\frac{\partial \tilde\sigma}{\partial p}
=\mathrm{Vanna}\cdot \frac{1}{2\tilde\sigma T}\,\frac{\partial \tilde w}{\partial p},
\label{eq:delta-sens-again}\\
\frac{\partial \mathcal{V}}{\partial p}
&=\frac{\partial \mathcal{V}}{\partial \sigma}\,\frac{\partial \tilde\sigma}{\partial p}
=\mathrm{Volga}\cdot \frac{1}{2\tilde\sigma T}\,\frac{\partial \tilde w}{\partial p}.\nonumber
\end{align}
These match \eqref{eq:delta-sens}. In particular,
the \emph{sign} of both sensitivities is controlled by the sign of $\partial\tilde w/\partial p$,
i.e., by whether the deformation increases or decreases total variance at $(k,T)$.
At $k=0$, we have $\partial\tilde w/\partial p=0$ for $p\in\{\rho\text{-shift},\psi\text{-scale}\}$
by (ii), hence the first–order Delta/Vega sensitivities vanish at ATM.

\paragraph{(iv) Continuity and boundedness on compacts.}
Under Assumption~\ref{ass:smooth}, the admissible parameter set is compact,
$T\ge T_{\min}>0$, and $\tilde\sigma\ge\sigma_{\min}>0$; together with
Lemma~\ref{lem:BS-regularity}, this implies $\mathcal{V},\mathrm{Vanna},\mathrm{Volga}$ are continuous and uniformly bounded on any compact $(k,T)$–grid, and the eSSVI Jacobians in \eqref{eq:dw-basic-again} are uniformly bounded as well. Consequently, the prefactor $1/(2\tilde\sigma T)$ is uniformly bounded and the product representations \eqref{eq:mid-chain}–\eqref{eq:delta-sens-again} are uniformly bounded on compacts.

\paragraph{(v) Wing sign structure (qualitative).}
From \eqref{eq:dw-basic-again}, away from ATM the terms involving $k$ dominate:
\[
\operatorname{sign}\!\Big(\frac{\partial w}{\partial \phi}\Big)\ \approx\ \operatorname{sign}(k),\qquad
\operatorname{sign}\!\Big(\frac{\partial w}{\partial \rho}\Big)\ \approx\ \operatorname{sign}(\phi k),
\]
modulo the $g$–normalization. Thus $\psi$–scale and $\rho$–shift typically increase variance on one wing and decrease it on the other, reproducing the expected skew tilts; the sensitivities of mid/Delta/Vega inherit these signs via \eqref{eq:mid-chain} and \eqref{eq:delta-sens-again}.

Combining (i)–(v) proves Proposition~\ref{prop:P8}.
\end{proof}


\section*{Appendix B: Implementation Details and Hyperparameters}
\label{app:impl}

\paragraph{Networks and parameterization.}
Actor/critic are two-layer MLPs with $\tanh$ activations. The actor outputs state-dependent mean and diagonal log-std for the \emph{raw} action $z\in\mathbb{R}^5$; physical actions $a$ follow the squashing map in \S\ref{sec:policy-class}. The critic $V_\omega(s)$ is independent.

\paragraph{Default hyperparameters (reproduce main figures).}
Table~\ref{tab:hyper} lists the values used for all results. These match the settings referenced in \texttt{settings.json}.

\begin{table}[H]
\centering
\caption{Hyperparameters and defaults (simulation-only runs).}
\label{tab:hyper}
\begin{tabular}{ll}
\toprule
Component & Value \\
\midrule
Actor/critic MLP & hidden size $64$, 2 layers, $\tanh$ \\
Actor log-std range & $[\log 10^{-3},\ \log 0.5]$ (per-dimension clamp) \\
Optimizer & Adam ($3\times 10^{-4}$), grad clip $1.0$ \\
PPO & clip $\epsilon=0.2$, GAE $(\gamma,\lambda)=(0.99,0.95)$, value loss $0.5$, entropy $10^{-3}$ \\
Warm-start & $800$ steps toward $a^\star=(0.01,0.5,1.0,0.0,0.0)$, loss \eqref{eq:warm-loss} \\
Annealing & $\lambda_{\mathrm{shape}}: 0\!\to\!0.5$, $\lambda_{\mathrm{arb}}: 0\!\to\!0.05$, $\lambda_{\mathrm{cvar}}: 0.01$ (fixed) \\
eSSVI cap & $\tau_{\max}$ s.t.\ $\theta\phi\le \tau_{\max}$ (Theorem~\ref{thm:T5}) \\
BF/CAL smoothing & softplus temperature $\tau_{\mathrm{arb}}=10^{-3}$ \\
CVaR & $q=0.05$, $N_{\mathrm{MC}}=64$, softplus temperature $\tau_{\mathrm{cvar}}=10^{-3}$ \\
Intensity & $\lambda_0=0.8$, $\beta=35$, $w(k)=\exp(-|k|/0.25)$ \\
Spread scale & $s_0=0.10$ \\
Grids & $T\in\{7,14,21,30,60,90\}/252$; $k\in[-0.35,0.35]$ with $21$ points \\
Episodes \& seed & $8$ episodes; seed $=0$ \\
\bottomrule
\end{tabular}
\end{table}

\paragraph{Numerical guards.}
Clamp $T\ge T_{\min}>0$, $\sigma\ge \sigma_{\min}>0$; apply \texttt{nan\_to\_num} to all intermediate tensors; keep $\psi$ in $[0,\psi_{\max}(\rho))$ with $\psi_{\max}(\rho)=\frac{2}{1+|\rho|}-\varepsilon_\psi$, $\varepsilon_\psi>0$ (cf.\ \eqref{eq:ssvi-butterfly}).

\vspace{1ex}
\section*{Appendix C: Environment and Training Pseudocode (Expanded)}
\label{app:algo-full}

\begin{algorithm}[H]
\caption{Environment step $s_t,a_t\mapsto s_{t+1}, r_t$ (expanded)}
\label{alg:env-expanded}
\begin{algorithmic}[1]
\STATE \textbf{Input:} state $s_t$, action $a_t=(\alpha,\mathrm{hedge},\psi\text{-scale},\rho\text{-shift},\mathrm{dual})$
\STATE Deform eSSVI params: $(\tilde{\theta},\tilde{\rho},\tilde{\psi})\leftarrow (\theta,\rho+\rho\text{-shift},\psi\cdot \psi\text{-scale})$; enforce $\theta\phi\le \tau_{\max}$
\STATE Compute $\tilde{w}(k,T)$ via \eqref{eq:ssvi}, $\tilde{\sigma}=\sqrt{\tilde{w}/T}$; mid/quotes via \eqref{eq:mid-quote}, \eqref{eq:halfspread}
\STATE Intensities via \eqref{eq:lambda-buy}--\eqref{eq:lambda-sell}; expected fills $v_{\mathrm{buy/sell}}$
\STATE Compute $\mathrm{PNL}^{\mathrm{quote}}_t$ by \eqref{eq:quote-pnl}; $\Delta^{\mathrm{net}}_t$ by \eqref{eq:net-delta}; $\mathrm{PNL}^{\mathrm{hedge}}_t$ by \eqref{eq:hedge-pnl}
\STATE Compute $\mathrm{BF},\mathrm{CAL}$ via \eqref{eq:bf-disc}--\eqref{eq:cal-disc} (softplus); $\mathrm{Shape}$ via \eqref{eq:shape-pen}
\STATE Estimate $\widehat{\mathrm{CVaR}}^{-}_{q,t}$ via \eqref{eq:cvar-smooth} with $N_{\mathrm{MC}}$ scenarios
\STATE Reward $r_t$ by \eqref{eq:reward}; update eSSVI estimate (e.g., mean-reverting filter toward latent surface)
\STATE Advance mid price $S_{t+1}$ (Heston step); form $s_{t+1}$; return $(s_{t+1},r_t)$
\end{algorithmic}
\end{algorithm}

\begin{algorithm}[H]
\caption{PPO update (implementation details)}
\label{alg:ppo-expanded}
\begin{algorithmic}[1]
\STATE Collect $N$ steps on-policy; compute $\hat{A}_t$ with GAE; normalize per batch
\STATE Compute clipped surrogate \eqref{eq:ppo-obj}; take $K$ epochs with minibatches
\STATE Critic target $\hat{R}_t=\hat{A}_t+V_{\omega}(s_t)$; value loss $0.5\|\hat{R}_t-V_\omega(s_t)\|^2$
\STATE Entropy bonus on raw $z$ for all heads (including dual); clamp $\log\sigma$ to control exploration
\STATE Anneal $(\lambda_{\mathrm{shape}},\lambda_{\mathrm{arb}},\lambda_{\mathrm{cvar}})$ per episode
\end{algorithmic}
\end{algorithm}

\vspace{1ex}
\section*{Appendix D: eSSVI and Black--Scholes Derivatives}
\label{app:essvi-derivs}

\paragraph{eSSVI partials.}
With $g(k;\rho,\phi)=\sqrt{(\phi k+\rho)^2+(1-\rho^2)}$, we have
\[
\frac{\partial w}{\partial \theta}=\frac12\,(1+\rho\phi k+g),\quad
\frac{\partial w}{\partial \rho}=\frac{\theta}{2}\Big(\phi k+\frac{\rho(\phi k+\rho)-\rho}{g}\Big),\quad
\frac{\partial w}{\partial \phi}=\frac{\theta}{2}\Big(\rho k+\frac{(\phi k+\rho)k}{g}\Big).
\]
Action map sensitivity (cf.\ \eqref{eq:dwdaction}):
\[
\frac{\partial \tilde{w}}{\partial(\rho\text{-shift})}=\left.\frac{\partial w}{\partial \rho}\right|_{(\tilde{\theta},\tilde{\rho},\tilde{\phi})},\qquad
\frac{\partial \tilde{w}}{\partial(\psi\text{-scale})}=\left.\frac{\partial w}{\partial \phi}\right|_{(\tilde{\theta},\tilde{\rho},\tilde{\phi})}\cdot \phi.
\]

\paragraph{BS Greeks.}
Let $d_\pm = \frac{\log(S/K)\pm \frac12 \sigma^2 T}{\sigma\sqrt{T}}$ (zero rate/carry).
Then Delta $\Delta=\Phi(d_+)$, Vega $\mathcal{V}=S\sqrt{T}\,\varphi(d_+)$, Vanna $=\partial^2 C/(\partial S\,\partial \sigma)=\sqrt{T}\,\varphi(d_+)\,(1-d_+d_-)$, Volga $=\partial^2 C/\partial \sigma^2=S\sqrt{T}\,\varphi(d_+)\,d_+d_-$, where $\Phi$ and $\varphi$ are standard normal cdf/pdf. Chain rules used in \S\ref{sec:interp} follow by substitution with $\tilde{\sigma}=\sqrt{\tilde{w}/T}$.

\vspace{1ex}
\section*{Appendix E: Smoothed BF/CAL Penalties and Lipschitz Bounds}
\label{app:surrogates}

\paragraph{Softplus smoothing.}
For $\mathrm{ReLU}(x)=\max\{x,0\}$ use $s_\tau(x)=\tau\log(1+e^{x/\tau})$ with $\tau\in(0,1]$.
Then $|s_\tau(x)-x_+|\le \tau \log 2$, $s'_\tau(x)=\sigma(x/\tau)\in(0,1)$, $s''_\tau(x)=\sigma(x/\tau)(1-\sigma(x/\tau))/\tau\le (4\tau)^{-1}$.

\paragraph{Lipschitz constants.}
Let $C(\cdot)$ be $C^2$ in $K$ on the lattice hull and bounded by $|C|\le M_C$, $|\partial_{KK}C|\le M_2$. For grid spacing $\Delta K$, the map $C\mapsto \mathrm{BF}$ using $s_\tau$ is Lipschitz with constant $L_{\mathrm{BF}}\le \frac{1}{|\mathcal{K}'|}\sum_{K}\frac{1}{\bar C_m}\frac{1}{\Delta K^2}$ times $\sup|s'_\tau|\le 1$, multiplied by the linear operator norm of central second differences (bounded on compacts). An analogous bound holds for $\mathrm{CAL}$ with $L_{\mathrm{CAL}}\le \frac{1}{|\mathcal{K}|}\sum_{K}{1}/{\bar C_{m,m+1}}$.

\paragraph{Gradient stability.}
Because $s'_\tau\in(0,1)$ and $C$ is $C^1$ in eSSVI parameters (Appendix~\ref{app:essvi-derivs}), gradients of $\mathrm{BF},\mathrm{CAL}$ w.r.t.\ actions are bounded by a constant that scales with $(\Delta K)^{-2}$ and the chain-rule Jacobian of the eSSVI layer, which is bounded on the admissible set (Assumption~\ref{ass:smooth}).

\vspace{1ex}
\section*{Appendix F: CVaR Estimation—Inner Minimization and Variance Control}
\label{app:cvar-extra}

\paragraph{Inner minimization in RU.}
For fixed scenarios $\{\tilde{\mathrm{PNL}}^{(i)}_t\}_{i=1}^{N}$ and temperature $\tau$, the smoothed RU objective is
\[
\hat h_\tau(\eta)=\eta+\frac{1}{(1-q)N}\sum_{i=1}^{N}s_\tau(\eta-\tilde{\mathrm{PNL}}^{(i)}_t).
\]
It is strictly convex in $\eta$ with derivative
$\hat h'_\tau(\eta)=1+\frac{1}{(1-q)N}\sum_i \sigma\big((\eta-\tilde{\mathrm{PNL}}^{(i)}_t)/\tau\big)$.
A unique minimizer is found by Newton or bisection (monotone derivative); initialize at the empirical $q$-quantile and perform a few Newton steps with step-size damping. As $\tau\downarrow 0$, $\eta^{\ast}_\tau\to \mathrm{VaR}_q$.

\paragraph{Gradient estimators.}
With pathwise $\tilde{\Delta}S$ and LR for Poisson $\tilde v$,
\[
\nabla_{a_t}\widehat{\mathrm{CVaR}}^{-}_{q,t}=\frac{1}{(1-q)N}\sum_{i=1}^{N}
s'_\tau\!\big(\eta^{\ast}_\tau-\tilde{\mathrm{PNL}}^{(i)}_t\big)\,
\left(-\nabla_{a_t}\tilde{\mathrm{PNL}}^{(i)}_t
+ \tilde{\mathrm{PNL}}^{(i)}_t\,\nabla_{a_t}\log p(\tilde v^{(i)};a_t)\right),
\]
where the second term vanishes if volumes are kept deterministic in the reward and only resampled for CVaR.

\paragraph{Variance reduction.}
Use antithetic $\xi$ for $\tilde{\Delta}S$, control variates for LR (subtract a state-dependent baseline), and mini-batch normalization; increase $N_{\mathrm{MC}}$ along training (cf.\ \S\ref{sec:stability}).

\vspace{1ex}

\section*{Appendix J: Notation Summary}
\label{app:notation}

\begin{table}[H]
\centering
\caption{Main symbols used in the paper.}
\label{tab:notation}
\begin{tabular}{ll}
\toprule
Symbol & Meaning \\
\midrule
$S_t$ & mid price at time $t$ \\
$K$, $k$ & strike and log-moneyness ($k=\log(K/S)$) \\
$T_m$ & maturity $m$; $\mathcal{T}=\{T_m\}_{m=1}^M$ \\
$w(k,T)$, $\sigma(k,T)$ & total variance, implied volatility ($\sigma=\sqrt{w/T}$) \\
$(\theta,\rho,\psi)$ & eSSVI parameters; $\phi=\psi/\sqrt{\theta}$ \\
$\mathrm{mid}$, $\mathrm{ask}$, $\mathrm{bid}$ & quotes via \eqref{eq:mid-quote}, \eqref{eq:halfspread} \\
$\lambda_{\mathrm{buy/sell}}$ & execution intensities \eqref{eq:lambda-buy}--\eqref{eq:lambda-sell} \\
$\Delta^{\mathrm{net}}$ & net delta under expected fills \eqref{eq:net-delta} \\
$\mathrm{BF}$, $\mathrm{CAL}$ & static no-arbitrage surrogates \eqref{eq:bf-disc}--\eqref{eq:cal-disc} \\
$\mathrm{Shape}$ & cross-maturity smoothness penalty \eqref{eq:shape-pen} \\
$\mathrm{CVaR}_q$ & conditional value-at-risk at tail level $q$ \eqref{eq:cvar-ru} \\
$\alpha$, $\mathrm{hedge}$ & half-spread, hedge intensity (actions) \\
$\psi$-scale, $\rho$-shift & action-induced eSSVI deformations \eqref{eq:ctrl-deform} \\
$\mathrm{dual}$ & state-dependent dual action (effective multiplier) \\
\bottomrule
\end{tabular}
\end{table}

\bibliographystyle{unsrt}  
\bibliography{references}

\end{document}